\documentclass[pre,superscriptaddress, twocolumn,nofootinbib,amsmath,amssymb,aps,longbibliography]{revtex4-1}
\usepackage{graphicx}
\usepackage{dcolumn}
\usepackage{bm}
\usepackage{braket}
\usepackage{amsmath}
\usepackage{bbold}
\usepackage{caption}
\usepackage{subcaption}
\usepackage{natbib}
\usepackage{xcolor}
\usepackage{amsthm}
\usepackage{physics}
\usepackage{enumitem}
\usepackage{comment}
\usepackage{ulem} 

\newcolumntype{N}{@{}m{0pt}@{}}

\newtheorem{theorem}{Theorem}
\newtheorem{proposition}{Proposition}
\newtheorem{corollary}{Corollary}[theorem]
\newtheorem{lemma}{Lemma}
\theoremstyle{definition}

\theoremstyle{remark}
\newtheorem*{remark}{Remark}

\renewcommand{\arraystretch}{1.7}

\usepackage{hyperref}
\usepackage[bottom]{footmisc}
\usepackage{scrextend}
\deffootnote[1em]{1.5em}{1em}{\textsuperscript{\thefootnotemark}}

\def\blk{\color{black}}

\begin{document}



\title{Corrections to the Hamiltonian Induced by  Finite-Strength Coupling to the Environment}

\author{Marcin {\L}obejko}
\affiliation{Institute of Theoretical Physics and Astrophysics, Faculty of Mathematics,
Physics and Informatics, University of Gda\'nsk, 80-308 Gda\'nsk, Poland}
\affiliation{International Centre for Theory of Quantum Technologies, University of Gda\'nsk, 80-308 Gda\'nsk, Poland}
\author{Marek Winczewski}
\affiliation{International Centre for Theory of Quantum Technologies, University of Gda\'nsk, 80-308 Gda\'nsk, Poland}
\author{Gerardo Suárez}
\affiliation{International Centre for Theory of Quantum Technologies, University of Gda\'nsk, 80-308 Gda\'nsk, Poland}
\author{Robert Alicki}
\affiliation{International Centre for Theory of Quantum Technologies, University of Gda\'nsk, 80-308 Gda\'nsk, Poland}
\author{Micha{\l} Horodecki}
\affiliation{International Centre for Theory of Quantum Technologies, University of Gda\'nsk, 80-308 Gda\'nsk, Poland}

\date{\today}
\begin{abstract}

If a quantum system interacts with the environment, the Hamiltonian acquires a correction known as the Lamb-shift term. There are two other corrections to the Hamiltonian, related to the stationary state. Namely, the stationary state is to first approximation a Gibbs state with respect to original Hamiltonian. However, if we have finite coupling, the true stationary state will be different, and regarding it as a Gibbs state to some effective Hamiltonian, one can extract a correction, which is called ``steady-state" correction. 
Alternatively, one can take a static point of view, and consider the reduced state of total equilibrium state, i.e., system plus bath Gibbs state. The extracted Hamiltonian correction is called the ``mean-force" correction. This paper presents several analytical results on second-order corrections (in coupling strength) of the three types mentioned above. Instead of the steady state, we focus on a state annihilated by the Liouvillian of the master equation, labeling it as the ``quasi-steady state". Specifically, we derive a general formula for the mean-force correction as well as the quasi-steady state and Lamb-shift correction for a general class of master equations. Furthermore, specific formulas for corrections are obtained for the Davies, Bloch-Redfield, and cumulant equation (refined weak coupling). In particular, the cumulant equation serves as a case study of the Liouvillian, featuring a non-trivial fourth-order generator. This generator forms the basis for calculating the diagonal quasi-steady state correction. We consider spin-boson model as an example, and in addition to using our formulas for corrections, we consider mean-force correction from 
reaction-coordinate approach. 
\end{abstract}

\maketitle


\section{Introduction} 

As is well known, when a system interacts with its environment, it undergoes not only dissipative but also experiences the renormalization of Hamiltonian dynamics, leading to the so-called 'Lamb-shift' corrected Hamiltonian \cite{AlickiLendi1987,Breuer+2006}. Furthermore, for finite coupling strength with a single heat bath, the equilibrium state of the open system's  dynamics will (assuming ergodicity~\cite{AlickiLendi1987,Breuer+2006}) deviate from the Gibbs state of the non-interacting system Hamiltonian. Instead, it is widely believed, and in many cases confirmed, that the appropriate candidate for the steady state is the reduced state of the total equilibrium, encompassing both the system and the environment (see \cite{Trushechkin_2021_MeanForce} and references therein). Generally, this can be regarded as a Gibbs state with respect to some effective  Hamiltonian, referred to as the 'mean-force' Hamiltonian. The  state is called the 'mean-force Gibbs' state. Similarly, one can deduce  a Hamiltonian from the true steady state (referred to as the 'steady-state Hamiltonian'), with the expectation that the two Hamiltonians - i.e. the mean-force and the steady state Hamiltonian  - are equal (see, in this context, \cite{Roux1999,fleming,thingna2012,anders-oscillator,Guarnieri,Trushechkin_2021_MeanForce,miyashita,Subasi}).

If the coupling with the environment is weak, albeit finite, all the three Hamiltonians mentioned above\footnote{i.e., Lamb-shift corrected, mean-force Gibbs state Hamiltonian, and the one deduced from the steady state.} take the form of the original (i.e., bare) Hamiltonians plus a 
correction term. Investigating these corrections is currently the subject of intense research~\cite{Trushechkin_2021_MeanForce,miyashita,tupkary2022fundamental}. The primary challenge in analyzing and comparing these corrections lies in the necessity of employing approximations in description of an open quantum system, resulting in more or less accurate master equations (ME) \cite{AlickiLendi1987,Breuer+2006}. In this context, our goal is to ensure that these master equations yield a steady-state Hamiltonian as close as possible to the mean-force Hamiltonian, particularly up to the second order in coupling strength.

The objective of this paper is to present analytical results regarding the three types of corrections to the bare Hamiltonian, up to second-order (in coupling strength)  for a specific class of open system models (in a previous work mostly the corrections to states have been provided, see e.g. \cite{Subasi,Geva2000,miyashita,thingna2012,Guarnieri,anders-oscillator,latune,Latune2022}).
As previously mentioned, we are dealing with three types of corrections: the Lamb-shift, mean-force, and steady-state correction. However, obtaining the steady-state correction is particularly challenging, so we opt for an alternative approach. Specifically, we analyze the Hamiltonian not derived from the steady state, but from a state annihilated by the generator of the dynamics. Since this state is likely to be equal to the true steady state (for instance, this holds true for a time-independent Liouvillian), we refer to it as the ``quasi-steady" state, along with its corresponding Hamiltonian correction.

An important motivation for considering Hamiltonian corrections instead states corrections is that the Hamiltonian corrections might be useful for those looking at effective Hamiltonian theories \cite{effective_qed,hep_lieb}. Our approach might help extend these effective Hamiltonians into the continuous frequency regime, also it might prove useful in the context of the fluctuation dissipation theorem where the mean-force Hamiltonian is sometimes computed \cite{Campisi_2009,fluctuation_mf}

Here are our most general results that do not involve any particular master equation. We derive:
\begin{itemize}
\item The \textit{general form of the  second order mean-force correction} for an arbitrary open system, a result previously known only for specific couplings. This form has also been independently derived by G. Timofeev and A. Trushechkin \cite{trushechkin2022}.
\item The \textit{off-diagonal} elements of the quasi-steady state correction for a relatively broad class of master equations expressed in terms of their Kossakowski matrix.
\end{itemize}
Additionally, our formulas for the above corrections explicitly reveal their relations with the Lamb-shift correction.

Next, we focus on three major descriptions of open systems: Davies ME \cite{davies1974markovian, gorini1976completely}, Bloch-Redfield ME \cite{bloch1957generalized,redfield1957theory}, and cumulant equation (also known as refined weak coupling) \cite{Alicki89,Rivas_2017,Rivas_2019}. It is important to note that, unlike Bloch-Redfield, both Davies and cumulant equations provide completely positive evolution. We demonstrate that for a general coupling, the Bloch-Redfield and cumulant equations predict the \textit{off-diagonal} elements of a correction to the \textit{quasi-steady state} Hamiltonian that coincide with those of the mean-force (previously known only for specific models, as seen in \cite{Guarnieri,thingna2012,miyashita,Subasi}). In contrast, the off-diagonal elements of quasi-steady state correction for the Davies equation aligns with the 
non-standard \textit{Lamb-shift} correction, namely the one, for which  secular approximation is not done (indeed, to obtain completely positive dynamics, is is enough to apply secular approximation just to dissipative part).  

Subsequently, we address the more intricate matter of diagonal elements, presenting a methodology for computing the diagonal elements of the quasi-steady state correction. Consequently, we derive analytical formulas for these elements in the case of a two-level system. Our findings indicate that for the discussed master equations with Liouvillians defined up to the second order (e.g., Davies and Bloch-Redfield master equations), there is no correction to the diagonal elements. 

However, such corrections do appear if we consider cumulant equation. Namely, we write it in the form of a master equation, and truncate the Liouvillian up to fourth order. 
The obtained correction 
exhibits a notable agreement (verified numerically, see below) with the corresponding mean-force correction. It is worth noting that still a discrepancy is here expected, given that the cumulant equation neglects contributions from higher-order cumulants.  

Lastly, we computed the derived corrections for the spin-boson model. As we mentioned before, this provides us with a direct verification of mean-force and quasi-steady state predictions. Moreover, we additionally involved the reaction coordinate method \cite{latune,Nazir2018,lambert} to verify numerically the analytical expression for the mean-force correction, which proves the perfect agreement in the regime of validity of the method.    

A few remarks are here in order. As usual, the obtained corrections will be cut-off dependent and often diverge with the growing cut-off frequency. This is actually ubiquitous in the literature on the topic  (see in this context  \cite{WinczewskiMHA} and \cite{WinczewskiAlicki}). Secondly, we do not touch on the issue of renormalization: the derivation of the master equation  should be not be based on bare Hamiltonian but should somehow involve the renormalized one (as advocated in \cite{Alicki89} and \cite{WinczewskiAlicki}). We have not followed this in the present manuscript to keep  clear the main message.

\section{Hamiltonian corrections} \label{sect:hamiltonian_correction}
We consider a general Hamiltonian of the interacting system $S$ with the thermal reservoir $R$ of the form:
\begin{eqnarray} \label{total_hamiltonian}
  H &=& H_0 + H_R + \lambda H_I, \\
  H_I &=& \sum_\alpha A_\alpha \otimes R_\alpha, \label{interaction_sum}
\end{eqnarray}
where $H_0$ is a bare Hamiltonian of the system, $H_R$ is free Hamiltonian of the bath,  $A_\alpha, R_\alpha$ are interaction operators (acting on the system and bath Hilbert spaces, respectively), and $\lambda$ is a coupling constant. 
In the following, we define a Gibbs state of the thermal reservoir $\gamma_R = \mathcal{Z}_R^{-1} e^{-\beta H_R}$ at inverse temperature $\beta$, where $\mathcal{Z}_R = \Tr[e^{-\beta H_R}]$. Additionally, we consider the   operators evolving in the interaction picture $A(t) = e^{i (H_0 + H_R) t} A e^{-i (H_0 + H_R) t}$, and we use an abbreviation $\langle A \rangle_{\gamma_R} = \Tr[A \gamma_R]$. We assume that that bath operators are centralized, i.e.,  $\langle R_\alpha \rangle_{\gamma_R} = 0$.

The main object of interest of this article are three different  second-order corrections to the bare Hamiltonian of the system $H_0$ in the weak coupling limit (i.e., $\lambda \ll 1$). Namely, the corrections related to the following three Hamiltonians: {\it  Lamb-shift} $H_{\text{LS}}$, {\it quasi-steady state } $H_{\text{st}}$ and the {\it mean-force} $H_{\text{mf}}$. 
Due to centralization of the bath operators, the leading order of the perturbation calculus is $\lambda^2$. In accordance, the corrections are given by the relations
\begin{eqnarray}
H_\mathrm{cor} - H_0 = \lambda^2 H_\mathrm{cor}^{(2)} + \dots,
\end{eqnarray}
which we represent by means of jump operators:
\begin{eqnarray} \label{hamiltonian_correction}
H_\mathrm{cor}^{(2)}(t) = \sum_{\omega,\omega'} \sum_{\alpha, \beta} \Upsilon^{(\mathrm{cor})}_{\alpha \beta}(\omega,\omega',t) A^\dag_{\alpha}(\omega) A_\beta(\omega'),
\end{eqnarray}
where $\mathrm{cor}$ indicates Lamb-shift  ($LS$),  quasi-steady state  ($st$) or mean-force ($mf$) correction, respectively, and the jump operators are given by:
\begin{eqnarray} \label{jump_operators_definition}
A_\alpha(\omega) = \sum_{\epsilon' - \epsilon = \omega} \Pi(\epsilon) A_\alpha \Pi(\epsilon').
\end{eqnarray}
where $\Pi(\epsilon)$ is the projector on the subspace with energy $\epsilon$, such that $H_0 = \sum_\epsilon \epsilon \ \Pi(\epsilon)$.
\begin{remark}
We need to justify that in \eqref{hamiltonian_correction} only pair of jump operators are enough.
For instance, the  operators 
of the form $A^\dag_{\alpha}(\omega) A_\beta(\omega')$
may not span the whole space of the system.
Of course this need not lead to violation of \eqref{hamiltonian_correction}, however at least it means that justification is needed. As we will see in Appendix \ref{mean_force_corrections_appendix}, Eqs. \eqref{second_order_equation_mean_force}-\eqref{mean_force_RHS} that for the mean-force corretion
this is true by definition.
In the case of the Lamb shift correction it is also true for all the models of open systems that we consider (Redfield, GKLS-Davies and the cumulant equations). 
Finally, for the quasi-steady state correction it is not clear whether the ansatz of \eqref{hamiltonian_correction}
is general enough. It is for sure correct, if we assume that $H_0$ is Bohr non-degenerate, which implies that operators 
$A^\dag_{\alpha}(\omega) A_\beta(\omega')$
indeed span the full space of system operators. 
\end{remark}

\subsection{ Lamb-shift correction}
We start with the {\it Lamb-shift correction} $H^{(2)}_{\text{LS}}$, which is defined according to the Liouvillian of the master equation in the Schr\"{o}dinger picture of the following general form:
\begin{multline} \label{master_equation}
  \mathcal{L}_t[\rho]  = i [\rho, H_0 + \lambda^2 H^{(2)}_{\text{LS}} (t)] \\ 
  + \lambda^2 \sum_{\omega,\omega'}\sum_{\alpha \beta} K_{\alpha \beta}(\omega,\omega',t) \mathcal{D}_{\alpha \beta}(\omega,\omega')[\rho] + O(\lambda^4),
\end{multline}
where:
\begin{multline}
\mathcal{D}_{\alpha \beta}(\omega,\omega')[\rho]  = A_\beta(\omega') \rho A_\alpha^\dag(\omega) - \frac{1}{2} \{A_\alpha^\dag(\omega) A_\beta(\omega'), \rho \}.
\end{multline} 
$K$ is the so-called Kossakowski matrix and for a while we do not determine it: For particular choice of $K$ and $\Upsilon^{(\text{LS})}$ we will obtain a given master equation, such as the Bloch-Redfield or Davies one. Notice that the Liouvillian in Eq. \eqref{master_equation} leads to a completely positive dynamics if the matrix $K_{\alpha \beta}(\omega,\omega', t)$ is positive semi-definite. 
Later, we consider the long-time limits (assuming they exist), when $t \to \infty$, for which we use the abbreviations:
\begin{align}
&\Upsilon^{(\text{LS})}_{\alpha \beta}(\omega,\omega') \equiv \lim_{t \to \infty} \Upsilon^{(\text{LS})}_{\alpha \beta}(\omega,\omega',t), \\
& K_{\alpha \beta}(\omega,\omega') \equiv \lim_{t \to \infty} K_{\alpha \beta}(\omega,\omega',t).
\end{align}

The Lamb-shift correction introduces a renormalization of the bare Hamiltonian $H_0$ of the open system due to the finite-strength coupling with the environment. The correction affects the Hamiltonian part of the evolution.
\subsection{Mean-force correction}\label{mean_force_correction_section}
Let us then introduce the mean-force Hamiltonian $H_\text{mf}$, defined according to the marginal Gibbs state of the global equilibrium, i.e.,
\begin{eqnarray} \label{mean_force_equation_main}
\rho_{\text{mf}} = \frac{e^{-\beta H_\text{mf}}}{\Tr_S[e^{-\beta H_\text{mf}}]} = \frac{\Tr_R[e^{-\beta H}]}{\Tr_{SR}[e^{-\beta H}]}.
\end{eqnarray}
The mean-force Gibbs state refers to the local equilibrium of the open system corresponding to the global equilibrium of the full system (i.e., open system plus the environment). The definition solely relies on static equilibrium and hence is not related to the dynamics.  

Concerning the mean-force Hamiltonian $H_\text{mf}$, one should notice that Eq. \eqref{mean_force_equation_main} does not specify uniquely its ground state energy, since the equation is invariant under the transformation $H_\text{mf} \to H_\text{mf} + \delta$ for arbitrary real constant $\delta$. Commonly, this constant is fixed by the convention \cite{Roux1999,fleming,thingna2012,anders-oscillator,Guarnieri,Trushechkin_2021_MeanForce,miyashita,Subasi}:
\begin{eqnarray}
\label{eq:gauge}
\Tr_{SR}[e^{-\beta H}] = \Tr_S[e^{-\beta H_{\text{mf}}}] \Tr_R[e^{-\beta H_R}],
\end{eqnarray}
such that we have the following relation 
\begin{eqnarray} \label{mean_force_definition}
H_{\text{mf}}  = -\frac{1}{\beta} \log\left[\frac{\Tr_R[e^{-\beta H}]}{\mathcal{Z}_R}\right],
\end{eqnarray}
where $\mathcal{Z}_R = \Tr_R[e^{-\beta H_R}]$. Another way of fixing the gauge is 
to demand that $H_{\text{mf}}$ is traceless:
\begin{align}
\label{eq:gauge-trace}
    \Tr (H_{\text{mf}})=0
\end{align}
The latter gauge will be convenient, when extracting mean-force correction Hamiltonian numerically from a state 
(e.g. steady state of dynamics), as we do in 
\eqref{mean_force_state_RC},
while the former is more convenient, while deriving mean-force correction from the definition, as in Theorem \ref{mean_force_theorem}. 


In this paper, we concentrate solely on the second-order \textit{mean-force correction}  $H^{(2)}_{\text{mf}}$, defined via the expansion:
\begin{eqnarray}\label{mean_force_correction}
H_{\text{mf}} = H_0 + \lambda^2 H_\text{mf}^{(2)} + \dots,
\end{eqnarray}
where the zeroth-order term is fixed by putting $\lambda = 0$ (in the gague fixed by condition of traceless $H_{\text{mf}}$). 

\subsection{Quasi-steady state correction} \label{steady_state_main_section}
The last correction is defined with respect to the so-called quasi-steady Gibbs state:
\begin{eqnarray} \label{steady_state}
  \varrho \propto e^{-\beta H_{\text{st}}},
\end{eqnarray}
which is defined as the fixed point of the Liouvillian in the long-time limit: 
\begin{eqnarray} \label{steady_state_equation}
  \lim_{t \to \infty} \mathcal{L}_t[\varrho] \equiv \mathcal{L}_\infty [\varrho]  = 0.
\end{eqnarray}
Contrary to static mean-force state, $\varrho$ corresponds the dynamical equilibrium, defined according to particular Liouvillian in the long-time limit. In analogy to the mean-force and Lamb-shift, we are interested in the leading order correction defined by the expansion:
\begin{eqnarray} \label{steady_state_Hamiltonian_expansion}
H_{\text{st}} = H_0 + \lambda^2 H_\text{st}^{(2)} + \dots.
\end{eqnarray}
Then, to derive a solution for the correction $H^{(2)}_{\text{st}}$, we adapt the perturbative method. We expand the generator of the master equation:
\begin{align} \label{steady_state_eq_series}
\mathcal{L}_\infty [\varrho] = \mathcal{L}_{\infty}^{(0)}[\varrho] + \lambda^2 \mathcal{L}_{\infty}^{(2)}[\varrho] + \lambda^4 \mathcal{L}_{\infty}^{(4)}[\varrho] + \dots,
\end{align}
and similarly, we expand the postulated stationary state, i.e., 
\begin{eqnarray}
\varrho = \varrho_0 + \lambda^2 \varrho_2 + \lambda^4 \varrho_4 \dots ,
\end{eqnarray}
where by using the Dyson series for the lowest orders we get: 
\begin{align} \label{steady_state_expansion}
&\varrho_0 = e^{-\beta H_0}, \\ &\varrho_2 = -e^{-\beta H_0} \int_0^\beta dt \ e^{t H_0} H^{(2)}_\text{st} e^{-t H_0}.
\end{align}
A solution to Eq. \eqref{steady_state_equation} can be now constructed by requiring that  {\it in each order} of coupling strength $\lambda$, the action of generator $\mathcal{L}_\infty$ on the state should vanish.  Then, we obtain the following set of equations:
\begin{align}
&\mathcal{L}_{\infty}^{(0)}[\varrho_0] = 0 \label{zeroth_order_equation}\\
&\mathcal{L}_{\infty}^{(0)}[\varrho_2] + \mathcal{L}_{\infty}^{(2)}[\varrho_0] = 0, \label{second_order_equation} \\
&\mathcal{L}_{\infty}^{(0)}[\varrho_4] + \mathcal{L}_{\infty}^{(2)}[\varrho_2] + \mathcal{L}_{\infty}^{(4)}[\varrho_0] = 0,  \label{fourht_order_equation}\\
&\dots \nonumber
\end{align}
However, one should note that this is a stronger condition than \eqref{steady_state_equation}. 

Our main goal is to solve the following equations to find the solution for the quasi-steady state correction $H^{(2)}_{\text{st}}$. However, as it was highlighted in \cite{thingna2012}, the second-order equation \eqref{second_order_equation} only provides the solution for off-diagonal terms of the correction, whereas the solution for the diagonal part involves the fourth-order equation \eqref{fourht_order_equation}. In the following Section \ref{sect:steady_state_offdiag}, we present the general solution for off-diagonal elements for the Liouvillian \eqref{master_equation} and we provide some general methods for solving the diagonal part from the fourth-order equation (see Section \ref{sect:steady_state_diag}). In accordance, we derive explicit formulas of the quasi-steady state correction for specific types of Liouvillians introduced in the next Section \ref{sect:models}.   

Two additional remarks have to be mentioned. Firstly, the zeroth-order term $H_0$ in Eq. \eqref{steady_state_Hamiltonian_expansion}, leading to $\varrho_0 \propto e^{-\beta H_0}$, is consistent with the equation \eqref{zeroth_order_equation} for the Liouvillian $\mathcal{L}_\infty^{(0)} = i[\cdot, H_0]$.  However, unlike in the mean-force correction, we cannot 
obtain that $\varrho_0 \propto e^{-\beta H_0}$ by setting $\lambda=0$. 
Indeed, doing so we are left with 
the mentioned 
$\mathcal{L}_\infty^{(0)} = i[\cdot, H_0]$, and we see that 
$\mathcal{L}_\infty^{(0)}(\varrho_0)$ has many solutions. 
Hence, later we will provide the additional condition for a master equation, namely  the detailed balance condition. Then $\varrho_0$ can be obtained from 
higher order steady state by letting $\lambda$ to zero.

Secondly, in a manner similar to the mean-force correction, Eq. \eqref{steady_state}  
determines  $H_{\text{st}}$
up to a gauge. 
 
\blk

\section{Models of open systems } \label{sect:models}
In this section, we introduce three specific Liouvillians used in the modeling of quantum open systems: the Liouvillians of Davies ME, Bloch-Redfield ME, and the cumulant equation. The first two are interconnected through the so-called secular approximation, which ensures that Davies ME is completely positive. Later, we will specifically elucidate how this impacts the quasi-steady-state correction. On the other hand, the cumulant equation restores the positivity of the map by incorporating higher-order generators into the Liouvillian. This presents an interesting case study for addressing diagonal corrections that involve fourth-order generators.    

\subsubsection{Bloch-Redfield and Davies master equations}
Let us start with the most known master equations, i.e., the Bloch-Redfield and the Davies master equations. In general, for the Bloch-Redfield ME we define a time-dependent Kossakowski matrix and Lamb-shift coefficient (see Eq. \eqref{master_equation}): 
\begin{align}\label{gamma_redfield}
&K_{\alpha \beta}(\omega,\omega',t) = \Gamma^t_{\alpha \beta}(\omega') + \Gamma^t_{\beta \alpha}(\omega)^* 
\equiv \gamma^t_{\alpha \beta}(\omega,\omega'), \\
&\Upsilon^{(\text{LS})}_{\alpha \beta}(\omega,\omega',t) = \frac{1}{2i} [\Gamma^t_{\alpha \beta}(\omega') - \Gamma_{\beta \alpha}^t(\omega)^*] \equiv \mathcal{S}^t_{\alpha \beta}(\omega,\omega'),
\end{align} 
where 
\begin{eqnarray}
\Gamma^t_{\alpha \beta}(\omega) &=& \int_0^{t} ds \ e^{i \omega s} \langle R_\alpha (s) R_\beta (0) \rangle_{\gamma_R}. 
\end{eqnarray}
However, in this paper we are mostly interested in the long-time limit of the Liouvillian, based on which we define the quasi-steady state correction. For Bloch-Redfield ME we have: 
\begin{multline} \label{redfield_liouvillian}
    \mathcal{L}_\infty^{R} = i [\cdot , H_0] \\ 
    + \lambda^2 \sum_{\omega,\omega'} \sum_{\alpha \beta} \Big( i\mathcal{S}_{\alpha \beta}(\omega,\omega') [\cdot , A^\dag_{\alpha}(\omega) A_\beta(\omega')] \\ + \gamma_{\alpha \beta}(\omega,\omega') \mathcal{D}_{\alpha \beta}(\omega,\omega') \Big).
\end{multline}
where 
\begin{align}
    \mathcal{S}_{\alpha \beta}(\omega,\omega')= \lim_{t\to \infty}
\mathcal{S}^t_{\alpha \beta}(\omega,\omega')
\quad
\gamma_{\alpha \beta}(\omega,\omega')=\lim_{t\to \infty}\gamma^t_{\alpha \beta}(\omega,\omega')
\end{align}
It is well-known that the Bloch-Redfield equation, in general, does not preserve the positivity of the state since $\gamma_{\alpha \beta}(\omega,\omega')$ is not a positive semi-definite matrix. Commonly, this issue is solved by applying the so-called secular approximation, which leads to the Davies master equation in the GKSL form. In accordance, applying the secular approximation, we obtain the Kossakowski matrix for the Davies dynamics:
\begin{align} \label{secular_approx}
& \gamma_{\alpha \beta}(\omega,\omega') \xrightarrow{\text{sec. approx.}}  \gamma_{\alpha \beta} (\omega) \delta_{\omega,\omega'}, \\
& \mathcal{S}_{\alpha \beta}(\omega,\omega') \xrightarrow{\text{sec. approx.}}  \mathcal{S}_{\alpha \beta} (\omega) \delta_{\omega,\omega'}, \label{secular}
\end{align}
where 
\begin{equation} \label{pumping_dumping_rates}
    \gamma_{\alpha \beta} (\omega) \equiv \gamma_{\alpha \beta} (\omega,\omega)
    = \int_{-\infty}^{+\infty} ds \ e^{i \omega s} \langle R_\alpha (s) R_\beta (0) \rangle_{\gamma_R},  
\end{equation}
is the Fourier transform of the auto-correlation function, and
\begin{equation} \label{S_principal_value_main}
  \mathcal{S}_{\alpha \beta}(\omega) \equiv \mathcal{S}_{\alpha \beta}(\omega,\omega) = \mathcal{P} \frac{1}{2\pi}\int_{-\infty}^{+\infty}d\Omega~ \frac{\gamma_{\alpha \beta}(\Omega)}{\omega-\Omega},
\end{equation}
where $\mathcal{P}$ denotes the principal value integral. We also have the following relation:
\begin{eqnarray} \label{Gamma_real_and_imaginary}
\lim_{t \to \infty} \Gamma^t_{\alpha \beta}(\omega) \equiv  \Gamma_{\alpha \beta}(\omega) = \frac{1}{2} \gamma_{\alpha \beta}(\omega) + i \mathcal{S}_{\alpha \beta}(\omega).
\end{eqnarray}
Finally, the Davies generator is given by: 
\begin{multline}\label{Davies_long_time_limit}
    \mathcal{L}_\infty^{D} = i [\cdot , H_0] 
    + \lambda^2 \sum_{\omega} \sum_{\alpha \beta} \Big( i\mathcal{S}_{\alpha \beta}(\omega) [\cdot , A^\dag_{\alpha}(\omega) A_\beta(\omega)] \\ + \gamma_{\alpha \beta}(\omega) \mathcal{D}_{\alpha \beta}(\omega,\omega) \Big).
\end{multline}
However, we want to notice that to restore the positivity of the Bloch-Redfield master equation, it is enough to do the secular approximation only for the dissipative part. For this reason, we additionally consider a so-called \textit{non-secular Davies} defined as:
\begin{multline}
\label{eq:Davies-non-sec}
    \mathcal{L}_\infty^{D, ns} = i [\cdot , H_0] \\
    + \lambda^2 \sum_{\omega,\omega'} \sum_{\alpha \beta} \Big( i\mathcal{S}_{\alpha \beta}(\omega, \omega') [\cdot , A^\dag_{\alpha}(\omega) A_\beta(\omega')] \\ + \delta_{\omega,\omega'}\gamma_{\alpha \beta}(\omega) \mathcal{D}_{\alpha \beta}(\omega,\omega) \Big).
\end{multline}

In the following, we will also use the representation of the (time-dependent) Redfield generator in the interaction picture given by the expression:
\begin{multline}
    \tilde{\mathcal{L}}_t^{R} = \lambda^2 \sum_{\omega,\omega'} \sum_{\alpha \beta} \Big( i \tilde{\mathcal{S}}^t_{\alpha \beta}(\omega,\omega') [\cdot , A^\dag_{\alpha}(\omega) A_\beta(\omega')] \\ + \tilde{\gamma}^t_{\alpha \beta}(\omega,\omega') \mathcal{D}_{\alpha \beta}(\omega,\omega') \Big),
\end{multline}
with the following definitions:
\begin{eqnarray}
\tilde \gamma^t_{\alpha \beta}(\omega, \omega') &=& e^{i(\omega-\omega')t} \gamma^t_{\alpha \beta}(\omega, \omega'),  \\
  \tilde{\mathcal{S}}^t_{\alpha \beta}(\omega, \omega') &=&  e^{i(\omega-\omega')t} \mathcal{S}^t_{\alpha \beta}(\omega, \omega').
\end{eqnarray}

\subsubsection{Cumulant equation (refined weak-coupling)}
\label{subsec:cumulant}
Let us now introduce the cumulant equation 
\cite{Alicki89, Rivas_2017,Rivas_2019}. Unlike the previous models of open system, which are in the form of differential equations, the cumulant equation is introduced as the dynamical map (which does not involve the Markovian approximation):
\begin{eqnarray} \label{cumulant_equation_main}
  \tilde \rho(t) = e^{\tilde K^{(2)}_t} \tilde \rho(0),
\end{eqnarray}
where 
\begin{multline}
   \tilde K^{(2)}_t[\rho(0)] = \\ -\lambda^{2}\int_{0}^{t} dt_{1}\int_{0}^{t_1} dt_{2}\Tr_{R}\left(\left[H_{I}(t_1),\left[H_{I}(t_2),\rho_{S}(0)\otimes\rho_{R}\right]\right]\right)   
\end{multline}
is the generator of the map in the interaction picture, such that $\tilde \rho(t) = e^{i H_0 t} \rho(t) e^{-i H_0 t}$. 

It has been showed that cumulant equation is an alternative way to describe non-Markovian dynamics in the weak-coupling regime. Its essential feature is the GKSL form of the $\tilde K^{(2)}_t$ super-operator. In this way, the cumulant equation defines one parameter family of CPTP dynamical maps. This feature of the cumulant equation is its advantage over the Bloch-Redfield equation, for which the fundamental property of completely positive evolution is not satisfied.

Since the Lamb-shift and quasi-steady state corrections are defined according to the generator of the master equations, from the dynamical map \eqref{cumulant_equation_main} we derive the corresponding differential equation, namely
\begin{eqnarray} \label{cumulant_master_equation_interaction}
\frac{d}{dt} \tilde \rho(t) &=& \left[\left(\frac{d}{dt} e^{\tilde K^{(2)}_t}\right) e^{-\tilde K^{(2)}_t} \right] \tilde \rho(t).
\end{eqnarray}
This defines the Liouvillian of the cumulant equation in the interaction picture:
\begin{eqnarray}
\tilde{\mathcal{L}}_t^C &=& \left(\frac{d}{dt} e^{\tilde K^{(2)}_t}\right) e^{-\tilde K^{(2)}_t} \\
&=& \frac{d}{dt} \tilde K^{(2)}_t + \frac{1}{2} [\tilde K^{(2)}_t, \frac{d}{dt}\tilde K^{(2)}_t] + \dots
\end{eqnarray}
Interestingly, we have revealed that the cumulant super-operator $\tilde K^{(2)}_t$ is very much related to the Bloch-Redfield generator (in the interaction picture) by the following relation:
\begin{eqnarray} \label{cumulant_redfield_derivative}
  \frac{d\tilde{K}^{(2)}_t}{dt}= \tilde{\mathcal{L}}^{R}_t.
\end{eqnarray}
Finally, applying \eqref{cumulant_redfield_derivative} and transforming it to the Schr\"{o}dinger picture, we get 
 \begin{align} \label{cumulant_liouvillian}
 &\mathcal{L}^C_t [\rho] = \mathcal{L}_t^R[\rho] \nonumber \\
 &+ \frac{1}{2} \int_0^t ds \ e^{-i H_0 t} \left( \left[\tilde{\mathcal{L}}_s^R , \tilde{\mathcal{L}}_t^R \right] [e^{i H_0 t} \rho e^{-i H_0 t}] \right) e^{i H_0 t} \nonumber \\
 &+ O(\lambda^6).
 \end{align}
We see that up to the second-order the Liouvillian of the cumulant is equal to the Bloch-Redfield one. Nevertheless, cumulant equation provides non-trivial higher order generators that in principle leads to different predictions of quasi-steady state diagonal correction.  

\section{Results} \label{sect:results}
In this section, we derive formulas for different corrections and reveal their mutual relations. In the following, we will provide an explicit expression for all of the corrections \eqref{hamiltonian_correction}. The coefficients of the corrections will be written in the universal integral form given by:
\begin{equation} \label{integral_representation}
    \Upsilon^{(\mathrm{cor})}_{\alpha \beta}(\omega, \omega') = \frac{1}{2\pi} \mathcal{P} \int_{-\infty}^{+\infty} d \Omega \ D_\mathrm{cor}(\omega,\omega', \Omega) \ \gamma_{\alpha \beta}(\Omega),
\end{equation}
where $\gamma_{\alpha \beta}(\Omega)$ is the relaxation rate defined by equation \eqref{pumping_dumping_rates}. $\mathcal{P}$ denotes the principal value integral. We provide the kernel $D_\mathrm{cor}(\omega,\omega', \Omega)$ for a particular corrections and Liouvillians, which are finally summarised in Table \ref{tab:corrections_kernels}. 

\begin{widetext}
\begin{center}
\begin{table}[h!]
    \centering
    \begin{tabular}{|c|c|c|c|c|c|}
    \hline
    & \quad Davies \quad & \quad Davies (non-secular) \quad & \quad Bloch-Redfield \quad & \quad Cumulant \quad \\ \hline
    $D_{\text{LS}}(\omega,\omega,\Omega)$  & \multicolumn{4}{c|}{$\frac{1}{\omega - \Omega}$}  \\ \hline
    $D_{\text{LS}}(\omega,\omega',\Omega)$  &  0 & \multicolumn{3}{c|}{$\frac{1}{2}(\frac{1}{\omega - \Omega}+\frac{1}{\omega' - \Omega}) + \frac{i}{4} (\delta(\Omega-\omega) - \delta(\Omega-\omega'))$}  \\ \hline
    $D_{\text{st}}(\omega,\omega,\Omega)$ & 0 & \multicolumn{3}{c|}{\text{general form not derived here}} \\ \hline
    $D_{\text{st}}(\omega,\omega',\Omega)$ & 0 & $D_{\text{LS}}(\omega,\omega',\Omega)$  & \multicolumn{2}{c|}{$D_{\text{mf}}(\omega,\omega',\Omega)$} \\ \hline
    $D_{\text{mf}}(\omega,\omega,\Omega)$ & \multicolumn{4}{c|}{$\frac{1-e^{\beta(\omega-\Omega)}+\beta(\omega-\Omega)}{\beta (\omega-\Omega)^2}$ } \\  \hline
    $D_{\text{mf}}(\omega,\omega',\Omega)$ & \multicolumn{4}{c|}{$\frac{1}{\omega'-\Omega} - \frac{(\omega-\omega')(e^{\beta(\omega-\Omega)}-1)}{(\omega-\Omega)(\omega'-\Omega)(e^{\beta(\omega-\omega')}-1)}$ }      \\ \hline
    \end{tabular}
    \caption{Explicit kernels $D_\mathrm{cor}(\omega,\omega', \Omega)$ according to the representation \eqref{integral_representation} for all Hamiltonian corrections, i.e., $\mathrm{cor} = \text{LS}$ (Lamb-shift) \eqref{master_equation}, $\mathrm{cor}=\text{mf}$ (mean-force) \eqref{mean_force_correction} and $\mathrm{cor}=\text{st}$ (quasi-steady state) \eqref{steady_state_expansion}. Lamb-shift and quasi-steady state correction have been calculated according to the following Liouvillians: Davies \eqref{Davies_long_time_limit}, (non-secular) Davies \eqref{eq:Davies-non-sec}, Bloch-Redfield \eqref{redfield_liouvillian} and cumulant \eqref{cumulant_liouvillian}.} 
\label{tab:corrections_kernels}
\end{table}
\end{center} 
\end{widetext}
\subsection{Mean-force correction}
We are ready to state our first main result regarding the mean-force correction.
\begin{theorem} \label{mean_force_theorem}  In the gauge \eqref{eq:gauge},
the mean-force correction is given by 
\begin{eqnarray} \label{hamiltonian_correction-mf}
H_{\text{mf}}^{(2)}(t) = \sum_{\omega,\omega'} \sum_{\alpha, \beta} \Upsilon^{(\text{mf})}_{\alpha \beta}(\omega,\omega',t) A^\dag_{\alpha}(\omega) A_\beta(\omega'),
\end{eqnarray}
with the coefficients:
\begin{align} \label{mean_force_without_poles}
\Upsilon^{(\text{mf})}_{\alpha \beta}(\omega, \omega') &= \frac{1}{2 \pi}\int_{-\infty}^{+\infty}  d \Omega \  D_\text{mf}(\omega, \omega', \Omega) \ \gamma_{\alpha \beta}(\Omega), \\  
D_\text{mf}(\omega, \omega', \Omega) &= \frac{1}{\omega'-\Omega} - \frac{(\omega-\omega')(e^{\beta(\omega-\Omega)}-1)}{(\omega-\Omega)(\omega'-\Omega)(e^{\beta(\omega-\omega')}-1)},
\end{align}
or equivalently, in terms of the $\mathcal{S}_{\alpha \beta}(\omega)$ function \eqref{S_principal_value_main}, it takes the form:
\begin{multline} \label{mean_force_coef_S}
\Upsilon_{\alpha \beta}^{(\text{mf})}(\omega,\omega')=
     \frac{1}{e^{\beta \omega}-e^{\beta \omega'}} (e^{\beta \omega}\mathcal{S}_{\alpha \beta}(\omega') - e^{\beta \omega'} \mathcal{S}_{\alpha \beta}(\omega) \\ + e^{\beta (\omega+\omega')} \left( \mathcal{S}_{\beta \alpha}(-\omega') - \mathcal{S}_{\beta \alpha}(-\omega) \right)).
\end{multline}
\end{theorem}
\begin{remark}
What is interesting about the expression \eqref{mean_force_without_poles} is that despite of its form, it does not exhibit any poles, which  is quite  unusual for a  Lamb-shift correction. Thus, the principal value integral is not needed in this case (cf. Eq. \eqref{integral_representation}).
\end{remark}
\begin{proof}
The sketch of the proof is as follows. To solve Eq. \eqref{mean_force_equation_main} we write the exponents from both sides of the equality in terms of the Dyson series, which formally can be expressed via the time-ordering operator $\mathcal{T}$ as:
\begin{align}
    &e^{-\beta H_{\text{mf}}} \equiv e^{-\beta (H_0 + \delta H_\text{mf})} = e^{-\beta H_0} \mathcal{T}  e^{- \int_0^\beta dt \ \delta \hat {H}_\text{mf} (t)},  \\ 
    &e^{-\beta H} \equiv  e^{-\beta (H_0 + H_R + \lambda H_I)} = e^{-\beta (H_0+H_{R})} \mathcal{T} e^{- \lambda \int_0^\beta dt \ \hat H_I(t)}, 
\end{align}
where we put $\delta {H}_\text{mf} = H_\text{mf} - H_0$ and we define an imaginary-time-dependent operators $\hat A(t) = e^{t (H_0+H_R)} A e^{-t (H_0+H_R)}$. Consequently, using gauge \eqref{eq:gauge} the equality \eqref{mean_force_equation_main} can be rewritten in the form:
\begin{equation}
    \Tr_R \left[ \mathcal{T} \left( e^{- \int_0^\beta dt \ \delta \hat H_{\text{mf}} (t)} -  e^{- \lambda \int_0^\beta dt \ \hat H_I(t)} \right) \gamma_R \right] = 0.
\end{equation}
where $\gamma_R$ is the Gibbs state with respect to $H_R$. 
Note that this is exact for arbitrary coupling strength $\lambda$. Then, considering the weak-coupling limit ($\lambda \ll 1$), we expand the above equality and obtain, within the second-order, the following condition for the mean-force correction: 
\begin{eqnarray} 
 \int_0^\beta dt \ \hat H_\text{mf}^{(2)}(t)  = -\int_0^\beta dt \int_0^{t} ds \left \langle \hat H_I(t) \hat H_I(s) \right \rangle_{\gamma_R}
\end{eqnarray}
After substituting the representation of the mean-force Hamiltonian given by Eq. \eqref{hamiltonian_correction} and the interaction term \eqref{interaction_sum} with the definition of jump operators \eqref{jump_operators_definition}, we are able to calculate the coefficients $\Upsilon^{(\text{mf})}_{\alpha \beta}(\omega, \omega')$ (see the detailed proof in Appendix \ref{mean_force_section}).
\end{proof}
Theorem \ref{mean_force_theorem} is the first derivation in the literature of the second-order mean-force Hamiltonian for a general weak-coupling of the form \eqref{total_hamiltonian} (this was independently done in  \cite{trushechkin2022};  the expression for the correction to mean-force {\it Gibbs  state} has been given earlier in \cite{Subasi,miyashita,Trushechkin_2021_MeanForce}). One observes that the coefficients are symmetric in $\omega$'s, i.e., $\Upsilon^{(\text{mf})}_{\alpha \beta}(\omega, \omega') = \Upsilon^{(\text{mf})}_{\alpha \beta}(\omega', \omega)$, which together with $\mathcal{S}^*_{\alpha \beta}(\omega) = \mathcal{S}_{\beta \alpha}(\omega)$, ensures the hermiticity of the Hamiltonian.

According to this expression, let us compare the dynamical Hamiltonian with the mean-force. Using Eq. \eqref{Gamma_real_and_imaginary}, we write down the  Lamb-shift correction in terms of $\gamma_{\alpha \beta}(\omega)$ and $\mathcal{S}_{\alpha \beta}(\omega)$, such that for the Bloch-Redfield ME we get:
\begin{multline}
\Upsilon^{(\text{LS})}_{\alpha \beta}(\omega,\omega') 
= \frac{1}{2} (\mathcal{S}_{\alpha \beta} (\omega) + \mathcal{S}_{\alpha \beta} (\omega')) \nonumber \\ + \frac{i}{4} (\gamma_{\alpha \beta} (\omega) -\gamma_{\alpha \beta} (\omega')),
\end{multline}
whereas for the Davies equation the off-diagonal elements vanish (due to the secular approximation). It is seen that the mean-force correction is different than the  Lamb-shift one; in particular, the  Lamb-shift correction has non-zero anti-hermitian part (in indices $\alpha,\beta$) in contrast to the hermitian coefficients of the mean-force.
Indeed,  
$\gamma_{\alpha\beta}(\omega)$ and $S_{\alpha\beta}(\omega)$  are hermitian matrices, and therefore the second term of Lamb-shift correction is non-hermitian, while there is no such term in mean-force correction. 

\subsection{Quasi-steady state correction} 
In this section, we propose the general formulas for quasi-steady state correction in terms of the coefficients $\Upsilon_{\alpha \beta}^{(\text{st})}(\omega, \omega')$:
\begin{eqnarray} \label{hamiltonian_correction-st}
H_{\text{st}}^{(2)}(t) = \sum_{\omega,\omega'} \sum_{\alpha, \beta} \Upsilon^{(\text{st})}_{\alpha \beta}(\omega,\omega',t) A^\dag_{\alpha}(\omega) A_\beta(\omega').
\end{eqnarray}
First we consider a problem of specifying the off-diagonal contribution to the Hamiltonian $H_{\text{st}}^{(2)}$ (in terms of the coefficients $\Upsilon_{\alpha \beta}^{(\text{st})}(\omega, \omega')$ for $\omega \neq \omega'$), and later the diagonal one (given by $\Upsilon_{\alpha \beta}^{(\text{st})}(\omega, \omega)$).
 
We start with off-diagonal elements, because (as indicated in  \cite{thingna2012}) the first nontrivial correction for diagonal states one can get only in the fourth order. So we start with simpler case of off-diagonal elements.  

\subsubsection{Off-diagonal elements} \label{sect:steady_state_offdiag}
To derive a solution for off-diagonal elements of the Hamiltonian $H_{\text{st}}^{(2)}$, one needs to specify the solution of the second-order equation (see Eq. \eqref{second_order_equation}):  
\begin{eqnarray} \label{second_order_equation_repeated}
    \mathcal{L}_{\infty}^{(0)}[\varrho_2] + \mathcal{L}_{\infty}^{(2)}[\varrho_0] = 0.
\end{eqnarray}
To provide of such a solution we propose the following representation:
\begin{align} \label{second_order_method}
    \mathcal{L}^{(k)}_\infty[\rho_l] = e^{-\beta H_0} \sum_{\alpha, \beta} \sum_{\omega,\omega'} g_{\alpha \beta}^{(kl)}(\omega,\omega')  A_\alpha(\omega)A_\beta(\omega'),
\end{align}
for $l+k=2$, i.e. $k=0,l=2$ and $k=2,l=0$ (see Appendix \ref{appendix:second_order}). Our methodology is based on replacing the equation for operators \eqref{second_order_equation_repeated} by the algebraic equation, which is stated in the following Lemma:
\begin{lemma}
Eq. \eqref{second_order_equation} can be rewritten as:
\begin{multline}
\sum_{\alpha, \beta} \sum_{\omega,\omega'} (g_{\alpha \beta}^{(02)}(\omega,\omega')+g_{\alpha \beta}^{(20)}(\omega,\omega')) A_\alpha(\omega)A_\beta(\omega') = 0,
\end{multline}
which, in particular, is satisfied if:
\begin{eqnarray}
g_{\alpha \beta}^{(02)}(\omega,\omega')+g_{\alpha \beta}^{(20)}(\omega,\omega') = 0 
\end{eqnarray}
for all $\omega,\omega'$ and $\alpha, \beta$.
\end{lemma}

Next, we specify the explicit formula for coefficients of the most common models of master equations. 
\begin{lemma}
For the general form of the second-order Liouvillian  \eqref{master_equation}:
\begin{align}
g_{\alpha\beta}^{(02)}(\omega,\omega') &= -i (\omega+\omega') \Upsilon^{(\text{st})}_{\alpha \beta}(-\omega, \omega') \alpha(\omega+\omega'), \\
g_{\alpha\beta}^{(20)}(\omega,\omega') &= i\Upsilon_{\alpha \beta}^{(\text{LS})}(-\omega,\omega') (1 - e^{-\beta(\omega+\omega')}) \nonumber \\
& + e^{\beta \omega} K_{\beta \alpha}(-\omega',-\omega) \nonumber \\
&- \frac{1}{2} K_{\alpha \beta}(-\omega, \omega') (e^{-\beta (\omega + \omega')} + 1),
\end{align}
where $\alpha(\omega) =  \int_0^\beta dt \ e^{- t \omega}$.
\end{lemma}

Finally, combining together those two Lemmas, we propose the following Theorem:
\begin{theorem} \label{steady_state_theorem}
The second-order equation 
\eqref{second_order_equation} is satisfied by the state $\varrho \propto e^{-\beta(H_0 + \lambda^2 H_\text{st}^{(2)} + \dots)}$ if
\begin{align} 
&(i) \ \  K_{\alpha \beta}(\omega,\omega) = K_{\beta \alpha}(-\omega,-\omega) e^{\beta \omega}, \label{detailed_balance}\\
&(ii) \ \Upsilon_{\alpha \beta}^{(\text{st})}(\omega, \omega')  = \Upsilon_{\alpha \beta}^{(\text{LS})}(\omega,\omega')  + \frac{i}{e^{\beta \omega} - e^{\beta \omega'}} \nonumber \\ 
&\times \left(K_{\beta \alpha}(-\omega,-\omega') e^{\beta (\omega+\omega')} - \frac{1}{2}  K_{\alpha \beta}(\omega',\omega) (e^{\beta \omega} + e^{\beta \omega'}) \right) \label{steady_vs_dynamical_coef}
\end{align}
for $\omega \neq \omega'$.
\end{theorem}
The condition \eqref{detailed_balance} is the so-called detailed-balance relation, which is satisfied for all considered here Liouvillians (i.e., for Davies, Bloch-Redfield and cumulant). Then, let us compare the mean-force correction with the quasi-steady state correction for specific choices of Kossakowski matrix $K_{\alpha \beta}(\omega,\omega')$ and Lamb-shift $\Upsilon^{(\text{LS})}_{\alpha \beta }(\omega,\omega')$. First, let us observe that (see proof in Appendix \ref{appendix:steady_state_mean_force_proof}):
\begin{corollary} \label{corollary_steady_state_mean_force1}
If $K_{\alpha \beta}(\omega,\omega') = \gamma_{\alpha \beta}(\omega,\omega')$ and $\Upsilon^{(\text{LS})}_{\alpha \beta }(\omega,\omega') = \mathcal{S}_{\alpha \beta}(\omega,\omega')$, then for $\omega \neq \omega'$:
\begin{eqnarray}
  \Upsilon_{\alpha \beta}^{(\text{st})}(\omega, \omega') = \Upsilon_{\alpha \beta}^{(\text{mf})}(\omega,\omega').
\end{eqnarray}
\end{corollary}
\begin{remark}
    Note that this relation is gauge independent, as the gauge only affects diagonal corrections.
\end{remark}
 As it follows from Eq. \eqref{redfield_liouvillian} and \eqref{cumulant_liouvillian}) this is the case for the Bloch-Redfield and the truncated cumulant Liouvillian. However, if the secular approximation is applied for the Kossakowski matrix, we get: 
\begin{corollary} \label{corollary_steady_state_mean_force2}
If  $K_{\alpha \beta}(\omega,\omega') = \gamma_{\alpha \beta}(\omega) \delta_{\omega,\omega'}$, then for $\quad \omega \neq \omega'$:
\begin{eqnarray}
    \Upsilon_{\alpha \beta}^{(\text{st})}(\omega, \omega') = \Upsilon_{\alpha \beta}^{(\text{LS})}(\omega,\omega').
\end{eqnarray}
\end{corollary}
This is the case for the so-called non-secular Davies \eqref{eq:Davies-non-sec}. However, commonly the secular approximation is also applied for the Lamb-shift term, such that for the standard Davies equation, the off-diagonal elements vanish, i.e., $\Upsilon_{\alpha \beta}^{(\text{st})}(\omega, \omega') = \Upsilon_{\alpha \beta}^{(\text{LS})}(\omega,\omega') = 0$ for $\omega\not=\omega')$. 

Comparing Corollary \ref{corollary_steady_state_mean_force1} and \ref{corollary_steady_state_mean_force2} we see an interesting interplay between all three corrections. The specific (non-diagonal) form of the Lamb-shift and Kossakowski matrix for the Bloch-Redfield generator (and second-order contribution to cumulant as well) provides a coincidence of the mean-force and quasi-steady state correction for off-diagonal elements. In particular, these non-zero values results in the so-called ``steady-state coherences'' of the equilibrium density matrix $\varrho$ discussed in \cite{Guarnieri,cattaneo2021comment,brumer}. On the contrary, by applying the secular approximation (i.e., making the Kossakowski matrix diagonal in $\omega$'s), the off-diagonal elements of the quasi-steady state correction rather coincide with the Lamb-shift one. 
\begin{widetext}
\begin{center}
    \begin{table}
    \centering
    \begin{tabular}{|c|c|c|c|c|c|}
    \hline
    & \quad Davies \quad & \quad Davies (non-secular) \quad & \quad Bloch-Redfield \quad & \quad Cumulant \quad \\ \hline
    $D_{\text{LS}}(-\omega,-\omega, \Omega) - D_{\text{LS}}(\omega,\omega, \Omega)$  & \multicolumn{4}{c|}{$\frac{2\omega}{\Omega^2-\omega^2}$}  \\ \hline
     $D_{\text{LS}}(\omega,0, \Omega) - D_{\text{LS}}(0,-\omega, \Omega)$  &  0 & \multicolumn{3}{c|}{$-\frac{\omega}{\Omega^2-\omega^2} + \frac{i}{4} (\delta(\Omega-\omega) + \delta(\Omega+\omega) - 2 \delta(\Omega) )$}  \\ \hline
    $D_{\text{st}}(-\omega,-\omega, \Omega) - D_{\text{st}}(\omega,\omega, \Omega)$ & \multicolumn{3}{c|}{0} & 
    $\frac{1-e^{-\beta(\omega+\Omega)}}{\beta(\omega+\Omega)^2} - \frac{1-e^{\beta(\omega-\Omega)}}{\beta(\omega-\Omega)^2}$
    \\ \hline
    $D_{\text{st}}(\omega,0, \Omega) - D_{\text{st}}(0,-\omega, \Omega)$& 0 & $D_{\text{LS}}(\omega,0, \Omega) - D_{\text{LS}}(0,-\omega, \Omega)$  & \multicolumn{2}{c|}{  $D_{\text{mf}}(\omega,0, \Omega) - D_{\text{mf}}(0,-\omega, \Omega)$} \\ \hline
    $D_{\text{mf}}(-\omega,-\omega, \Omega) - D_{\text{mf}}(\omega,\omega, \Omega)$ & \multicolumn{4}{c|}{$\frac{1-e^{-\beta(\omega+\Omega)}}{\beta(\omega+\Omega)^2} - \frac{1-e^{\beta(\omega-\Omega)}}{\beta(\omega-\Omega)^2} + \frac{2\omega}{\Omega^2-\omega^2}$
    } \\  \hline
    $D_{\text{mf}}(\omega,0, \Omega) - D_{\text{mf}}(0,-\omega, \Omega)$& \multicolumn{4}{c|}{$\frac{\omega  e^{-\beta  \Omega } \left(\omega  \left(e^{\beta  \omega }+1\right) \left(e^{\beta  \Omega }-1\right)-\Omega  \left(e^{\beta  \omega }-1\right) \left(e^{\beta  \Omega }+1\right)\right)}{\Omega  \left(e^{\beta  \omega }-1\right) (\Omega -\omega ) (\omega +\Omega )} $}      \\ \hline
    \end{tabular}
    \caption{The relevant kernels $D_\mathrm{cor}(\omega,\omega', \Omega)$ for a two-level system according to the Pauli representation \eqref{qubit_coefficients}. See a detailed description in Table \ref{tab:corrections_kernels}.}
    \label{tab:qubit_table}
\end{table}
\end{center}
\end{widetext}

\subsubsection{Diagonal elements} \label{sect:steady_state_diag}

To provide the diagonal part of the correction $H_\text{st}^{(2)}$, (given by the coefficients $\Upsilon_{\alpha \beta}^{(\text{st})}(\omega,\omega)$), the fourth-order equation must be solved \eqref{fourht_order_equation}, i.e.,
\begin{eqnarray} \label{fourht_order_equation_repeated}
\mathcal{L}_{\infty}^{(0)}[\varrho_4] + \mathcal{L}_{\infty}^{(2)}[\varrho_2] + \mathcal{L}_{\infty}^{(4)}[\varrho_0] = 0.
\end{eqnarray}
Contrary to the solution for off-diagonal terms, it is much more complex problem. For that reason and the clarity of presentation, we simplify the model, such that throughout of this section the interaction Hamiltonian is given by $H_I =  A \otimes R$ (i.e., we replace the sum \eqref{interaction_sum} by a single term, which can be generalized by adding the corresponding indices). Accordingly, we also simplify notation, such that $\Upsilon_{\alpha \beta}^{(\mathrm{cor})} \equiv \Upsilon_\mathrm{cor}$. Then, similarly to the second-order (see Eq. \eqref{second_order_method}), we start with writing the action of the generators in the basis of jump operators:
\begin{equation}\label{fourth_order_representation}
     \mathcal{L}^{(n)}_\infty[\varrho_l] = e^{-\beta H_0} \sum_{\vec \omega} g_{nl}(\vec \omega)  A(\vec \omega) 
\end{equation}
where $l+n=4$, $\vec \omega = (\omega_1,\omega_2,\omega_3,\omega_4)$ and $A(\vec \omega) \equiv A(\omega_1)A(\omega_2)A(\omega_3)A(\omega_4)$. The fourth-order equation \eqref{fourht_order_equation_repeated} takes then the following form:
\begin{equation} 
\sum_{\vec \omega} (g_{04}(\vec \omega) + g_{22}(\vec \omega) + g_{40}(\vec \omega)) A(\vec \omega) = 0,
\end{equation}
from which one can further derive the set of equations (for each $k$):
\begin{equation} \label{fourth_order_equation_with_g}
\sum_{\vec \omega} (g_{22}(\vec \omega) + g_{40}(\vec \omega)) \Pi_k A(\vec \omega) \Pi_k = 0,
\end{equation}
where $\Pi_k = \Pi(\epsilon_k)$ is a projector into subspace with energy $\epsilon_k$ (see Eq. \eqref{jump_operators_definition}). Disappearance of the $g_{04}$ term comes from the specific form of the $\mathcal{L}^{(0)}_\infty = i[H_0, \cdot]$, namely 
\begin{align}
    &\sum_{\vec \omega} g_{04}(\vec \omega) \Pi_k A(\vec \omega) \Pi_k = i e^{\beta \epsilon_k} \Pi_k [H_0, \varrho_4] \Pi_k = 0.
\end{align}
Furthermore, we propose the following Lemma
proved in Appendix \ref{appendix_fourth_order}:
\begin{lemma}
    $\Pi(\epsilon_k) A(\vec \omega) \Pi(\epsilon_k) \neq 0$ only if $\vec \omega$ belongs to the set  of all four-tuples of the form: 
\begin{align}
    \vec \omega=(\epsilon_{m_1}-\epsilon_k,\epsilon_{m_2}-\epsilon_{m_1},\epsilon_{m_3}-\epsilon_{m_2},\epsilon_k-\epsilon_{m_3}).
\end{align}
\end{lemma}
  Finally, since the operators $\Pi_k A(\vec \omega) \Pi_k$ are linearly independent (for different $k$), the following proposition follows:
\begin{proposition} \label{diagonal_elements_summation_proposition}
Eq. \eqref{fourth_order_equation_with_g} is satisfied if and only if for each $k$, such that $\Pi_k A(\vec \omega) \Pi_k \neq 0$, we have:
\begin{align} \label{final_g22_g44_eq}
    \sum_{\vec \omega \in G( \ket{k} \to  \ket{k})}
    \Big( g_{22}(\vec \omega) + g_{40}(\vec \omega)\Big) = 0.
\end{align}
\end{proposition}
Eqs. \eqref{final_g22_g44_eq} provide necessary conditions for the coefficients $\Upsilon_{\text{st}}(\omega,\omega')$ (encoded in the function $g_{22}$) to be a solution of the quasi-steady state. For the general form of the second-order Liouvillian  \eqref{master_equation} the function $g_{22}$ is explicitly given by:
\begin{multline} \label{general_g22}
g_{22}(\vec \omega) = \Upsilon_{\text{st}}(-\omega_3,\omega_4) \alpha(\omega_3+\omega_4) \\ 
\times \left(i \Upsilon_{\text{LS}}(-\omega_1,\omega_2) + \frac{1}{2} K(-\omega_1,\omega_2) \right) \\ 
- \Upsilon_{\text{st}}(-\omega_1,\omega_2) \alpha(\omega_1+\omega_2) \\ 
\times \left(i \Upsilon_{\text{LS}}(-\omega_3,\omega_4) - \frac{1}{2} K(-\omega_3,\omega_4) \right) \\ 
- e^{-\beta \omega_1} \Upsilon_{\text{st}}(-\omega_2,\omega_3) \alpha(\omega_2+\omega_3) K(-\omega_4,\omega_1) 
\end{multline}
where $\alpha(\omega) =  \int_0^\beta dt \ e^{- t \omega}$. 

From Eqs. \eqref{final_g22_g44_eq} one can get a solution for diagonal elements of the quasi-steady state correction $\Upsilon_{\text{st}}(\omega,\omega)$, for a particular form of the function $g_{40}$ (derived from the fourth-order Liouvillian $\mathcal{L}^{(4)}_\infty$). For Bloch-Redfield and Davies master equation, the Liouvillian is only defined up to the second-order, which trivially implies $g_{40}(\vec \omega) = 0$. On the contrary, the cumulant equation provides the fourth-order Liouvillian in the form: 
\begin{equation} \label{fourth_order_liouvillian_cumulant}
\mathcal{L}^{(4)}_t [\rho] 
= \frac{1}{2} \int_0^t ds \ e^{-i H_0 t} \left( \left[\tilde{\mathcal{L}}_s^R , \tilde{\mathcal{L}}_t^R \right] \left(e^{i H_0 t} \rho e^{-i H_0 t}\right) \right) e^{i H_0 t}. 
\end{equation}

The relevant coefficients $g_{40}$ can be derived from the projection of the Eq. \eqref{fourth_order_representation}, i.e.,
\begin{equation}
     \Pi_k \mathcal{L}^{(4)}_\infty[\varrho_0] \Pi_k = e^{-\beta \epsilon_k} \sum_{\vec \omega} g_{40}(\vec \omega)  \Pi_k A(\vec \omega) \Pi_k,  
\end{equation}
where according to \eqref{fourth_order_liouvillian_cumulant}:
\begin{equation}
     \label{eq:L4} \Pi_k\mathcal{L}^{(4)}_\infty[\varrho_0] \Pi_k =  \frac{1}{2} \lim_{t\to \infty} \int_0^t ds \ \Pi_k \left( \left[\tilde{\mathcal{L}}_s^R , \tilde{\mathcal{L}}_t^R \right] [\varrho_0] \right) \Pi_k.
\end{equation}
The explicit form of the function $g_{40}$ for the cumulant is provided in the Appendix, Eq. \eqref{g40_expression_for_cumulant}.

Let us remark, that 
the limits of such expressions as 
$\mathcal{L}^{(n)}_\infty[\varrho_l]$ may not exist (due to oscillating phases).
However after sandwiching 
them with $\Pi_k$ (as in \eqref{eq:L4}, the limit already exists.

Let us note, that for Davies (secular) equation, the diagonal correction vanishes for simple reason - namely the full generator
annihilates Gibbs state according to bare Hamiltonian (i.e., $\mathcal{L}_\infty^{D}[\varrho_0] = 0$). 

In the next subsection, we apply the proposed results for the two-level system, where we derive a diagonal quasi-steady state correction for considered here Liouvillians (i.e., Davies, Bloch-Redfield and cumulant). For Davies equation we confirm the above remark, obtaining vanishing correction. We will also get vanishing correction for Bloch-Redfield, but we do not know whether it holds in general, like in Davies case. 

\subsection{Two-level system} \label{sect:two_level_system}

Let us consider a two-level system with the bare Hamiltonian $H_0 = - \frac{\omega_0}{2} \sigma_z$ coupled to the thermal bath via a single interaction term $H_I = A \otimes R$, where $A = \vec r \cdot \vec \sigma$, $\vec r = (x,y,z)$, and $\vec \sigma = (\sigma_x,\sigma_y,\sigma_z)$ are the Pauli matrices. Then, we define 
\begin{align}
    &\Sigma^{\rm off}_\mathrm{cor}(\omega) \equiv \Upsilon_\mathrm{cor}(\omega,0)-\Upsilon_\mathrm{cor}(0, -\omega), \label{sigma_off} \\
    &\Sigma^{\rm diag}_\mathrm{cor}(\omega) \equiv \Upsilon_\mathrm{cor}(-\omega, -\omega)-\Upsilon_\mathrm{cor}(\omega, \omega), \label{sigma_diag}
\end{align}
such that the correction $H_\mathrm{cor}^{(2)}$ in the Pauli basis is given by:
\begin{align} \label{qubit_coefficients}
    \Tr[H_\mathrm{cor}^{(2)} \sigma_x] &= xz \ \Sigma^{\rm off}_\mathrm{cor} (\omega_0), \\ 
    \Tr[H_\mathrm{cor}^{(2)} \sigma_y] &= yz \ \Sigma^{\rm off}_\mathrm{cor} (\omega_0), \\ 
    \Tr[H_\mathrm{cor}^{(2)} \sigma_z] &= \frac{x^2 + y^2}{2} \ \Sigma^{\rm diag}_\mathrm{cor}(\omega_0). 
\end{align}
Without loss of generality, we have introduced a gauge $\Tr H_\mathrm{cor}^{(2)} =0$ since, as we discussed in the Section \ref{sect:hamiltonian_correction}, all corrections are defined up to an arbitrary constant. 

For the Liouvillian \eqref{master_equation}, the off-diagonal term $\Sigma^{\rm off}_\mathrm{cor}(\omega)$ can be calculated from Eq. \eqref{steady_vs_dynamical_coef}, which for particular choices of master equations is given by Corollary \ref{corollary_steady_state_mean_force1} and \ref{corollary_steady_state_mean_force2}. Furthermore, for diagonal elements we provide the following Proposition: 
\begin{proposition} \label{diagonal_theorem}
For the Liouvillian \eqref{master_equation}, obeying the detailed balance condition: $K(\omega,\omega) = e^{\beta \omega} K(-\omega,-\omega)$, diagonal elements of the two-level system quasi-steady state correction are given by: 
\begin{enumerate} [label=(\roman*)]
    \item for $g_{40}(\vec \omega) = 0$ (Davies and Bloch-Redfield): 
    \begin{eqnarray}
    \Sigma^{\rm diag}_{\rm QSS}(\omega) = 0
    \end{eqnarray}
    \item for $g_{40}(\vec \omega)$ of the cumulant Liouvillian \eqref{fourth_order_liouvillian_cumulant}:  
    \begin{align} \label{steady_state_cumulant_solution_theorem}
     \Sigma^{\rm diag}_{\rm QSS}(\omega) =  \Sigma^{\rm diag}_{\rm MF}(\omega) + \mathcal{S}(-\omega) - \mathcal{S}(\omega)
    \end{align}
\end{enumerate}
\end{proposition}
For more details refer to Appendix \ref{steady_state_section}.
We see that for the two-level system the absence of higher-order Liouvillians ($g_{40} = 0$) results in no correction to the diagonal of the quasi-steady state. On the contrary,  Liouvillian of the cumulant
truncated to the fourth order provides a non-trivial correction given by Eq. \eqref{steady_state_cumulant_solution_theorem}. Since we expect that proper thermalization should result in mean-force Hamiltonian, one can interpret $\mathcal{S}(-\omega) - \mathcal{S}(\omega)$ as the error. In the next section, we analyze numerically this discrepancy for a particular spin-boson model. Notice also that $\mathcal{S}(\omega) = \Upsilon_\text{LS}(\omega,\omega)$ (for cumulant, Bloch-Redfield and Davies). 

We summarize all of the explicit formulas for the two-level system in Table \ref{tab:qubit_table}.

\begin{figure*}
    \centering
    \includegraphics[width=0.49\linewidth]{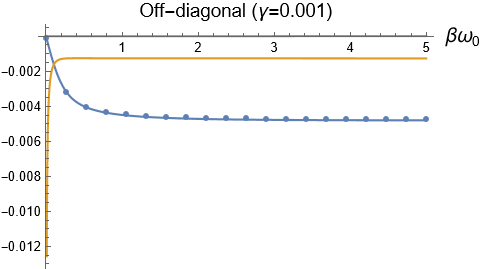}
    \includegraphics[width=0.49\linewidth]{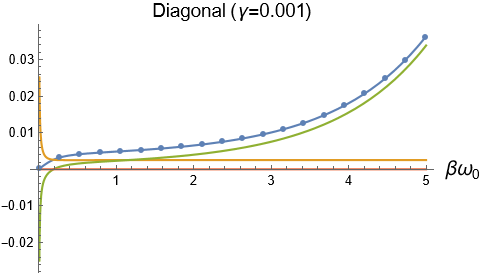}
    \includegraphics[width=0.49\linewidth]{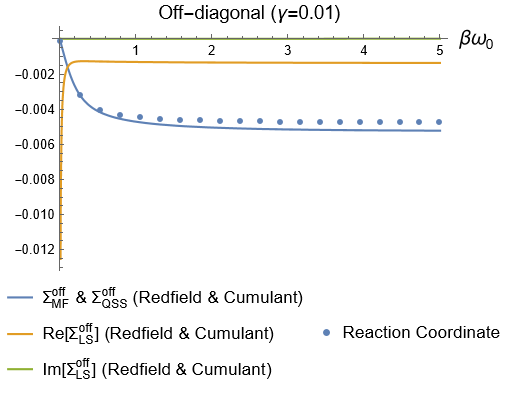}
    \includegraphics[width=0.49\linewidth]{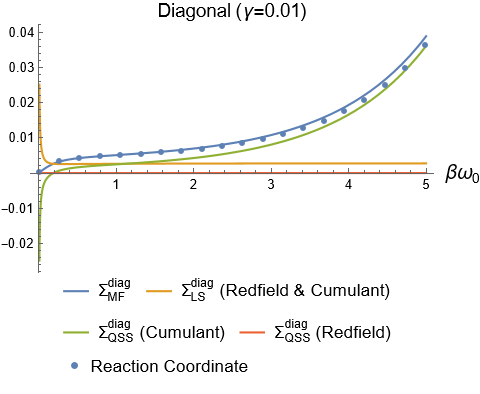}
    \caption{In the figure we present the coefficients $\Sigma^{\rm off}_\mathrm{cor}(\omega)$ and $\Sigma^{\rm diag}_\mathrm{cor}(\omega)$ that compose different Hamiltonian corrections  (Eqs. \eqref{sigma_off} and \eqref{sigma_diag}), computed for the spin-boson model (Eq. \eqref{spin_boson_hamiltonian}) for the spectral density Eq. \eqref{spectral_density_underdumped} and for $\omega_{rc} = 20 \omega_0$. As we proved (see Corollary \ref{corollary_steady_state_mean_force1}), the off-diagonal correction of the mean-force coincide with the quasi-steady state correction of the Bloch-Redfield and cumulant equation. However, for the diagonal case the cumulant quasi-steady state correction is much closer to the mean-force than the Bloch-Redfield one, which is zero for all values of $\beta \omega_0$. The discrepancy between cumulant and mean-force for the diagonal correction is precisely given by the Lamb-shift correction (see Eq. \eqref{steady_state_cumulant_solution_theorem}), which for the particular spectral density is small.    
    We also calculate the mean-force correction based on the reaction coordinate Gibbs state (Eqs. \eqref{mean_force_state_RC} and \eqref{second_order_correction_from_mean_force_state}). 
    From Eq. \eqref{mean_force_state_RC} the reaction coordinate Gibbs state should converge to the exact mean-force correction when $\gamma \to 0$. As expected, we observe a very good agreement for $\gamma = 0.001$, whereas some discrepancy is present for $\gamma = 0.01$. This validates our analytical formulas for the mean-force correction and shows the limits of the reaction coordinate method.}   
    \label{fig:corrections}
\end{figure*}

\section{Spin-boson model: numerical study}
Let us now apply our results to the well-known spin-boson model. We consider a qubit coupled to a bosonic bath with the Hamiltonian:
\begin{equation} \label{spin_boson_hamiltonian}
H= -\frac{\omega_0}{2} \sigma_z +\sum_k \Omega_k a^{\dagger}_k a_k + \vec r \cdot \vec \sigma \sum_{k} \lambda_k (a_k +a_{k}^{\dagger}) 
\end{equation}
where $a_k$ and $a_k^\dag$ are the bosonic annihilation and creation operators with a spectral density:
\begin{eqnarray}
J(\omega) = \sum_k \lambda_k^2 \delta(\omega_k - \omega).
\end{eqnarray}
By extending a domain of the spectral density, such that $J(-\omega) = -J(\omega)$, one can prove that:
\begin{eqnarray}
    \gamma(\Omega) = \frac{2 \pi J(\Omega)}{1-e^{-\beta \Omega}}.
\end{eqnarray}
In accordance, for a spin-boson model \eqref{spin_boson_hamiltonian}, all of the corrections can be calculated as the integral over spectral density, namely
\begin{equation} \label{integral_representation_with_spectral}
    \Upsilon^{(\mathrm{cor})}_{\alpha \beta}(\omega, \omega') = \mathcal{P} \int_{-\infty}^{+\infty} d \Omega \ \frac{D_\mathrm{cor}(\omega,\omega', \Omega)}{1-e^{-\beta \Omega}}  \ J(\Omega)
\end{equation}
(cf. the general formula given by Eq. \eqref{integral_representation}). 

\subsection{Reaction coordinate}
In this section, we further analyze the spin-boson model by using the \textit{reaction coordinate} (rc) method \cite{lambert,Nazir2018}. In particular, via the $rc$ mapping we are able to extract the second-order correction to the mean-force Hamiltonian. For this reason, we can numerically verify our analytical formula with predictions that comes from $rc$. Additionally, since the ``extraction'' of the second-order into a single $rc$ mode depends highly on the form of the spectral density, we will also discuss a validity regime of the $rc$ method. 

Let us concentrate on spectral density given by the form \cite{Nazir2018}:
\begin{eqnarray} \label{spectral_density_underdumped}
    J(\Omega) =  \lambda^2 \frac{4 \gamma \Omega \omega_{rc}^2}{(\omega_{rc}^2 - \Omega^2)^2 + (2 \pi \gamma \Omega \omega_{rc})^2}. 
\end{eqnarray}
The parameter $\gamma$ regulates the width of the spectral density $J(\omega)$, which is centered at the frequency $\omega_{rc}$. This indicates that in the case of narrowly-peaked spectral densities, i.e., when $\gamma \ll 1$, the mode with frequency $\omega_{rc}$ predominates in the environment. Consequently,  that specific mode, i.e., the so-called `reaction coordinate' with bosonic operators $b, b^\dag$, is incorporated into the effective Hamiltonian (after the Bogoliubov transformation):
\begin{eqnarray}
    H' = -\frac{\omega_0}{2} \sigma_z + \omega_{rc} b^\dag b +  \lambda \vec r \cdot \vec \sigma (b + b^\dag),
\end{eqnarray}
that describes the composite `spin-reaction coordinate' system \cite{Nazir2018,lambert,latune}. After transformation the system is coupled to the effective environment with Ohmic spectral density $J_{rc}(\Omega) = \gamma \Omega e^{-\frac{\Omega}{\Lambda}}$ (in the limit $\Lambda \to \infty$). In other words, after the reaction coordinate  transformation, the system described by the Hamiltonian $H'$ is coupled to the environment via constant $\gamma$. Indeed, by decreasing $\gamma$, the reaction coordinate dominates more and more, which results in smaller and smaller effective coupling with the rest of the modes.  

Then, according to the definition \eqref{mean_force_equation_main}, the mean-force state is given by:
\begin{eqnarray} \label{mean_force_state_RC}
    \rho_{\text{mf}} \propto \Tr_R[e^{-\beta H}] \propto \Tr_{rc}[e^{-\beta H'}] + O(\gamma).
\end{eqnarray}
Here we have traced over
"reaction coordinate" system, called $rc$, represented by the mode $b$. 
Consequently, for sufficiently small $\gamma$, one can approximate the mean-force state via the effective Hamiltonian $H'$.  

The resulting 
state depends on all orders of $\lambda$, and we want to extract just second order correction. 
In general, for a state of the form  
\begin{equation}
    \rho_{\lambda} = \frac{e^{-\beta H}}{\mathcal{Z}}
\end{equation}
where 
\begin{align}
    & H = H_0 + \lambda^2 H^{(2)} + \lambda^4 H^{(4)}+ \dots 
\end{align}
we can extract second order correction
via the formula: 
\begin{equation} \label{second_order_correction_from_mean_force_state}
     H^{(2)} = \lim_{\lambda \to 0} \frac{1}{\lambda^{2}} \Big[ \frac{1}{\beta} \Big( \frac{1}{d} \Tr[\log(\rho_{\lambda})] -\log(\rho_{\lambda})\Big) -  H_{0} \Big],
\end{equation}
where $d$ is the dimension of the system Hilbert space and we have used gauge $\Tr[H] = 0$.

\subsection{Numerical simulation}
We numerically computed the relevant coefficients for the general spin-boson model given by Hamiltonian \eqref{spin_boson_hamiltonian}:
\begin{eqnarray}
    \Upsilon_\mathrm{cor}(\omega_0,0)-\Upsilon_\mathrm{cor}(0, -\omega_0) \quad \text{(off-diagonal)}, \\
    \Upsilon_\mathrm{cor}(-\omega_0, -\omega_0)-\Upsilon_\mathrm{cor}(\omega_0, \omega_0) \quad \text{(diagonal)}. 
\end{eqnarray}
Corrections have been computed according to expression \eqref{integral_representation_with_spectral} for the spectral density given by Eq. \eqref{spectral_density_underdumped} with the specific form of the kernels summarized in Table \ref{tab:qubit_table}. 

Additionally, the mean-fore correction has been computed independently based on reaction coordinate method, expressed in Eq. \eqref{mean_force_state_RC} and \eqref{second_order_correction_from_mean_force_state}. Since the mean-force state, calculated from the effective Hamiltonian $H'$, is approximated up to terms of the order $O(\gamma)$, we provide a simulation for two regimes: $\gamma = 0.001$ and $\gamma = 0.01$. According to Eq. \eqref{mean_force_state_RC}, the reaction coordinate Gibbs state should converge to the exact mean-force correction when $\gamma \to 0$, which is confirmed by the presented simulation in Fig. \ref{fig:corrections}.

\section{Conclusions}
We have presented several results of independent significance. Firstly, we derived the general formula for the (second-order) mean-force Hamiltonian, along with the quasi-steady state Hamiltonian for the canonical class of Liouvillians (i.e., with the leading order expressed in GKLS form for a general Kossakowski matrix). Subsequently, we discussed the relationships among different corrections (mean-force, quasi-steady state, and Lamb-shift). Finally, we applied our findings to the most well-known descriptions of open systems, specifically examining the extent to which the Bloch-Redfield, Davies, and cumulant equations satisfy the condition of converging to equilibrium. We emphasize that the corrections to Hamiltonians derived in this paper hold greater significance than corrections to states because they can be independently applied to various renormalization procedures (cf. \cite{WinczewskiAlicki}).

Our findings align with prior observations \cite{Trushechkin_2021_MeanForce,miyashita,fleming}, affirming that the leading-order term of the Liouvillian should adopt the Bloch-Redfield form to yield the correct off-diagonal quasi-steady state correction. The cumulant 
equation (if written in the form of master equation) meets this condition but it goes further by providing non-trivial higher-order generators (by definition absent in the Bloch-Redfield equation). Our analysis demonstrates that this feature offers an accurate approximation of the diagonal elements of the correction, coinciding pretty well with those of the mean-force, in contrast to the Bloch-Redfield (or Davies) master equation.

This may appear contradictory to the findings in \cite{Rivas_2017}, where it is asserted that the dynamics governed by the cumulant equation converges to the dynamics governed by the Davies equation in the long time limit, implying that the stationary state should be the Gibbs state of the bare Hamiltonian. However, our result specifically addresses the second-order correction, and therefore, mathematically, it does not contradict the aforementioned conclusion. It is noteworthy that by collecting of all orders in the cumulant Liouvillian, it ensures the positivity of the dynamics. This suggests a nuanced trade-off between positivity and thermalization towards the mean-force Gibbs state in the cumulant description. Truncating the Liouvillian to a specific order approximates the mean-force state but compromises positivity, and vice versa.

\section*{Acknowledgement}

We thank Antonio Mandarino for drawing our attention to Ref. \cite{Guarnieri}. We also acknowledge discussions with participants of Open Systems Seminars at ICTQT. This work is supported by Foundation for Polish Science (FNP), IRAP project ICTQT, contract no.
2018/MAB/5, co-financed by EU Smart Growth Operational Programme. 
MH is also partially supported by Polish National Science Centre grant OPUS-21 (grant No: 2021/41/B/ST2/03207).
MW acknowledges grant PRELUDIUM-20 (grant number: 2021/41/N/ST2/01349) from the
National Science Center.

\bibliography{bib}

\appendix
\newpage

\onecolumngrid

\section{Preliminaries}
We consider the system and bath Hamiltonian:
\begin{eqnarray}
H = H_0 + H_R + \lambda H_I
\end{eqnarray}
such that $H_0$ and $H_R$ are the free Hamiltonians of the system and the bath, respectively, and $H_I$ is the interaction Hamiltonian. Except the section \ref{mean_force_section_dyson_series} and \ref{mean_force_section_weak_coupling}, where we do not assume any particular form of the interaction term, throughout the paper we consider the following explicit form:
\begin{eqnarray} \label{interaction_hamiltonian}
H_I = \sum_\alpha A_\alpha \otimes R_\alpha. 
\end{eqnarray}
We introduce the time-dependent operators:
\begin{eqnarray}
A_\alpha(t) =  e^{i H_0 t} A_\alpha e^{- i H_0 t}, \ \ R_\alpha(t) =  e^{i H_R t} R_\alpha e^{- i H_R t}.
\end{eqnarray}
and jump operators (acting on the system Hilbert space):
\begin{eqnarray} \label{jump_operators_definition_appendix}
A_\alpha (\omega)  = \sum_{\epsilon' - \epsilon = \omega} \Pi(\epsilon) A_\alpha \Pi(\epsilon')
\end{eqnarray}
where $\Pi(\epsilon)$ is a projector on subspace with energy $\epsilon$, such that $H_0 = \sum_\epsilon \epsilon \ \Pi(\epsilon)$. These obey the following commutation relation:
\begin{eqnarray} \label{jump_operators_commutator}
[A_\alpha (\omega), H_0] &=&  \omega A_\alpha(\omega).
\end{eqnarray}
as well as the relations:
\begin{eqnarray} \label{jump_operators_dag_and_sum}
A_\alpha^\dag(\omega) = A_\alpha(-\omega), \ \ \sum_\omega A_\alpha (\omega) = A_\alpha.
\end{eqnarray} 
From this follows also 
\begin{eqnarray} \label{jump_operator_exponent_commutation}
A_\alpha(\omega) e^{c H_0} = e^{c \omega} e^{c H_0} A_\alpha(\omega)
\end{eqnarray}
where $c$ is the complex number, such that, in particular, the time-dependent operator is given by:
\begin{eqnarray} \label{jump_operators_time}
A_\alpha(t) = \sum_\omega e^{-i\omega t} A_\alpha(\omega).
\end{eqnarray}

Finally, we consider the Bloch-Redfield master equation in the Schr\"{o}dinger picture:
\begin{eqnarray} \label{redfield_eq_Schrodinger}
\frac{d}{dt} \rho(t) = \mathcal{L}_t [\rho(t)] &=&  i [ \rho(t), H_0 + \sum_{\omega,\omega'} \sum_{\alpha, \beta} \mathcal{S}_{\alpha \beta}(\omega,\omega',t) A^\dag_\alpha (\omega) A_\beta (\omega')]\\
 &+&  \sum_{\omega,\omega'} \sum_{\alpha, \beta}  \gamma_{\alpha \beta}(\omega, \omega', t) \left( A_\beta (\omega') \rho(t) A^\dag_\alpha(\omega) -\frac{1}{2} \{A^\dag_\alpha (\omega) A_\beta (\omega'), \rho(t) \} \right) 
\end{eqnarray}
where $\rho$ is the system density matrix and $\mathcal{L}_t$ is the generator, with 
\begin{eqnarray} 
\gamma_{\alpha \beta}(\omega, \omega', t) &=& \Gamma_{\alpha \beta}(\omega', t) + \Gamma^*_{\beta \alpha}(\omega, t), \label{gamma_omega_omega}\\
\mathcal{S}_{\alpha \beta}(\omega, \omega', t) &=&  \frac{1}{2i} \left[\Gamma_{\alpha \beta}(\omega', t) - \Gamma^*_{\beta \alpha}(\omega, t) \right], \label{S_omega_omega}\\
\Gamma_{\alpha \beta}(\omega, t) &=& \int_0^{t} ds \ e^{i \omega s } \langle R_\alpha (s) R_\beta (0) \rangle_{\gamma_R}. \label{capital_gamma}
\end{eqnarray}
In the interaction picture (with respect to the bare Hamiltonian $H_0$), the Bloch-Redfield equation takes the following form:
\begin{eqnarray} \label{redfield_eq_interaction}
\frac{d}{dt} \tilde \rho(t) = \tilde{\mathcal{L}}_t [\tilde \rho(t)]  &=&  i [\tilde \rho(t), \sum_{\omega,\omega'} \sum_{\alpha, \beta} \tilde{\mathcal{S}}_{\alpha \beta}(\omega,\omega',t) A^\dag_\alpha (\omega) A_\beta (\omega')]\\
 &+&  \sum_{\omega,\omega'} \sum_{\alpha, \beta}  \tilde \gamma_{\alpha \beta}(\omega, \omega', t) \left( A_\beta (\omega') \tilde \rho(t) A^\dag_\alpha(\omega) -\frac{1}{2} \{A^\dag_\alpha (\omega) A_\beta (\omega'), \tilde \rho(t) \} \right) 
\end{eqnarray}
where $\tilde \rho(t) = e^{i H_0 t} \rho(t) e^{- i H_0 t}$, and 
\begin{eqnarray} 
 \tilde \gamma_{\alpha \beta}(\omega, \omega', t) &=& e^{i(\omega-\omega')t} \gamma_{\alpha \beta}(\omega, \omega', t), \label{gamma_interaction} \\
  \tilde{\mathcal{S}}_{\alpha \beta}(\omega, \omega', t) &=&  e^{i(\omega-\omega')t} \mathcal{S}_{\alpha \beta}(\omega, \omega', t). \label{S_interaction}
\end{eqnarray}
 Additionally, we also use the abbreviation
 \begin{eqnarray}
 \mathcal{S}_{\alpha \beta}(\omega) = \lim_{t\to\infty} \mathcal{S}_{\alpha \beta}(\omega, \omega, t), \ \ \gamma_{\alpha \beta}(\omega) = \lim_{t\to\infty} \gamma_{\alpha \beta}(\omega, \omega, t)
 \end{eqnarray}
 From the definition it is seen that $\gamma_{\alpha \beta}(\omega)$ is the Fourier transform of the autocorrelation function, i.e., 
 \begin{eqnarray} \label{gamma_fourier_transform}
 \gamma_{\alpha \beta}(\omega) &=& \Gamma_{\alpha \beta}(\omega) + \Gamma^*_{\beta \alpha}(\omega) \\
 &=& \int_0^{\infty} ds \ e^{i \omega s } \langle R_\alpha (s) R_\beta (0) \rangle_{\gamma_R} + \int_0^{\infty} ds \ e^{-i \omega s } \langle R_\beta (s) R_\alpha (0) \rangle^*_{\gamma_R} \\
 &=& \int_0^{\infty} ds \ e^{i \omega s } \langle R_\alpha (s) R_\beta (0) \rangle_{\gamma_R} + \int_0^{\infty} ds \ e^{-i \omega s } \langle R_\alpha (-s) R_\beta (0) \rangle_{\gamma_R} \\
 &=& \int_{-\infty}^{+\infty} ds \ e^{i \omega s } \langle R_\alpha (s) R_\beta (0) \rangle_{\gamma_R} 
 \end{eqnarray}
from which it follows that $\gamma_{\alpha \beta}(\omega)$ obeys the detailed balance condition, i.e.,
  \begin{eqnarray} \label{detailed balance_condition}
 \gamma_{\alpha \beta}(\omega) = \gamma_{\beta \alpha}(-\omega) e^{\beta \omega}.  
 \end{eqnarray}
Using the inverse Fourier transform for the autocorrelation function and the Sokhostki-Plemelj identity in the form:
\begin{eqnarray} \label{sokhostki_plemelj_identity}
\int_0^{\infty} ds \ e^{\pm i \omega s} = \frac{1}{\pi} \delta(\omega) \pm i\mathcal{P} \frac{1}{\omega},
\end{eqnarray}
we shall represent $\Gamma_{\alpha \beta}(\omega)$ as the principal value integral:
\begin{eqnarray}
\Gamma_{\alpha \beta}(\omega)  &=& \int_0^{\infty} ds \ e^{i \omega s } \langle R_\alpha (s) R_\beta (0) \rangle_{\gamma_R} = \frac{1}{2\pi} \int_{-\infty}^{+\infty} d \Omega \  \gamma_{\alpha \beta}(\Omega) \int_0^{\infty} ds \ e^{i (\omega - \Omega )s} \\
&=& \frac{1}{2} \gamma_{\alpha \beta}(\omega) +  \mathcal{P} \frac{1}{2 \pi} \int_{-\infty}^{+\infty} d \Omega \  \frac{i \gamma_{\alpha \beta}(\Omega)}{\omega - \Omega} \label{correlationfunc}
\end{eqnarray}
According to this, and since the $\gamma^*_{\alpha \beta}(\omega) = \gamma_{\beta \alpha}(\omega)$, we have 
 \begin{eqnarray} \label{S_principal_value}
 \mathcal{S}_{\alpha \beta}(\omega) &=& \frac{1}{2i} \left[\Gamma_{\alpha \beta}(\omega) - \Gamma^*_{\beta \alpha}(\omega) \right] = \mathcal{P} \frac{1}{2 \pi} \int_{-\infty}^{+\infty} d \Omega \  \frac{\gamma_{\alpha \beta}(\Omega)}{\omega - \Omega } 
 \end{eqnarray}

\section{Mean-force Hamiltonian} \label{mean_force_section}
We search for the solution for the mean-force Hamiltonian $H_\text{mf}$ from the equation:
\begin{eqnarray} \label{mean_force_equation}
e^{-\beta H_{\text{mf}}} = \frac{1}{\mathcal{Z}_R} \Tr_R[e^{-\beta H}],
\end{eqnarray}
where $H_{\text{mf}} = H_0 + \delta H_{\text{mf}}$ and $\mathcal{Z}_R = \Tr[e^{-\beta H_R}]$. Notice, that we used here the following gauge, i.e.,
\begin{eqnarray}
\Tr[e^{-\beta H}] = \Tr[e^{-\beta H_{\text{mf}}}] \Tr[e^{-\beta H_R}]
\end{eqnarray}
that fixes the ground state energy of the mean-force Hamiltonian. 

\subsection{Dyson series} \label{mean_force_section_dyson_series}
In the following section, we use an abbreviation $\hat A (t) = e^{t (H_0+H_R)} A e^{-t (H_0+H_R)}$. We start with the LHS of the Eq. \eqref{mean_force_equation}, which we represent by the formal Dyson form:
\begin{equation}
    e^{-\beta H_{\text{mf}}} = e^{-\beta H_0} e^{\beta H_0} e^{-\beta H_\text{mf}} =  e^{-\beta H_0} \mathcal{T}  e^{- \int_0^\beta dt \ \delta \hat{H}_\text{mf} (t)},
\end{equation}
which gives us the series expansion:
\begin{equation}
    \mathcal{T}  e^{- \int_0^\beta dt \ \delta \hat H_{\text{mf}} (t)} = \mathbb{1} - \int_0^\beta dt_1 \ \delta \hat H_{\text{mf}} (t_1)  + \int_0^\beta dt_1 \int_0^{t_1} dt_2 \ \delta \hat H_{\text{mf}} (t_1) \ \delta \hat H_{\text{mf}} (t_2) + \dots
\end{equation}
Similarly, for the RHS, we have
\begin{equation}
\label{eq:general-mean-force}
    e^{-\beta H} = e^{-\beta H_0} \mathcal{T} e^{- \lambda \int_0^\beta dt \ \hat H_I(t)},
\end{equation}
such that
\begin{equation}
\mathcal{T}  e^{- \lambda \int_0^\beta dt \hat H_I(t)} = \mathbb{1} - \lambda \int_0^\beta dt_1 \ \hat H_I(t_1)  + \lambda^2 \int_0^\beta dt_1 \int_0^{t_1} dt_2 \ \hat H_I (t_1) \ \hat H_I (t_2) + \dots
\end{equation}
Finally, one can write
\begin{equation}
     \frac{1}{\mathcal{Z}_R} \Tr_R[e^{-\beta H}] = \frac{1}{\mathcal{Z}_R} \Tr_R[e^{-\beta H_0} \mathcal{T} e^{- \lambda \int_0^\beta dt \ \hat H_I(t)}] =   e^{-\beta H_0} \Tr_R[\mathcal{T} e^{- \lambda \int_0^\beta dt \ \hat H_I(t)} \gamma_R] 
\end{equation}
where $\gamma_R = \frac{e^{-\beta \hat H_R}}{\mathcal{Z}_R}$ is the Gibbs state of the bath. In analogy, we have:
\begin{equation}
    e^{-\beta H_{\text{mf}}} = e^{-\beta H_0} \Tr_R [\mathcal{T}  e^{- \int_0^\beta dt \ \delta \hat H_{\text{mf}} (t)} \gamma_R] 
\end{equation}
such that Eq. \eqref{mean_force_equation} can be written as:
\begin{equation}
    \Tr_R \left[ \mathcal{T} \left( e^{- \int_0^\beta dt \ \delta \hat H_{\text{mf}} (t)} -  e^{- \lambda \int_0^\beta dt \ \hat H_I(t)} \right) \gamma_R \right] = 0 
\end{equation}
or in the series form as:
\begin{equation} \label{series_form}
    \sum_{k=1}^\infty (-1)^k \int_0^\beta dt_1 \int_0^{t_1} dt_2 \dots \int_0^{t_{k-1}} dt_k \left(\delta \hat H_\text{mf} (t_1) \ \delta \hat H_\text{mf} (t_2) \dots \delta \hat H_\text{mf} (t_k) - \lambda^k \left \langle \hat H_I(t_1) \ \hat H_I(t_2) \dots \hat H_I(t_k) \right \rangle_{\gamma_R} \right) = 0 
\end{equation}
where $\langle \cdot \rangle_{\gamma_R} = \Tr_R[ \ \cdot \ \gamma_R]$.

\subsection{Weak coupling} \label{mean_force_section_weak_coupling}

Now, let us assume that $\lambda \ll 1$, and we expand:
\begin{equation}
    \hat H_\text{mf} = \hat H_0 + \lambda \hat H_\text{mf}^{(1)} + \lambda^2 \hat H_\text{mf}^{(2)} + \dots 
\end{equation}
such that $\delta \hat H_\text{mf} = \lambda  \hat H_\text{mf}^{(1)} + \lambda^2 \hat H_\text{mf}^{(2)} + \dots$

Then, we collect terms in the same order of $\lambda$ appearing in Eq. \eqref{series_form}, i.e.,
\begin{eqnarray}
\lambda&:& \int_0^\beta dt_1 \hat H_\text{mf}^{(1)}(t_1) = \int_0^\beta dt_1 \left \langle \hat H_I(t_1) \right \rangle_{\gamma_R} \\
\lambda^2&:& -\int_0^\beta dt_1 \hat H_\text{mf}^{(2)}(t_1) + \int_0^\beta dt_1 \int_0^{t_1} dt_2 \ \hat H_\text{mf}^{(1)}(t_1) \ \hat H_\text{mf}^{(1)}(t_2) = \int_0^\beta dt_1 \int_0^{t_1} dt_2 \left \langle \hat H_I(t_1) \ \hat H_I(t_2) \right \rangle_{\gamma_R} \\
\lambda^3&:& -\int_0^\beta dt_1 \hat H^{(3)}_\text{mf}(t_1) + \int_0^\beta dt_1  \int_0^{t_1} dt_2 \left(\hat H_\text{mf}^{(1)}(t_1) \hat H_\text{mf}^{(2)}(t_2) + \hat H_\text{mf}^{(2)}(t_1) \hat H_\text{mf}^{(1)}(t_1) \right) \nonumber \\ 
&-& \int_0^\beta dt_1 \int_0^{t_1} dt_2 \int_0^{t_2} dt_3 \ \hat H_\text{mf}^{(1)}(t_1) \ \hat H_\text{mf}^{(1)}(t_2) \ \hat H_\text{mf}^{(1)}(t_3) = \int_0^\beta dt_1 \int_0^{t_1} dt_2 \int_0^{t_2} dt_3 \left \langle \hat H_I(t_1) \ \hat H_I(t_2) \ \hat H_I(t_3) \right \rangle_{\gamma_R}  \\
&\dots& \nonumber
\end{eqnarray}
In general for the $n$-th order we have 
\begin{multline}
 \sum_{m=1}^n (-1)^m \sum_{k_1, \dots, k_m \in \mathcal{C}_n^m} \int_0^\beta dt_1 \dots \int_0^{t_{m-1}} dt_m \ \hat H_S^{(k_1)}(t_1) \ \hat H_S^{(k_2)}(t_2) \dots \hat H_S^{(k_m)}(t_m) \\ = (-1)^n \int_0^\beta dt_1 \dots \int_0^{t_{n-1}} dt_n \left \langle \hat H_I(t_1) \dots \hat H_I(t_n) \right \rangle_{\gamma_R} \nonumber
\end{multline}
where $\mathcal{C}_n^k$ is set of the $k$-th order composition of the number $n$, e.g., $\mathcal{C}_4^3 = \{(2,1,1), (1,2,1), (1,1,2)\}$.  

\subsection{Derivation of general formulas for corrections to the mean-force Hamiltonian}
\label{mean_force_corrections_appendix}
\subsubsection{First-order correction}
Let us first solve the equation for the first-order correction, i.e.,
\begin{eqnarray}
\int_0^\beta dt_1 \hat H_\text{mf}^{(1)}(t_1) = \int_0^\beta dt_1 \left \langle \hat H_I(t_1) \right \rangle_{\gamma_R} = \sum_\alpha \int_0^\beta dt_1 \hat A_\alpha (t_1) \left \langle \hat R_\alpha (t_1) \right \rangle_{\gamma_R}.
\end{eqnarray}
Furthermore, since $\langle \hat R_\alpha (t_1) \rangle_{\gamma_R} = \langle R_\alpha \rangle_{\gamma_R}$ (due to commutation of the Gibbs state $\gamma_R$ with the free Hamiltonian $H_R$), we get the solution:
\begin{eqnarray}
H_\text{mf}^{(1)} = \sum_\alpha  \langle R_\alpha \rangle_{\gamma_R} \  A_\alpha. 
\end{eqnarray}
From now on, we assume that bath operators are centralized such that $\langle  R_\alpha \rangle_{\gamma_R} = 0$, which implies $H_\text{mf}^{(1)} = 0$.

\subsubsection{Second-order correction}

In this section we provide general formula for second order correction for mean-force Hamiltonian. 
Remarkably the expression does not exhibit any poles, in contrast to  Lamb-shift correction. Yet we also decompose it into 
bricks that are used also to build the  Lamb-shift corrections, which do exhibit poles, and require principal value to be well defined. 

\begin{theorem}
The explicit form of second order correction for mean-force Hamiltonian is the following:

\begin{eqnarray} \label{mean_force_jump_representation}
    H_\text{mf}^{(2)} = \sum_{\omega, \omega'} \Upsilon^{(\text{mf})}_{\alpha \beta} (\omega, \omega') A_\alpha^\dag(\omega) A_\beta(\omega'),
\end{eqnarray}
where
\begin{eqnarray}
\label{eq:gamma-mean-force}
    && \Upsilon^{(\text{mf})}_{\alpha \beta}(\omega, \omega') = \frac{1}{2 \pi}\int_{-\infty}^{+\infty}  d \Omega \  D_\mathrm{mf}(\omega, \omega', \Omega) \ \gamma_{\alpha \beta}(\Omega), \ \  D_\mathrm{mf}(\omega, \omega', \Omega) = \frac{1}{\omega'-\Omega} - \frac{(\omega-\omega')(e^{\beta(\omega-\Omega)}-1)}{(\omega-\Omega)(\omega'-\Omega)(e^{\beta(\omega-\omega')}-1)}
\end{eqnarray}
The coefficients $\Upsilon_{\alpha \beta}(\omega,\omega')$
can be also expressed in terms of the imaginary part of $\Gamma_{\alpha \beta}(\omega)$ (see Eq. \eqref{S_principal_value})  as follows  
\begin{align}
\label{eq:gamma-Gamma}
\Upsilon^{(\text{mf})}_{\alpha \beta}(\omega,\omega')=
     \frac{1}{e^{\beta \omega}-e^{\beta \omega'}} \left(e^{\beta \omega}\mathcal{S}_{\alpha \beta}(\omega') - e^{\beta \omega'} \mathcal{S}_{\alpha \beta}(\omega) + e^{\beta (\omega+\omega')} \left( \mathcal{S}_{\beta \alpha}(-\omega') - \mathcal{S}_{\beta \alpha}(-\omega) \right)\right).
\end{align}
\end{theorem}

\begin{remark}
From \eqref{eq:gamma-Gamma} one sees that  
$\Upsilon^{(\text{mf})}_{\alpha \beta}$ is symmetric, i.e., $\Upsilon^{(\text{mf})}_{\alpha \beta}(\omega, \omega') = \Upsilon^{(\text{mf})}_{\alpha \beta}(\omega', \omega)$. 
This can be also seen by writing $D_\mathrm{mf}(\omega,\omega',\Omega)$ in explicitly symmetric form
\begin{eqnarray}
    D_\mathrm{mf}(\omega, \omega',\Omega) = \frac{1}{2}  \frac{(e^{\beta \omega}-e^{\beta \omega'})(\omega + \omega'-2\Omega) + (\omega-\omega')(e^{\beta \omega} + e^{\beta \omega'} - 2 e^{\beta (\omega + \omega'-\Omega)})}{(e^{\beta \omega}-e^{\beta \omega'})(\omega-\Omega)(\omega'-\Omega)}
\end{eqnarray}
\end{remark}

\begin{proof}
We shall first prove that Eq. \eqref{eq:gamma-Gamma}
comes from \eqref{eq:gamma-mean-force}. 
We shall use Eq. \eqref{S_principal_value}, i.e.,  
\begin{align} \label{S_principle_value_1}
  \mathcal{S}_{\alpha \beta}(\omega) = \frac{1}{2\pi}\int_{-\infty}^{+\infty} d\Omega \frac{\gamma_{\alpha \beta}(\Omega)}{\omega-\Omega},
 \end{align}
from which we also derive
\begin{eqnarray}
\label{eq:S-minus-S}
-\mathcal{S}_{\beta \alpha}(-\omega) = \frac{1}{2\pi}\int_{-\infty}^{+\infty} d\Omega \frac{\gamma_{\beta \alpha}(\Omega)}{\omega+\Omega} = 
\frac{1}{2\pi}\int_{-\infty}^{+\infty} d\Omega \frac{\gamma_{\beta \alpha}(-\Omega)}{\omega-\Omega} = 
\frac{1}{2 \pi} \int_{-\infty}^{+\infty} d\Omega \frac{\gamma_{\alpha \beta}(\Omega) e^{-\beta\Omega}}{\omega-\Omega},
\end{eqnarray}
where we used the detailed balance condition \eqref{detailed balance_condition}. 
Thus, to express $\Upsilon^{(\text{mf})}_{\alpha \beta}$ in terms of $\mathcal{S}_{\alpha \beta}$ we have to write $D_\mathrm{mf}$ in terms of $1/(\omega- \Omega)$ or  $1/(\omega' - \Omega)$.
Using 
\begin{align}
   \frac{\omega-\omega'}{(\omega'-\Omega)(\omega-\Omega)} =\frac{1}{\omega'-\Omega} - \frac{1}{\omega-\Omega}  
\end{align}
we thus get 
\begin{eqnarray}
    D_\mathrm{mf}(\omega, \omega', \Omega) &=& \frac{1}{e^{\beta(\omega-\omega')}-1} \left(\frac{e^{\beta(\omega-\omega')}-1}{\omega'-\Omega} - e^{\beta(\omega-\Omega)}(\frac{1}{\omega'-\Omega} - \frac{1}{\omega-\Omega})  + \frac{1}{\omega'-\Omega}
    - \frac{1}{\omega-\Omega}\right) \\
    &=& \frac{1}{e^{\beta \omega}-e^{\beta \omega'}} \left(\frac{e^{\beta \omega}}{\omega'-\Omega} - \frac{e^{\beta \omega'}}{\omega-\Omega} - e^{\beta(\omega + \omega' -\Omega)}(\frac{1}{\omega'-\Omega} - \frac{1}{\omega-\Omega})\right)
    \label{int:DmfInt}
\end{eqnarray}
Rearranging it a bit, and using  \eqref{S_principle_value_1} and \eqref{eq:S-minus-S}
we obtain \eqref{eq:gamma-Gamma}.

Let us now prove the expression \eqref{eq:gamma-mean-force}. We start with second-order equation with centralized bath operators, i.e.,
\begin{eqnarray} \label{second_order_equation_mean_force}
 \int_0^\beta dt \hat H_\text{mf}^{(2)}(t)  = -\int_0^\beta dt \int_0^{t} ds \left \langle \hat H_I(t) \hat H_I(s) \right \rangle_{\gamma_R}
\end{eqnarray}
Next, we put the representation \eqref{mean_force_jump_representation} and according to the relation \eqref{jump_operator_exponent_commutation}, we have
\begin{eqnarray}
\hat H_\text{mf}^{(2)} (t) = \sum_{\omega, \omega'} \Upsilon^{(\text{mf})}_{\alpha \beta}(\omega,\omega') e^{t H_0} A_{\alpha}^\dag(\omega) A_{\beta}(\omega') e^{-t H_0} = \sum_{\omega, \omega'} \Upsilon^{(\text{mf})}_{\alpha \beta}(\omega,\omega') e^{t (\omega - \omega')} A_{\alpha}^\dag(\omega) A_{\beta}(\omega').
\end{eqnarray}
Then, the LHS of Eq. \eqref{second_order_equation_mean_force} is equal to:
\begin{eqnarray}
\int_0^\beta dt \ \hat H_\text{mf}^{(2)}(t) =  \sum_{\omega, \omega'} A_{\alpha}^\dag(\omega) A_{\beta}(\omega') \left(\Upsilon^{(\text{mf})}_{\alpha \beta}(\omega,\omega') \int_0^\beta dt \ e^{t (\omega - \omega')} \right)
\end{eqnarray}
whereas the RHS is given by:
\begin{eqnarray} \label{mean_force_RHS}
-\int_0^\beta dt \int_0^{t} ds \left \langle \hat H_I(t) \hat H_I(s) \right \rangle_{\gamma_R} &=& -\int_0^\beta dt \int_0^{t} ds \ \hat A_{\alpha}(t) \hat A_{\beta}(s) \langle \hat R_{\alpha}(t) \ \hat R_{\beta}(s) \rangle \nonumber \\
&=& - \sum_{\omega, \omega'} \int_0^\beta dt \int_0^{t} ds \  e^{t \omega - s \omega'} A_{\alpha}^\dag(\omega) A_{\beta}(\omega') \langle \hat R_{\alpha}(t-s) \ \hat R_{\beta} \rangle \nonumber \\
&=& - \sum_{\omega, \omega'} A_{\alpha}^\dag(\omega) A_{\beta}(\omega') \int_0^\beta dt \ e^{t(\omega-\omega')} \int_0^{t} ds \  e^{s \omega'}  \langle \hat R_{\alpha}(s) \hat R_{\beta} \rangle 
\end{eqnarray}
where in the last line we change a variables $s \to t - s$. Next, according to Eq. \eqref{gamma_fourier_transform}, let us observe that 
\begin{eqnarray}
\langle \hat R_{\alpha}(it) \hat R_{\beta} \rangle = \langle R_{\alpha}(t) R_{\beta} \rangle = \frac{1}{2\pi } \int_{-\infty}^{+\infty} d\Omega \ e^{-i \Omega t}  \ \gamma_{\alpha \beta}(\Omega),
\end{eqnarray}
such that 
\begin{eqnarray}
\langle \hat R_{\alpha}(t) \hat R_{\beta} \rangle = \frac{1}{2\pi } \int_{-\infty}^{+\infty} d\Omega \ e^{-\Omega t}  \ \gamma_{\alpha \beta}(\Omega).
\end{eqnarray}
Finally, the RHS is equal to:
\begin{eqnarray}
-\int_0^\beta dt \int_0^{t} ds \left \langle \hat H_I(t)\hat H_I(s) \right \rangle_{\gamma_R} = - \sum_{\omega, \omega'} A_{\alpha}^\dag(\omega) A_{\beta}(\omega') \left(\frac{1}{2\pi} \int_{-\infty}^{+\infty} d\Omega \ \gamma_{\alpha \beta}(\Omega) \int_0^\beta dt \ e^{t(\omega-\omega')}   \int_0^{t} ds \ e^{s(\omega' - \Omega)} \right) \nonumber \\
\end{eqnarray}
Equating LHS=RHS, we get
\begin{eqnarray}
\sum_{\omega, \omega'} \left[\Upsilon^{(\text{mf})}_{\alpha \beta}(\omega,\omega') \int_0^\beta dt \ e^{t (\omega - \omega')} + \frac{1}{2\pi} \int_{-\infty}^{+\infty} d\Omega \ \gamma_{\alpha \beta}(\Omega) \int_0^\beta dt \ e^{t(\omega-\omega')}   \int_0^{t} ds \ e^{s (\omega' - \Omega)} \right] A_{\alpha}^\dag(\omega) A_{\beta}(\omega')  = 0
\end{eqnarray}
which is solved by
\begin{eqnarray}
\Upsilon^{(\text{mf})}_{\alpha \beta}(\omega,\omega') = -\frac{1}{2\pi} \int_{-\infty}^{+\infty} d\Omega \  \gamma_{\alpha\beta}(\Omega) \ \frac{\int_0^\beta dt \ e^{t(\omega-\omega')}   \int_0^{t} ds \ e^{s (\omega' - \Omega)}}{\int_0^\beta dt \ e^{t (\omega - \omega')}}
\end{eqnarray}
We thus obtain
\begin{eqnarray}
D_\mathrm{mf}(\omega,\omega',\Omega) &=& - \frac{\int_0^\beta dt \int_0^{t} ds \ e^{t(\omega-\omega')} e^{s (\omega' - \Omega)}}{\int_0^\beta dt \ e^{t (\omega - \omega')}}, 
\end{eqnarray}
what is readily integrated (see \eqref{int:DmfInt}).
Then the integrated form of $\Upsilon^{(\text{mf})}_{\alpha \beta}(\omega,\omega')$ is obtained with relation in \eqref{S_principle_value_1}.

%
\end{proof}

\section{Steady-state correction} \label{steady_state_section}

\subsection{General method}

We look for a solution of the equation:
\begin{eqnarray}
\mathcal{L} [\varrho] = 0,
\end{eqnarray}
where $\mathcal{L}$ is the generator of the master equation and $\varrho$ is its stationary state. We expand the generator and steady-state in the series, i.e., 
\begin{eqnarray}
\mathcal{L} [\rho] &=& \mathcal{L}_0 [\rho] + \lambda^2 \mathcal{L}_2 [\rho] + \lambda^4 \mathcal{L}_4 [\rho] + \dots \\
\varrho &=& \varrho_0 + \lambda^2 \varrho_2 + \lambda^4 \varrho_4 + \dots 
\end{eqnarray}
such that we have the following set of equations (for each order in $\lambda$):
\begin{eqnarray}
&&\mathcal{L}_0 [\varrho_0] = 0 \label{zeroth_order_eq}\\
&&\mathcal{L}_0 [\varrho_2] + \mathcal{L}_2 [\varrho_0] = 0 \label{second_order_eq} \\
&& \mathcal{L}_0 [\varrho_4] + \mathcal{L}_2 [\varrho_2] + \mathcal{L}_4 [\varrho_0] = 0  \label{fourth_order_eq} \\
&&\dots.
\end{eqnarray}
Hence we postulate the stationary state (in the Gibbs form):
\begin{eqnarray}
  \varrho = e^{-\beta(H_0 + \lambda^2 H^{(2)}_{\text{st}} + \lambda^4 H^{(4)}_{\text{st}}} + \dots) = \varrho_0 + \lambda^2 \varrho_2 + \lambda^4 \varrho_4 + \dots
\end{eqnarray}
such that
\begin{eqnarray}
\varrho_0 &=& e^{-\beta H_0} \\
\varrho_2 &=& -e^{-\beta H_0} \int_0^\beta dt \ e^{t H_0} H^{(2)}_\text{st} e^{-t H_0} \\
\varrho_4 &=& -e^{-\beta H_0} \int_0^\beta dt \ e^{t H_0} H^{(4)}_\text{st} e^{-t H_0} + e^{-\beta H_0} \int_0^\beta dt_1 \int_0^{t_1} dt_2 \ e^{t_1 H_0} H^{(2)}_\text{st} e^{-t_1 H_0} e^{t_2 H_0} H^{(2)}_\text{st} e^{-t_2 H_0}
\end{eqnarray}

In the following, we use the summation convention, i.e., the repeating indices are summed up. 
We start with representation of the second-order correction in the basis of jump operators: 
\begin{eqnarray}
H^{(2)}_\text{st} = \Upsilon_{\alpha \beta}^{(\text{st})}(\omega,\omega') A^\dag_ \alpha(\omega) A_\beta(\omega').
\end{eqnarray}
Note that contrary to mean force correction, the above form assumes that Bohr spectrum is nondegenerate. Indeed, then pairs of jump operators span linearly all the space of operators of the system.
In accordance, we have the following expression for $\varrho_2$, i.e.,
\begin{eqnarray}
\varrho_2 = -e^{-\beta H_0} \Upsilon^{(\text{st})}_{\alpha \beta}(\omega, \omega') \int_0^\beta dt \ e^{t H_0} A_\alpha^\dag(\omega) A_\beta(\omega') e^{-t H_0} = -\Upsilon^{(\text{st})}_{\alpha \beta}(\omega, \omega') \alpha(\omega'-\omega) e^{-\beta H_0} A_\alpha^\dag(\omega) A_\beta(\omega'),
\end{eqnarray}
where we define:
\begin{eqnarray} \label{alpha_definition}
\alpha(\omega) =  \int_0^\beta dt \ e^{- t \omega} = 
\begin{cases}
    \frac{1-e^{-\beta \omega}}{\omega},&  \omega \neq 0 \\
    \beta,          & \omega=0
\end{cases}
\end{eqnarray}
In general, we are going to transform the operator equations \eqref{second_order_eq} and \eqref{fourth_order_eq} into the algebraic ones. For this we define:
\begin{eqnarray}
\mathcal{L}_k [\varrho_l] = g^{(kl)}_{\alpha \beta}(\omega_1,\omega_2) e^{-\beta H_0} A_\alpha(\omega_1)A_\beta(\omega_2)
\end{eqnarray}
for the second-order (such that $k+l=2$), and 
\begin{eqnarray}
\mathcal{L}_k [\varrho_l] = g^{(kl)}_{\alpha \beta \gamma \delta}(\omega_1,\omega_2,\omega_3,\omega_4) e^{-\beta H_0} A_\alpha(\omega_1)A_\beta(\omega_2)A_\gamma(\omega_3)A_\delta(\omega_4)
\end{eqnarray}
for $k+l=4$. In accordance, for the second-order equation \eqref{second_order_eq}, we have 
\begin{eqnarray}  \label{second_order_algebraic}
\left(g^{(02)}_{\alpha \beta}(\omega_1,\omega_2) + g^{(20)}_{\alpha \beta}(\omega_1,\omega_2) \right) e^{-\beta H_0} A_\alpha(\omega_1)A_\beta(\omega_2) = 0
\end{eqnarray}
whereas for the fourth-order:
\begin{eqnarray} \label{fourth_order_algebraic}
\left(g^{(04)}_{\alpha \beta \gamma \delta}(\omega_1,\omega_2,\omega_3,\omega_4) + g^{(22)}_{\alpha \beta \gamma \delta}(\omega_1,\omega_2,\omega_3,\omega_4) + g^{(40)}_{\alpha \beta \gamma \delta}(\omega_1,\omega_2,\omega_3,\omega_4) \right) e^{-\beta H_0}  A_\alpha(\omega_1)A_\beta(\omega_2)A_\gamma(\omega_3)A_\delta(\omega_4) = 0 \nonumber \\
\end{eqnarray}

In the following, we will also use the commutation relations:
\begin{eqnarray}\label{commutation_relation}
[A_\alpha (\omega), H_0] = \omega A_\alpha (\omega)
\end{eqnarray}
from which we get:
\begin{eqnarray}
  A_\alpha (\omega) e^{-\beta H_0} =  e^{-\beta \omega} e^{-\beta H_0} A_\alpha(\omega).
\end{eqnarray}
The commutation relation \eqref{commutation_relation} can be further generalize for the product of jump operators, i.e.,
\begin{eqnarray} \label{commutation_relation_product}
[A_{\alpha_1} (\omega_1) A_{\alpha_2} (\omega_2) \dots A_{\alpha_2} (\omega_2), H_0] = (\omega_1 + \omega_2 + \dots + \omega_n) A_{\alpha_1} (\omega_1) A_{\alpha_2} (\omega_2) \dots A_{\alpha_2} (\omega_2).
\end{eqnarray}
Notice also that $A^\dag_\alpha(\omega) = A_\alpha(-\omega)$.

\subsection{Second-order equation}\label{appendix:second_order}
In the following, we solve Eq. \eqref{second_order_algebraic} for a master equation of the form:
\begin{eqnarray}
&&\mathcal{L}_0 [\rho] = i [\rho, H_0] \\
&&\mathcal{L}_2 [\rho] = \sum_{\alpha, \beta} \sum_{\omega,\omega'} \left[ \Upsilon_{\alpha \beta}^{(\text{LS})}(\omega,\omega') [\rho, A_\alpha^\dag(\omega) A_\beta(\omega')] + K_{\alpha \beta}(\omega,\omega') \left(A_\beta(\omega') \rho A_\alpha^\dag(\omega) - \frac{1}{2} \{A_\alpha^\dag(\omega) A_\beta(\omega'), \rho \} \right) \right]
\end{eqnarray}
We observe that the zeroth-order equation, i.e., $[\varrho_0, H_0] = 0$ is obviously satisfied for a choice $\varrho_0 = e^{-\beta H_0}$. 

Let us then calculate the coefficients $g_{\alpha\beta}^{(02)}$ and $g_{\alpha\beta}^{(20)}$ for the second-order equation. We start with:
\begin{eqnarray}
\mathcal{L}_0 [\varrho_2] &=& i [\varrho_2, H_0] = -i \Upsilon^{(\text{st})}_{\alpha \beta}(\omega_1, \omega_2) \alpha(\omega_1+\omega_2) [ e^{-\beta H_0} A_\alpha^\dag(\omega) A_\beta(\omega_2), H_0] \\
&=& -i (\omega_1 + \omega_2) \Upsilon^{(\text{st})}_{\alpha \beta}(-\omega_1, \omega_2) \alpha(\omega_1+\omega_2) e^{-\beta H_0} A_\alpha(\omega_1) A_\beta(\omega_2)
\end{eqnarray}
where we used Eq. \eqref{commutation_relation_product}, such that
\begin{eqnarray}
g_{\alpha\beta}^{(02)} = -i (\omega_1+\omega_2) \Upsilon^{(\text{st})}_{\alpha \beta}(-\omega_1, \omega_2) \alpha(\omega_1+\omega_2)
\end{eqnarray}
Next, we shall calculate:
\begin{eqnarray}
\mathcal{L}_2 [\varrho_0] &=& i\Upsilon_{\alpha \beta}^{(\text{LS})}(\omega_1,\omega_2) [\varrho_0, A_\alpha^\dag(\omega_1) A_\beta(\omega_2)] + K_{\alpha \beta}(\omega_1,\omega_2) \left(A_\beta(\omega_2) \varrho_0 A_\alpha^\dag(\omega_1) - \frac{1}{2} \{A_\alpha^\dag(\omega_1) A_\beta(\omega_2), \varrho_0 \} \right) \\
&=& i\Upsilon_{\alpha \beta}^{(\text{LS})}(\omega_1,\omega_2) [ e^{-\beta H_0}, A_\alpha^\dag(\omega_1) A_\beta(\omega_2)] + K_{\alpha \beta}(\omega_1,\omega_2) \left(A_\beta(\omega_2)  e^{-\beta H_0} A_\alpha^\dag(\omega_1) - \frac{1}{2} \{A_\alpha^\dag(\omega_1) A_\beta(\omega_2),  e^{-\beta H_0} \} \right) \nonumber
\end{eqnarray}
First, let us rewrite the Hamiltonian part in the form:
\begin{eqnarray}
i\Upsilon_{\alpha \beta}^{(\text{LS})}(\omega_1,\omega_2) [ e^{-\beta H_0}, A_\alpha^\dag(\omega_1) A_\beta(\omega_2)] &=& i\Upsilon_{\alpha \beta}^{(\text{LS})}(\omega_1,\omega_2) (e^{-\beta H_0} A_\alpha^\dag(\omega_1) A_\beta(\omega_2) - A_\alpha^\dag(\omega_1) A_\beta(\omega_2) e^{-\beta H_0})\\
&=& i\Upsilon_{\alpha \beta}^{(\text{LS})}(\omega_1,\omega_2) (1 - e^{-\beta(\omega_2-\omega_1)}) e^{-\beta H_0} A_\alpha^\dag(\omega_1) A_\beta(\omega_2)
\end{eqnarray}
and then the dissipative part as follows
\begin{align}
K_{\alpha \beta}(\omega_1,\omega_2) & \left(A_\beta(\omega_2)  e^{-\beta H_0} A_\alpha^\dag(\omega_1) - \frac{1}{2} \{A_\alpha^\dag(\omega_1) A_\beta(\omega_2),  e^{-\beta H_0} \} \right)  \\ 
&= K_{\alpha \beta}(\omega_1\omega_2) \left(A_\beta(\omega_2)  e^{-\beta H_0} A_\alpha^\dag(\omega_1) - \frac{1}{2} A_\alpha^\dag(\omega_1) A_\beta(\omega_2) e^{-\beta H_0} - \frac{1}{2} e^{-\beta H_0} A_\alpha^\dag(\omega_1) A_\beta(\omega_2) \} \right)  \\
&= K_{\alpha \beta}(\omega_1\omega_2) \left(e^{-\beta \omega} e^{-\beta H_0} A_\beta(\omega_2) A_\alpha^\dag(\omega_1) - \frac{1}{2} e^{-\beta (\omega -\omega_2)} e^{-\beta H_0} A_\alpha^\dag(\omega_1) A_\beta(\omega_2) - \frac{1}{2} e^{-\beta H_0} A_\alpha^\dag(\omega_1) A_\beta(\omega_2) \right)  \\
&= \left(e^{\beta \omega_1} K_{\beta \alpha}(-\omega_2,-\omega_1) - \frac{1}{2} K_{\alpha \beta}(\omega_1, \omega_2) (e^{-\beta (\omega_2 -\omega_1)} + 1) \right) e^{-\beta H_0} A_\alpha^\dag(\omega_1) A_\beta(\omega_2)
\end{align}
Finally, we get
\begin{eqnarray}
g_{\alpha\beta}^{(20)}(\omega_1,\omega_2) = i\Upsilon_{\alpha \beta}^{(\text{LS})}(-\omega_1,\omega_2) (1 - e^{-\beta(\omega_1+\omega_2)}) + e^{\beta \omega_1} K_{\beta \alpha}(-\omega_2,-\omega_1) - \frac{1}{2} K_{\alpha \beta}(-\omega_1, \omega_2) (e^{-\beta (\omega_1 + \omega_2)} + 1)
\end{eqnarray}
Now, we postulate the solution
\begin{eqnarray}
g_{\alpha\beta}^{(02)}(\omega_1,\omega_2) + g_{\alpha\beta}^{(20)}(\omega_1,\omega_2) = 0 
\end{eqnarray}
for each $\omega_1,\omega_2$ and $\alpha, \beta$. 
First, for $\omega_1=\omega_2 \equiv \omega$, we have
\begin{eqnarray}
e^{\beta \omega} K_{\beta \alpha}(-\omega, -\omega) - K_{\alpha \beta}(\omega, \omega) = 0
\end{eqnarray}
such that the coefficient $K_{\alpha \beta}(\omega, \omega)$ has to satisfy the detailed balance condition.
Furthermore, for $\omega_1 \neq \omega_2$ we get
\begin{multline}
i \Upsilon^{(\text{st})}_{\alpha \beta}(-\omega_1, \omega_2) (e^{-\beta(\omega_1+\omega_2)}-1) - i \Upsilon_{\alpha \beta}^{(\text{LS})}(-\omega_1, \omega_2) ( e^{-\beta(\omega_1+\omega_2)}-1) \\
+ e^{-\beta \omega_1} K_{\beta \alpha}(-\omega_2,\omega_1) - \frac{1}{2} K_{\alpha \beta}(-\omega_1, \omega_2) (e^{\beta (\omega_1 + \omega_2)} + 1) = 0 
\end{multline}
This can be further simplified to:
\begin{eqnarray} \label{mean_force_solution}
\Upsilon^{(\text{st})}_{\alpha \beta}(\omega_1, \omega_2) = \Upsilon_{\alpha \beta}^{(\text{LS})}(\omega_1,\omega_2) + \frac{i}{e^{\beta \omega_1} - e^{\beta \omega_2}} \left(e^{\beta (\omega_1+\omega_2)} K_{\beta \alpha}(-\omega_2,-\omega_1) - \frac{1}{2} K_{\alpha \beta}(\omega_1,\omega_2) (e^{\beta \omega_1} + e^{\beta \omega_2}) \right).
\end{eqnarray}

\subsubsection{Solutions for the Bloch-Redfield master equation and for secular approximation} \label{appendix:steady_state_mean_force_proof}
For the Bloch-Redfield master equation, we have:
\begin{eqnarray}
  \Upsilon_{\alpha \beta}^{(\text{LS})}(\omega,\omega') &=& \frac{1}{2i} (\Gamma_{\alpha \beta}(\omega') - \Gamma_{\beta \alpha}^*(\omega)), \\ K_{\alpha \beta}(\omega,\omega') &=& \Gamma_{\alpha \beta}(\omega') + \Gamma_{\beta \alpha}^*(\omega) \label{upsilon_dynamical}
\end{eqnarray}
Next, we put:
\begin{eqnarray}
\Gamma_{\alpha \beta}(\omega) = \frac{1}{2} \gamma_{\alpha \beta}(\omega) + i \mathcal{S}_{\alpha \beta}(\omega)
\end{eqnarray}
such that
\begin{eqnarray}
\Upsilon_{\alpha \beta}^{(\text{LS})}(\omega,\omega') &=& \frac{1}{2i} (\frac{1}{2} \gamma_{\alpha \beta}(\omega') + i \mathcal{S}_{\alpha \beta}(\omega') - \frac{1}{2} \gamma_{\alpha \beta}(\omega) + i \mathcal{S}_{\alpha \beta}(\omega)) \\
&=& \frac{1}{4i} (\gamma_{\alpha \beta}(\omega') -  \gamma_{\alpha \beta}(\omega)) + \frac{1}{2}( \mathcal{S}_{\alpha \beta}(\omega') + \mathcal{S}_{\alpha \beta}(\omega)) \label{upsilon_dyn}
\end{eqnarray}
and
\begin{eqnarray}
K_{\alpha \beta}(\omega,\omega') &=& \frac{1}{2} \gamma_{\alpha \beta}(\omega') + i \mathcal{S}_{\alpha \beta}(\omega') + \frac{1}{2} \gamma_{\alpha \beta}(\omega) - i \mathcal{S}_{\alpha \beta}(\omega) \\
&=&  \frac{1}{2} (\gamma_{\alpha \beta}(\omega') + \gamma_{\alpha \beta}(\omega)) + i (\mathcal{S}_{\alpha \beta}(\omega') - \mathcal{S}_{\alpha \beta}(\omega))
\end{eqnarray}
Let us put above expression into Eq. \eqref{mean_force_solution} and collect all of the terms with $\mathcal{S}_{\alpha \beta}$:
\begin{align}
&\frac{1}{2}( \mathcal{S}_{\alpha \beta}(\omega') + \mathcal{S}_{\alpha \beta}(\omega)) - \frac{1}{e^{\beta \omega} - e^{\beta \omega'}} \left(e^{\beta (\omega+\omega')} (\mathcal{S}_{\beta \alpha}(-\omega) - \mathcal{S}_{\beta \alpha}(-\omega')) - \frac{1}{2} (\mathcal{S}_{\alpha \beta}(\omega') - \mathcal{S}_{\alpha \beta}(\omega)) (e^{\beta \omega} + e^{\beta \omega'}) \right) \\
&= \frac{1}{e^{\beta \omega} - e^{\beta \omega'}} \left[\frac{e^{\beta \omega} - e^{\beta \omega'}}{2}( \mathcal{S}_{\alpha \beta}(\omega') + \mathcal{S}_{\alpha \beta}(\omega)) - \frac{e^{\beta \omega} + e^{\beta \omega'}}{2} (\mathcal{S}_{\alpha \beta}(\omega) - \mathcal{S}_{\alpha \beta}(\omega')) - e^{\beta (\omega+\omega')} (\mathcal{S}_{\beta \alpha}(-\omega) - \mathcal{S}_{\beta \alpha}(-\omega')) \right] \\
&= \frac{1}{e^{\beta \omega} - e^{\beta \omega'}} \left[ 
e^{\beta \omega} \mathcal{S}_{\alpha \beta}(\omega')  -  e^{\beta \omega'} \mathcal{S}_{\alpha \beta}(\omega) + e^{\beta (\omega+\omega')} (\mathcal{S}_{\beta \alpha}(-\omega') - \mathcal{S}_{\beta \alpha}(-\omega)) \right]
\end{align}
Next, we collect all of the terms with $\gamma_{\alpha \beta}$, i.e.,
\begin{align}
&\frac{1}{4i} (\gamma_{\alpha \beta}(\omega') -  \gamma_{\alpha \beta}(\omega)) + \frac{i}{e^{\beta \omega} - e^{\beta \omega'}} \left(e^{\beta (\omega+\omega')} \frac{1}{2} (\gamma_{\alpha \beta}(-\omega) + \gamma_{\alpha \beta}(-\omega'))- \frac{1}{4} (\gamma_{\alpha \beta}(\omega') + \gamma_{\alpha \beta}(\omega)) (e^{\beta \omega} + e^{\beta \omega'}) \right) \\
&=\frac{i}{4} \left(-\gamma_{\alpha \beta}(\omega') +  \gamma_{\alpha \beta}(\omega) + \frac{1}{e^{\beta \omega} - e^{\beta \omega'}} \left(2 (e^{\beta \omega'} \gamma_{\alpha \beta}(\omega) + e^{\beta \omega} \gamma_{\alpha \beta}(\omega'))- (\gamma_{\alpha \beta}(\omega') + \gamma_{\alpha \beta}(\omega)) (e^{\beta \omega} + e^{\beta \omega'}) \right) \right) \\
&=\frac{i}{4} \left(-\gamma_{\alpha \beta}(\omega') +  \gamma_{\alpha \beta}(\omega) + \frac{1}{e^{\beta \omega} - e^{\beta \omega'}} \left(2 e^{\beta \omega'} \gamma_{\alpha \beta}(\omega) + 2e^{\beta \omega} \gamma_{\alpha \beta}(\omega') - \gamma_{\alpha \beta}(\omega') (e^{\beta \omega} + e^{\beta \omega'}) - \gamma_{\alpha \beta}(\omega) (e^{\beta \omega} + e^{\beta \omega'}) \right) \right) \\
&=\frac{i}{4} \left(-\gamma_{\alpha \beta}(\omega') +  \gamma_{\alpha \beta}(\omega) + \frac{1}{e^{\beta \omega} - e^{\beta \omega'}} \left(e^{\beta \omega'} \gamma_{\alpha \beta}(\omega) + e^{\beta \omega} \gamma_{\alpha \beta}(\omega') - e^{\beta \omega'} \gamma_{\alpha \beta}(\omega') - e^{\beta \omega} \gamma_{\alpha \beta}(\omega) \right) \right) =0.
\end{align}
One sees that only terms $\mathcal{S}_{\alpha \beta}$ survives. Moreover, these are exactly equal to the expression for a mean-force Hamiltonian given by Eq. \eqref{eq:gamma-Gamma}, such that for the Bloch-Redfield or cumulant master equation we have simply:
\begin{eqnarray}
\Upsilon_{\alpha \beta}^{(\text{st})}(\omega,\omega') = \Upsilon_{\alpha \beta}^{(\text{mf})}(\omega,\omega'),
\end{eqnarray}
for $\omega \neq \omega'$.

Let us observe that if we apply the so-called secular approximation before (see \eqref{redfield_liouvillian}) for $\gamma_{\alpha \beta}$ coefficients, i.e.
\begin{eqnarray}
\gamma_{\alpha \beta}(\omega,\omega') \xrightarrow{\text{sec. approx}} \gamma_{\alpha \beta}(\omega) \delta_{\omega,\omega'}, 
\end{eqnarray}
then
\begin{eqnarray}
e^{\beta (\omega+\omega')} \gamma_{\beta \alpha}(-\omega',-\omega) - \frac{1}{2} \gamma_{\alpha \beta}(\omega,\omega') (e^{\beta \omega} + e^{\beta \omega'} = \delta_{\omega,\omega'} \left(e^{2\beta \omega} \gamma_{\beta \alpha}(-\omega) -  e^{\beta \omega} \gamma_{\alpha \beta}(\omega) \right) = 0
\end{eqnarray}
due to the detailed balance condition. Finally, for such master equation, for $\omega \neq \omega'$, we have 
\begin{eqnarray}
\Upsilon_{\alpha \beta}^{(\text{st})}(\omega,\omega') = \Upsilon_{\alpha \beta}^{(\text{LS})}(\omega,\omega').
\end{eqnarray}

\subsection{Fourth-order equation} \label{appendix_fourth_order}
Now, we are going to solve the fourth-order equation \eqref{fourth_order_algebraic}. For simplicity, we assume that the interaction term is given by $H_I = A \otimes R$, such that we drop the indices, i.e.,
\begin{eqnarray}
\left(g_{04}(\omega_1,\omega_2,\omega_3,\omega_4) + g_{22}(\omega_1,\omega_2,\omega_3,\omega_4) + g_{40}(\omega_1,\omega_2,\omega_3,\omega_4) \right) e^{-\beta H_0}  A(\omega_1)A(\omega_2)A(\omega_3)A(\omega_4) = 0,
\end{eqnarray}
where $g_{\alpha \beta \gamma \delta}^{(kl)} \equiv g_{kl}$. According to the Proposition \ref{diagonal_elements_summation_proposition} in the Section \ref{sect:steady_state_diag}, the above equation is satisfied if and only if the following set of equations is satisfied:
\begin{align} \label{final_g22_g44_eq_appendix}
    \sum_{(\omega_1,\omega_2,\omega_2,\omega_4)\in G(\ket k \to \ket k)}
    \Big( g_{22}(\omega_1,\omega_2,\omega_3,\omega_4) + g_{40}(\omega_1,\omega_2,\omega_3,\omega_4)\Big) = 0.
\end{align}
where $G(|k\rangle\to|k\rangle)$ denotes the set of all four-tuples 
\begin{align}
    (\omega_1,\omega_2,\omega_2,\omega_4)=(\epsilon_l-\epsilon_k,\epsilon_m-\epsilon_l,\epsilon_j-\epsilon_m,\epsilon_k-\epsilon_j).
\end{align}
\subsubsection{$g_{22}$ function}
We consider the term:
\begin{align}
    \mathcal{L}_2 [\varrho_2] = \sum_{\omega_{1},\omega_{2},\omega_{3},\omega_{4}} g_{22}(\omega_1,\omega_2,\omega_3,\omega_4) A(\omega_{1}) A(\omega_{2}) A(\omega_{3}) A(\omega_{4}) 
\end{align}
where
\begin{eqnarray}
   \varrho_2&=&\alpha(\omega_3+\omega_4) \Upsilon_{\text{st}}(-\omega_3,\omega_4)  \varrho_0 A(\omega_{3}) A(\omega_{4}) \\
   \mathcal{L}_2 [\rho]&=& - i \Upsilon_\text{LS}(-\omega_1,\omega_2) [A(\omega_{1}) A(\omega_{2}), \rho]+ K(-\omega_1,\omega_2) \left( A(\omega_{2}) \rho A(\omega_{1})-\frac{1}{2}\{A(\omega_{1})A(\omega_{2}), \rho \}\right)
\end{eqnarray}
Then, we have:
\begin{eqnarray}
    \mathcal{L}_2 [\varrho_2] &=&   i \Upsilon_\text{LS}(- \omega_{1},\omega_{2}) \Upsilon_{\text{st}}(- \omega_{3},\omega_{4}) \alpha(\omega_3+\omega_4) A(\omega_{1}) A(\omega_{2}) \varrho_2 A(\omega_{3}) A(\omega_{4})  \\ 
    &-& i \Upsilon_\text{LS}(- \omega_{1},\omega_{2}) \Upsilon_{\text{st}}(- \omega_{3},\omega_{4}) \alpha(\omega_3+\omega_4) \varrho_2 A(\omega_{3}) A(\omega_{4}) A(\omega_{1}) A(\omega_{2}) \\ 
    &+& \frac{1}{2} \left(\Upsilon_{\text{st}}(- \omega_{3},\omega_{4}) \alpha(\omega_3+\omega_4) K(- \omega_{1},\omega_{2} ) A(\omega_{1}) A(\omega_{2}) \varrho_2 A(\omega_{3}) A(\omega_{4}) \right)\\ 
    &-& \Upsilon_{\text{st}}(- \omega_{3},\omega_{4}) \alpha(\omega_3+\omega_4) K(- \omega_{1},\omega_{2} ) A(\omega_{2}) \varrho_2 A(\omega_{3}) A(\omega_{4}) A(\omega_{1})  \\ 
    &+& \frac{1}{2} \left(\Upsilon_{\text{st}}(- \omega_{3},\omega_{4}) \alpha(\omega_3+\omega_4) K(- \omega_{1},\omega_{2} ) \varrho_2 A(\omega_{3}) A(\omega_{4}) A(\omega_{1}) A(\omega_{2}) \right)
\end{eqnarray}
which we may rewrite as:
\begin{eqnarray}
\mathcal{L}_2 [\varrho_2] &=&  \left(- i \Upsilon_\text{LS}(- \omega_{1},\omega_{2}) \Upsilon_{\text{st}}(- \omega_{3},\omega_{4}) \alpha(\omega_3+\omega_4) + \frac{1}{2}\left(\Upsilon_{\text{st}}(- \omega_{3},\omega_{4}) \alpha(\omega_3+\omega_4) K(- \omega_{1},\omega_{2} )\right) \right) e^{-\beta H_0} A(\omega_{3}) A(\omega_{4}) A(\omega_{1}) A(\omega_{2}) \nonumber \\ 
&+& \left(i \Upsilon_\text{LS}(- \omega_{1},\omega_{2})  e^{- \beta( \omega_{2}+\omega_{1})} + \frac{1}{2} \left( K(- \omega_{1},\omega_{2} )  e^{- \beta( \omega_{2}+\omega_{1})}\right) \right)\Upsilon_{\text{st}}(- \omega_{3},\omega_{4}) \alpha(\omega_3+\omega_4) e^{-\beta H_0} A(\omega_{1}) A(\omega_{2}) A(\omega_{3}) A(\omega_{4})  \nonumber \\ 
&-& \Upsilon_{\text{st}}(- \omega_{3},\omega_{4}) \alpha(\omega_3+\omega_4) K(- \omega_{1},\omega_{2} ) e^{- \beta \omega_{2}} e^{-\beta H_0} A(\omega_{2}) A(\omega_{3}) A(\omega_{4}) A(\omega_{1}).
\end{eqnarray}
Since all $\omega_i$'s are mute indices, we change them such that one obtains:
\begin{multline} \label{g_22_appendix}
g_{22}(\omega_1,\omega_2,\omega_3,\omega_4) = \Upsilon_{\text{st}}(-\omega_3,\omega_4) \alpha(\omega_3+\omega_4) \left(i \Upsilon_{\text{LS}}(-\omega_1,\omega_2) + \frac{1}{2} K(-\omega_1,\omega_2) \right) \\ - \Upsilon_{\text{st}}(-\omega_1,\omega_2) \alpha(\omega_1+\omega_2)\left(i \Upsilon_{\text{LS}}(-\omega_3,\omega_4) - \frac{1}{2} K(-\omega_3,\omega_4) \right)- e^{-\beta \omega_1} \Upsilon_{\text{st}}(-\omega_2,\omega_3) \alpha(\omega_2+\omega_3) K(-\omega_4,\omega_1).
\end{multline}

\subsubsection{$g_{40}$ function (cumulant equation)}
We consider the fourth-order generator of the cumulant in the Schr\"odinger picture:
\begin{eqnarray}
 \mathcal{L}^{(4)}_t [\rho]  = \frac{1}{2} \int_0^t ds \ e^{-iH_0t}\left[\tilde{\mathcal{L}}_s^R , \tilde{\mathcal{L}}_t^R \right] [e^{iH_0t}\rho e^{-iH_0t}] e^{iH_0t}.
\end{eqnarray}
Acting on $\varrho_0$ that commutes with $H_0$, this simplifies to:
\begin{eqnarray}
 \mathcal{L}^{(4)}_t [\varrho_0]  = \frac{1}{2} \int_0^t ds \ e^{-iH_0t}\left[\tilde{\mathcal{L}}_s^R , \tilde{\mathcal{L}}_t^R \right] [\varrho_0] e^{iH_0t} 
\end{eqnarray}
We then define 
\begin{align}
  \mathcal{L}_t^{(4)} [\varrho_0]  = \sum_{\omega_1,\omega_2,\omega_3,\omega_4} g_{40}(\omega_1,\omega_2,\omega_3,\omega_4, t) e^{-\beta H_0}A(\omega_{1}) A(\omega_{2}) A(\omega_{3}) A(\omega_{4}).
\end{align}
To get an expression for $g_{40}$, we first compute the action of $\tilde{\mathcal{L}}_s^R \tilde{\mathcal{L}}_t^R$ on $\varrho_0$, i.e.,
\begin{eqnarray}
\tilde{\mathcal{L}}_s^R \tilde{\mathcal{L}}_t^R[\varrho_0] = \sum_{\omega_1,\omega_2,\omega_3,\omega_4} f(\omega_1,\omega_2,\omega_3,\omega_4) e^{-\beta H_0} A(\omega_{1}) A(\omega_{2}) A(\omega_{3}) A(\omega_{4}) 
\end{eqnarray}
where
\begin{align}
&f(\omega_1,\omega_2,\omega_3,\omega_4,t,s) = \tilde{\mathcal{S}}(- \omega_{1},\omega_{2},s) \tilde{\mathcal{S}}(- \omega_{3},\omega_{4},t) e^{- \beta (\omega_{1} +\omega_{2})} - \tilde{\mathcal{S}}(- \omega_{1},\omega_{2},s) \tilde{\mathcal{S}}(- \omega_{3},\omega_{4},t) e^{- \beta (\omega_{1} +\omega_{2}+\omega_{3}+\omega_{4})} \\
&+ \frac{i}{2} \tilde{\mathcal{S}}(- \omega_{1},\omega_{2},s) \tilde{\gamma}(- \omega_{3},\omega_{4},t) e^{- \beta (\omega_{1} +\omega_{2})} + \frac{i}{2} \tilde{\mathcal{S}}(- \omega_{1},\omega_{2},s) \tilde{\gamma}(- \omega_{3},\omega_{4},t) e^{- \beta (\omega_{1} +\omega_{2}+\omega_{3}+\omega_{4})} \\ 
&- i \tilde{\mathcal{S}}(- \omega_{1},\omega_{2},s) \tilde{\gamma}{\left(- \omega_{4},\omega_{3},t \right)} e^{- \beta (\omega_{1} +\omega_{2}+\omega_{3})} + \tilde{\mathcal{S}}{\left(- \omega_{1},\omega_{2},t \right)} \tilde{\mathcal{S}}(- \omega_{3},\omega_{4},s) e^{- \beta (\omega_{1} +\omega_{2})} \\ 
&- \tilde{\mathcal{S}}{\left(- \omega_{1},\omega_{2},t \right)} \tilde{\mathcal{S}}(- \omega_{3},\omega_{4},s) + \frac{i}{2} \tilde{\mathcal{S}}{\left(- \omega_{1},\omega_{2},t \right)} \tilde{\gamma}(- \omega_{3},\omega_{4},s) e^{- \beta (\omega_{1} +\omega_{2})} - \frac{i}{2} \tilde{\mathcal{S}}{\left(- \omega_{1},\omega_{2},t \right)} \tilde{\gamma}(- \omega_{3},\omega_{4},s) \\
&- i \tilde{\mathcal{S}}{\left(- \omega_{2},\omega_{3},t \right)} \tilde{\gamma}{\left(- \omega_{4},\omega_{1},s \right)} e^{- \beta (\omega_{1} +\omega_{2}+\omega_{3})} + i \tilde{\mathcal{S}}{\left(- \omega_{2},\omega_{3},t \right)} \tilde{\gamma}{\left(- \omega_{4},\omega_{1},s \right)} e^{- \beta \omega_{1}} \\
&- \frac{i}{2}\tilde{\mathcal{S}}(- \omega_{3},\omega_{4},s) \tilde{\gamma}{\left(- \omega_{1},\omega_{2},t \right)} e^{- \beta (\omega_{1} +\omega_{2})} - \frac{i}{2} \tilde{\mathcal{S}}(- \omega_{3},\omega_{4},s) \tilde{\gamma}{\left(- \omega_{1},\omega_{2},t \right)} + i \tilde{\mathcal{S}}(- \omega_{3},\omega_{4},s) \tilde{\gamma}(- \omega_{2},\omega_{1},t) e^{- \beta \omega_{1}} \\ 
&- \frac{i}{2} \tilde{\mathcal{S}}(- \omega_{3},\omega_{4},t) \tilde{\gamma}(- \omega_{1},\omega_{2},s) e^{- \beta (\omega_{1} +\omega_{2})} + \frac{i}{2}\tilde{\mathcal{S}}(- \omega_{3},\omega_{4},t) \tilde{\gamma}(- \omega_{1},\omega_{2},s) e^{- \beta (\omega_{1} +\omega_{2}+\omega_{3}+\omega_{4})} \\
&+ \frac{1}{4}\tilde{\gamma}(- \omega_{1},\omega_{2},s) \tilde{\gamma}(- \omega_{3},\omega_{4},t) e^{- \beta (\omega_{1} +\omega_{2})} + \frac{1}{4}\tilde{\gamma}(- \omega_{1},\omega_{2},s) \tilde{\gamma}(- \omega_{3},\omega_{4},t) e^{- \beta (\omega_{1} +\omega_{2}+\omega_{3}+\omega_{4})} \\ 
&- \frac{1}{2}\tilde{\gamma}(- \omega_{1},\omega_{2},s) \tilde{\gamma}{\left(- \omega_{4},\omega_{3},t \right)} e^{- \beta (\omega_{1} +\omega_{2}+\omega_{3})} + \frac{1}{4}\tilde{\gamma}{\left(- \omega_{1},\omega_{2},t \right)} \tilde{\gamma}(- \omega_{3},\omega_{4},s) e^{- \beta (\omega_{1} +\omega_{2})} \\ 
&+ \frac{1}{4}\tilde{\gamma}{\left(- \omega_{1},\omega_{2},t \right)} \tilde{\gamma}(- \omega_{3},\omega_{4},s) - \frac{1}{2}\tilde{\gamma}(- \omega_{2},\omega_{1},t) \tilde{\gamma}(- \omega_{3},\omega_{4},s) e^{- \beta \omega_{1}} \\
&- \frac{1}{2}\tilde{\gamma}{\left(- \omega_{2},\omega_{3},t \right)} \tilde{\gamma}{\left(- \omega_{4},\omega_{1},s \right)} e^{- \beta (\omega_{1} +\omega_{2}+\omega_{3})} - \frac{1}{2}\tilde{\gamma}{\left(- \omega_{2},\omega_{3},t \right)} \tilde{\gamma}{\left(- \omega_{4},\omega_{1},s \right)} e^{- \beta \omega_{1}} \\
&+ \tilde{\gamma}{\left(- \omega_{3},\omega_{2},t \right)} \tilde{\gamma}{\left(- \omega_{4},\omega_{1},s \right)} e^{- \beta (\omega_{1} +\omega_{2})} 
\end{align}
Consequently, we have
\begin{eqnarray} \label{g40_expression_for_cumulant}
g_{40}(\omega_1,\omega_2,\omega_3,\omega_4,t)= \frac{1}{2} e^{i(\omega_1+\omega_2+\omega_3+\omega_4)t}\int_0^t ds \ (f(\omega_1,\omega_2,\omega_3,\omega_4,t,s)-f(\omega_1,\omega_2,\omega_3,\omega_4,s,t)).
\end{eqnarray}

The above expression, in general would not have well defined limit for $t\to\infty$. 
However we will need the xxx


\subsection{Proof of Proposition \ref{diagonal_elements_summation_proposition}} \label{proof_diagonal_elements_propositon}
We consider diagonal elements of the Eq. \eqref{fourth_order_equation_with_g}, such that we obtain the following set of equations
\begin{equation}
    \sum_{\vec \omega} (g_{04}(\vec \omega) + g_{22}(\vec \omega) + g_{40}(\vec \omega)) e^{-\beta \epsilon_k} \bra k A(\omega_1) A(\omega_2) A(\omega_3) A(\omega_4) \ket k = 0.
\end{equation}
for $k= 0,1,2,\dots$. We see that since for arbitrary $\rho$ we have:
\begin{align}
    &\sum_{\vec \omega} g_{04}(\vec \omega) \bra k e^{-\beta H_0} A(\omega_1) A(\omega_2) A(\omega_3) A(\omega_4) \ket k \\
    &= \bra k \mathcal{L}_\infty^{(0)} [\rho] \ket k = \bra k [H_0,\rho] \ket k =0,
\end{align}
so our condition is now just 
\begin{equation}
\label{eq:4order-steady-state-cond}
    \sum_{\vec \omega} \Big( g_{22}(\vec \omega) + g_{40}(\vec \omega)\Big) e^{-\beta \epsilon_k} \bra k A(\omega_1) A(\omega_2) A(\omega_3) A(\omega_4) \ket k = 0.
\end{equation}
Moreover, one observes that $\bra k A(\omega_1) A(\omega_2) A(\omega_3) A(\omega_4) \ket k$ is nonzero only if $\sum_k \omega_k = 0$. Consequently, let us denote by $G(|k\rangle\to|k\rangle)$ the set of all four-tuples 
$\vec \omega$ of the form: 
\begin{align}
    \vec \omega=(\epsilon_l-\epsilon_k,\epsilon_m-\epsilon_l,\epsilon_j-\epsilon_m,\epsilon_k-\epsilon_j),
\end{align}
form which follows Eq. \eqref{final_g22_g44_eq}.

\subsection{Two-level system}
Now, we shall specialize to the case of a two-level system. We then have $k=0,1$ and $\epsilon_1-\epsilon_0 = \omega_0$, such that 
\begin{align}
    G(\ket 0 \to \ket 0)=
    &\{(0,0,0,0),(\omega_0,-\omega_0,0,0),(\omega_0,0,-\omega_0,0),(\omega_0,0,0,-\omega_0), \\ &(0,\omega_0,-\omega_0,0),(0,\omega_0,0,-\omega_0),(0,0,\omega_0,-\omega_0),(\omega_0,-\omega_0,\omega_0,-\omega_0\})
\end{align}
The set $G(|1\rangle\to|1\rangle)$ is the same but with changed sign of the qubit frequency $\omega_0 \to -\omega_0$. 
Then, according to Eq. \eqref{g_22_appendix}, one can first observe that the coefficient $g_{22}$ summed over first seven four-tuples vanishes, i.e.,
\begin{multline}
    g_{22}(0,0,0,0) + g_{22}(\omega_0,-\omega_0,0,0) + g_{22}(\omega_0,0,-\omega_0,0) + g_{22}(\omega_0,0,0,-\omega_0) \\ + g_{22}(0,\omega_0,-\omega_0,0) + g_{22}(0,\omega_0,0,-\omega_0) + g_{22}(0,0,\omega_0,-\omega_0) = 0
\end{multline}
whereas for the last one we have 
\begin{align}
g_{22}(\omega_0,-\omega_0,\omega_0,-\omega_0) 
= \beta\left (\Upsilon_{\text{st}}(-\omega_0,-\omega_0) K(-\omega_0,-\omega_0) - e^{-\beta \omega_0} \Upsilon_{\text{st}}(\omega_0,\omega_0) K(\omega_0,\omega_0) \right).
\end{align}
If additionally $K(\omega,\omega)$ obeys the detailed balance condition, then 
\begin{eqnarray}
g_{22}(\omega_0,-\omega_0,\omega_0,-\omega_0) = \beta e^{-\beta \omega_0} \left(\Upsilon_{\text{st}}(-\omega_0,-\omega_0) - \Upsilon_{\text{st}}(\omega_0,\omega_0) \right) K(\omega_0,\omega_0).
\end{eqnarray}

\subsubsection{Second-order master equation}
Now, since for arbitrary master equation of the form \eqref{master_equation}, which is up to second order in $\lambda$, we also have $g_{40}=0$. From this we conclude that Eq. \eqref{fourth_order_algebraic} is satisfied if 
\begin{eqnarray} \label{qubit_solution_diagonal}
\Upsilon_{\text{st}}(\omega_0,\omega_0) = \Upsilon_{\text{st}}(-\omega_0,-\omega_0).
\end{eqnarray}
Since for a two-level system, in general we have 
\begin{eqnarray}
\bra 0 H_\text{st}^{(2)} \ket 0 = \Upsilon_{\text{st}}(0,0) + \Upsilon_{\text{st}}(\omega_0,\omega_0),  \quad \bra 1 H_\text{st}^{(2)} \ket 1 = \Upsilon_{\text{st}}(0,0) + \Upsilon_{\text{st}}(-\omega_0,-\omega_0).
\end{eqnarray}
Thus, applying the condition \eqref{qubit_solution_diagonal}, we finally get:
\begin{eqnarray}
\bra 0 H_\text{st}^{(2)} \ket 0 = \bra 1 H_\text{st}^{(2)} \ket 1 = \Upsilon_{\text{st}}(\omega_0,\omega_0).
\end{eqnarray}

\subsubsection{Cumulant equation}
To solve the Eq. \eqref{fourth_order_algebraic} for the cumulant master equation we need to additionally calculate the term involving the coefficient $g_{40}$. Putting the expression \eqref{g40_expression_for_cumulant}, we observe that, similarly to the summation of $g_{22}$, the sum over first seven tuples vanishes, such that we obtain a very simple expression 
\begin{eqnarray}
\lim_{t \to \infty} \sum_{(\omega_1,\omega_2,\omega_2,\omega_4)\in G(\ket 0 \to \ket 0)} g_{40}(\omega_1,\omega_2,\omega_2,\omega_4,t) &=&  g_{40}(\omega_0,-\omega_0,\omega_0,-\omega_0), 
\end{eqnarray}
where 
\begin{eqnarray}
g_{40}(\omega_0,-\omega_0,\omega_0,-\omega_0) = \frac{1}{2} e^{-\beta \omega_0}(1+e^{\beta \omega_0}) \gamma(\omega_0) \int_0^\infty ds \ (e^{-\beta \omega_0}  \gamma(\omega_0,s) - \gamma(-\omega_0,s)).
\end{eqnarray}
Since, the leading order of the cumulant master equation is the Bloch-Redfield generator, we also have 
\begin{eqnarray}
g_{22}(\omega_0,-\omega_0,\omega_0,-\omega_0) = \beta e^{-\beta \omega_0} \left(\Upsilon_{\text{st}}(-\omega_0,-\omega_0) - \Upsilon_{\text{st}}(\omega_0,\omega_0) \right) \gamma(\omega_0).
\end{eqnarray}
Finally, we need to solve 
\begin{eqnarray}
g_{22}(\omega_0,-\omega_0,\omega_0,-\omega_0) + g_{40}(\omega_0,-\omega_0,\omega_0,-\omega_0) = 0 
\end{eqnarray}
which gives us 
\begin{align}
&\Upsilon_{\text{st}}(\omega_0,\omega_0) - \Upsilon_{\text{st}}(-\omega_0,-\omega_0)  = \frac{1}{2\beta} (1+e^{\beta \omega_0}) \int_0^\infty ds \ ( e^{-\beta \omega_0}  \gamma(\omega_0,s) - \gamma(-\omega_0,s)) \\
&= \frac{1}{2\beta} \int_0^\infty ds \ (\gamma(\omega_0,s) + e^{-\beta \omega_0} \gamma(\omega_0,s) - \gamma(-\omega_0,s) - e^{\beta \omega_0} \gamma(-\omega_0,s)) \\
&= \frac{1}{2\beta} \int_0^\infty ds \ (\gamma(\omega_0,s) - e^{\beta \omega_0} \gamma(-\omega_0,s)) - \frac{1}{2\beta} \int_0^\infty ds \ ( \gamma(-\omega_0,s) - e^{-\beta \omega_0} \gamma(\omega_0,s))
\end{align}
The above formula  does not yet allow to determine
$\Upsilon_{\text{st}}(\omega,\omega)$, since it is a difference of such quantities. 
However this indeterminacy is just a shift of the Hamiltonian by a constant, and therefore it is irrelevant. Actually this is just the gauge that has to be chosen at some point. 
We just can consider the simplest choice
\begin{eqnarray}
\Upsilon_{\text{st}}(\omega,\omega) = \frac{1}{2\beta} \int_0^\infty ds \ (\gamma(\omega,s) - e^{\beta \omega} \gamma(-\omega,s)).
\end{eqnarray}

In the end, we want to compare the steady-state correction with the mean-force one. The diagonal part of the mean-force coefficients is given by:
\begin{eqnarray}
\Upsilon_{\text{mf}}(\omega,\omega) = \frac{1}{2\pi} \int_{-\infty}^{+\infty} d \Omega \ D_\mathrm{mf}(\omega,\omega,\Omega) \gamma(\Omega)
\end{eqnarray}
where 
\begin{eqnarray} \label{no_poles_formula}
D_\mathrm{mf}(\omega,\omega,\Omega) &=& - \frac{1}{\beta} \int_0^\beta dt \int_0^{t} ds \ e^{s (\omega - \Omega)} = \frac{1-e^{\beta(\omega-\Omega)}+\beta(\omega-\Omega)}{\beta (\omega-\Omega)^2}.
\end{eqnarray}
Let us then represent a function $\gamma(\omega,t)$ in a similar way. From the definition, we have:
\begin{equation}
\gamma(\omega,t) = \Gamma(\omega,t)+\Gamma(\omega,t)^* = \int_{-t}^t ds \ e^{i \omega s} \langle R(s) R \rangle,
\end{equation}
such that by substituting 
\begin{equation}
    \langle R(s) R \rangle = \frac{1}{2\pi} \int_{-\infty}^\infty d\Omega \ e^{-i \Omega s} \gamma(\Omega)
\end{equation}
we get
\begin{equation}
\gamma(\omega,t) = \int_{-t}^t ds \ e^{i \omega s} \langle R(s) R \rangle =  \frac{1}{2\pi} \int_{-\infty}^\infty d\Omega \ \gamma(\Omega) \int_{-t}^t ds \ e^{i (\omega-\Omega) s}  = \frac{1}{\pi} \int_{-\infty}^\infty d\Omega \ \frac{\gamma(\Omega)}{\omega-\Omega} \sin[(\omega-\Omega) t].
\end{equation}
Accordingly, the steady-state coefficient is given by:
\begin{equation}
    \Upsilon_{\text{st}}(\omega,\omega) = \frac{1}{2 \pi \beta} \int_{-\infty}^\infty d\Omega \ \gamma(\Omega) \int_0^\infty dt \left(\frac{\sin[(\omega-\Omega) t]}{\omega-\Omega} - e^{\beta \omega} \frac{\sin[(\omega+\Omega) t]}{\omega+\Omega} \right).
\end{equation}
Now, let us observe that since the function $\gamma(\Omega)$ satisfies the detailed-balance condition, i.e., $\gamma(\Omega) = e^{\beta \Omega} \gamma(-\Omega)$, then one may write:
\begin{eqnarray}
 \int_{-\infty}^\infty d\Omega \ \gamma(\Omega) \frac{\sin[(\omega+\Omega) t]}{\omega+\Omega} = \int_{-\infty}^\infty d\Omega \ e^{\beta \Omega} \gamma(-\Omega) \frac{\sin[(\omega+\Omega) t]}{\omega+\Omega} = \int_{-\infty}^\infty d\Omega \  \gamma(\Omega) e^{-\beta \Omega} \frac{\sin[(\omega-\Omega) t]}{\omega-\Omega}.
\end{eqnarray}
Applying this to the previous equation, we get
\begin{equation}
    \Upsilon_{\text{st}}(\omega,\omega) = \frac{1}{2 \pi \beta} \int_{-\infty}^\infty d\Omega \ \gamma(\Omega) \frac{1- e^{\beta (\omega-\Omega)}}{\omega-\Omega}  \int_0^\infty dt \sin[(\omega-\Omega) t].
\end{equation}
Finally, we represent the integral over sine as the Cauchy principal value, i.e., by using the Sokhotski-Plemelj formula \eqref{sokhostki_plemelj_identity}, we may write
\begin{eqnarray}
    \int_0^\infty dt \sin[(\omega-\Omega) t] = \mathcal{P} \frac{1}{\omega-\Omega},
\end{eqnarray}
such that
\begin{eqnarray}
    \Upsilon_{\text{st}}(\omega,\omega) = \mathcal{P} \frac{1}{2 \pi \beta} \int_{-\infty}^\infty d\Omega \ \gamma(\Omega) \frac{1- e^{\beta (\omega-\Omega)}}{(\omega-\Omega)^2}.
\end{eqnarray}

Let us then back to the mean-force representation. The formula $D_\mathrm{mf}(\omega,\omega,\Omega)$ given by Eq. \eqref{no_poles_formula} has no poles, nevertheless, it can be split into two principal value integrals:
\begin{eqnarray}
    \Upsilon_{\text{mf}}(\omega,\omega) &=& \frac{1}{2\pi \beta} \int_{-\infty}^{+\infty} d \Omega \  \gamma(\Omega) \frac{1-e^{\beta(\omega-\Omega)}+\beta(\omega-\Omega)}{(\omega-\Omega)^2} \\
    &=& \mathcal{P}\frac{1}{2\pi \beta} \int_{-\infty}^{+\infty} d \Omega \  \gamma(\Omega) \frac{1-e^{\beta(\omega-\Omega)}}{(\omega-\Omega)^2} + \mathcal{P} \frac{1}{2\pi} \int_{-\infty}^{+\infty} d \Omega \  \gamma(\Omega) \frac{1}{\omega-\Omega}.
\end{eqnarray}
The first integral is precisely the representation of the steady-state corrections, whereas the second term is the previously defined function $\mathcal{S}(\omega)$ \eqref{S_principal_value}.

Finally, we have proved the following identity for the cumulant equation (for the two-level system):
\begin{eqnarray}
\label{eq:gamma-s}
    \Upsilon_{\text{mf}}(\omega,\omega) &=& \Upsilon_{\text{st}}(\omega,\omega) + \mathcal{S}(\omega).
\end{eqnarray}

\section{Cumulant equation} \label{cumulant_section}

Consider a system interacting with a thermal reservoir which Hamiltonian is given by:
\begin{equation}
    H=H_0 +H_R + \lambda H_I
\end{equation}
Let us also consider the Born Approximation such that $\rho(0)=\rho_{S}(0) \otimes \rho_R$ where $\rho_R$ is a stationary state of the environment. In the interaction picture the reduced state at time $t$ is:
\begin{equation}
    \rho_{S}(t)= \Tr_{R}\left({U(t,t_{0}) \rho_{S}(t_{0}) \otimes \rho_{R}(t_{0}) U^{\dagger}(t,t_{0})}\right)
\end{equation}
One may expand the evolution operator in the interaction picture $U(t,0)=\mathcal{T} e^{-i \int_{0}^{t} H_{I}(t') dt'}$ and rearrange terms (of the same power of $H_{I}$) to obtain:
\begin{eqnarray}
      \rho_{S}(t)= \rho_{S}(0) \underbrace{-\lambda^{2}\frac{\mathcal{T}}{2} \int_{0}^{t} dt_{1}\int_{0}^{t} dt_{2}\Tr_{R}\left(\left[H_{I}(t_1),\left[H_{I}(t_2),\rho_{S}(0)\otimes\rho_{R}\right]\right]\right) }_{\tilde K^{(2)}_t}
      +\mathcal{O}(\lambda^3)
\end{eqnarray}
The terms $\mathcal{O}(H_I^3)$ can be neglected for weak coupling or short times. We already considered the initial state of the bath to be thermal $\rho_{B}(0)=\rho_{\beta}=e^{-\beta H_B}/\Tr{e^{-\beta H_B}}$ and the bath operators to be centralized. Let us know focus on the second term, let us apply time-ordering explicitly so that:
\begin{eqnarray}
\tilde K^{(2)}_t&=&-\frac{\lambda^2}{2} \int_{0}^{t} dt_{1}\int_{0}^{t} dt_{2} \theta(t_1 -t_2 )\Tr_{R}\left(\left[H_{I}(t_1),\left[H_{I}(t_2),\rho_{S}(0)\otimes\rho_{R}\right]\right]\right)\nonumber \\&-&\frac{\lambda^2}{2} \int_{0}^{t} dt_{1}\int_{0}^{t} dt_{2} \theta(t_2 -t_1 )\Tr_{R}\left(\left[H_{I}(t_2),\left[H_{I}(t_1),\rho_{S}(0)\otimes\rho_{R}\right]\right]\right)
\end{eqnarray}
Let us know expand the double commutators:
\begin{eqnarray}
\tilde K^{(2)}_t&=&-\frac{\lambda^2}{2} \int_{0}^{t} dt_{1}\int_{0}^{t} dt_{2} \theta(t_1 -t_2 )\Tr_{R}\Big[ H_I(t_1) H_I(t_2) \rho_S(0)\rho_R - H_I(t_1) \rho_S(0) \rho_R H_I(t_2) - H_I(t_2) \rho_S(0) \rho_R H_I(t_1) \nonumber  \\ &+&  \rho_S(0) \rho_R H_I(t_1) H_I(t_2)\Big] -\frac{\lambda^2}{2} \int_{0}^{t} dt_{1}\int_{0}^{t} dt_{2} \theta(t_2 -t_1 )\Tr_{R}\Big[ H_I(t_2) H_I(t_1) \rho_S(0)\rho_R - H_I(t_2) \rho_S(0) \rho_R H_I(t_1) \nonumber \\ &-& H_I(t_1) \rho_S(0) \rho_R H_I(t_2) \nonumber  +  \rho_S(0) \rho_R H_I(t_2) H_I(t_1)\Big]
\end{eqnarray}
From here it can be seen that we have three kind of terms, namely $H_I^2 \rho,H_I \rho H_I, \rho H_I^2$. Let us consider each of those independently
\begin{eqnarray}
H_I \rho H_I &:&\frac{\lambda^2}{2} \int_{0}^{t} dt_{1}\int_{0}^{t} dt_{2} \big( \theta(t_1 -t_2 )+\theta(t_2 -t_1 ) \big)\Tr_{R}\Big[ H_I(t_1) \rho_S(0) \rho_R H_I(t_2) + H_I(t_2) \rho_S(0) \rho_R H_I(t_1)\Big]  \\
&=&\frac{\lambda^2}{2} \int_{0}^{t} dt_{1}\int_{0}^{t} dt_{2} \Tr_{R}\Big[ H_I(t_1) \rho_S(0) \rho_R H_I(t_2) + H_I(t_2) \rho_S(0) \rho_R H_I(t_1)\Big] \\
&=&\lambda^2 \int_{0}^{t} dt_{1}\int_{0}^{t} dt_{2} \Tr_{R}\Big[ H_I(t_1) \rho_S(0) \rho_R H_I(t_2)\Big]
\end{eqnarray}
where in the last step we used a change of variables on the second term, such that $t_1 \leftrightarrow{} t_2$. Next, we consider the other two missing terms
\begin{eqnarray}
H_I^{2} \rho &:& -\frac{\lambda^2}{2} \int_{0}^{t} dt_{1}\int_{0}^{t} dt_{2} \Big(\theta(t_1 -t_2 )\Tr_{R}\Big[ H_I(t_1) H_I(t_2) \rho_S(0)\rho_R \Big] + \theta(t_2 -t_1 )\Tr_{R}\Big[ H_I(t_2) H_I(t_1) \rho_S(0)\rho_R \Big]\Big) \\
&=& -\frac{\lambda^2}{2} \int_{0}^{t} dt_{1}\int_{0}^{t} dt_{2} \Big(\theta(t_1 -t_2 )\Tr_{R}\Big[ [H_I(t_1), H_I(t_2)] \rho_S(0)\rho_R \Big] + \Tr_{R}\Big[ H_I(t_2) H_I(t_1) \rho_S(0)\rho_R \Big] \Big) \\
 \rho H_I^{2} &:& -\frac{\lambda^2}{2} \int_{0}^{t} dt_{1}\int_{0}^{t} dt_{2} \Big(\theta(t_1 -t_2 )\Tr_{R}\Big[  \rho_S(0)\rho_R H_I(t_2) H_I(t_1) \Big] + \theta(t_2 -t_1 )\Tr_{R}\Big[ H_I(t_1) H_I(t_2) \rho_S(0)\rho_R \Big]\Big) \\
&=& -\frac{\lambda^2}{2} \int_{0}^{t} dt_{1}\int_{0}^{t} dt_{2} \Big(\theta(t_1 -t_2 )\Tr_{R}\Big[\rho_S(0)\rho_R  [H_I(t_2), H_I(t_1)] \Big] + \Tr_{R}\Big[ \rho_S(0)\rho_R H_I(t_1) H_I(t_2) \Big] \Big)
\end{eqnarray}
In both cases the step taken from one line to the other was summing a zero so that the terms could be recast in that form, they were ${\pm \frac{1}{2} \int_{0}^{t} dt_{1}\int_{0}^{t} dt_{2} \theta(t_1 -t_2 ) \Tr_{R}\Big[  H_I(t_2) H_I(t_1) \rho_S(0)\rho_R \Big]}$ in the first case and ${\pm \frac{1}{2} \int_{0}^{t} dt_{1}\int_{0}^{t} dt_{2} \theta(t_1 -t_2 ) \Tr_{R}\Big[ \rho_S(0)\rho_R  H_I(t_1) H_I(t_2) \Big]}$ in the second one. Regrouping all terms we have
\begin{eqnarray}
\tilde K^{(2)}_t&=&\lambda^2\int_{0}^{t} dt_{1}\int_{0}^{t} dt_{2} \Bigg(\Tr_{R}\Big[ H_I(t_1) \rho_S(0) \rho_R H_I(t_2)\Big]- \frac{1}{2} \big(\Tr_{R}\Big[ \rho_S(0)\rho_R H_I(t_1) H_I(t_2) \Big]+\Tr_{R}\Big[ H_I(t_2) H_I(t_1) \rho_S(0)\rho_R \Big]  \Big) \Bigg) \nonumber \\ &-&\frac{\lambda^2}{2} \int_{0}^{t} dt_{1}\int_{0}^{t} dt_{2} \theta(t_1-t_2)\Big(\Tr_{R}\Big[ [H_I(t_1), H_I(t_2)] \rho_S(0)\rho_R \Big] -\Tr_{R}\Big[\rho_S(0)\rho_R  [H_I(t_1), H_I(t_2)] \Big] \Big) \\ &=& \lambda^2 \int_{0}^{t} dt_{1}\int_{0}^{t} dt_{2} \Bigg(\Tr_{R}\Big[ H_I(t_1) \rho_S(0) \rho_R H_I(t_2)\Big]- \frac{1}{2} \big(\Tr_{R}\Big[ \rho_S(0)\rho_R H_I(t_1) H_I(t_2) \Big]+\Tr_{R}\Big[ H_I(t_2) H_I(t_1) \rho_S(0)\rho_R \Big]  \Big) \Bigg) \nonumber \\ &-& i\lambda^2 [\Lambda(t),\rho_S(0)]
\end{eqnarray}
where:
\begin{eqnarray}
    \Lambda(t)&=&\frac{1}{2 i} \int_{0}^{t} dt_{1}\int_{0}^{t} dt_{2} \theta(t_1-t_2) \Tr_{R}\Big[ [H_I(t_1), H_I(t_2)] \rho_R \Big] \\
    &=& \frac{1}{2 i} \int_{0}^{t} dt_{1}\int_{0}^{t} dt_{2} sgn(t_1-t_2) \Tr_{R}\Big[ H_I(t_1) H_I(t_2)\rho_R \Big]
\end{eqnarray}
where we used  $\theta(x)=\frac{1+sgn(x)}{2}$. Now, if we expand the interaction Hamiltonian in the interaction picture ${H_I=\sum_{w,k} e^{iwt} A_{k}(w) B_{k}=\sum_{w,k} e^{-iwt} A^{\dagger}_{k}(w) B_{k}}$ 
\begin{eqnarray}
\Lambda(t)&=&\frac{1}{2i} \sum_{w,w'} \sum_{\alpha \beta} \int_{0}^{t} dt_{1}\int_{0}^{t} dt_{2} sgn(t_1-t_2) e^{i(w t_1- w' t_2)} A^{\dagger}_{\alpha}(w) A_{\beta}(w') \langle B_{\alpha}(t_1) B_{\beta}(t_2) \rangle_{R} \nonumber \\
&=&  \sum_{w,w'} \sum_{\alpha \beta} \Xi(w,w',t) A^{\dagger}_{\alpha}(w) A_{\beta}(w')
\end{eqnarray}
So we obtain:
\begin{eqnarray}
\tilde K^{(2)}_t [\rho_S(0)]&=& - i  \sum_{w,w'} \sum_{\alpha \beta} \Xi(w,w',t) [A^{\dagger}_{\alpha}(w) A_{\beta}(w') ,\rho_S(0)]+  
\xi_{\alpha \beta}(w,w',t) \Big( A_{\beta}(w') \rho_S(0) A^{\dagger}_{\alpha}(w)  - \frac{1}{2} \{ A^{\dagger}_{\alpha}(w) A_{\beta}(w'),\rho_S(0)\} \Big),
\nonumber 
\end{eqnarray}
where
 \begin{equation}
     \xi_{\alpha \beta}(w,w',t)=\int_{0}^{t} dt_1 \int_{0}^{t} dt_2 e^{i (w t_1 - w' t_2)} \langle R_{\alpha}(t_1) R_{\beta}(t_2) \rangle.
 \end{equation}
We may rewrite this in terms of previously obtained quantities as:
\begin{eqnarray}
  \xi_{\alpha \beta}(w,w',t)&=&\int_{0}^{t} ds \int_{0}^{t} d\omega e^{i (w s - w' \omega)} \langle R_{\alpha}(s) R_{\beta}(w) \rangle  \\
  &=& \int_{0}^{t} ds \int_{s}^{t} d\omega e^{i (w s - w' \omega)} \langle R_{\alpha}(s) R_{\beta}(w) \rangle
  +\int_{0}^{t} ds \int_{0}^{s} d\omega e^{i (w s - w' \omega)} \langle R_{\alpha}(s) R_{\beta}(w) \rangle \\
  &=&\int_{0}^{t} d\omega \int_{0}^{\omega} ds e^{i (w s - w' \omega)} \langle R_{\alpha}(s) R_{\beta}(w) \rangle +\int_{0}^{t} ds \int_{0}^{s} d\omega e^{i (w s - w' \omega)} \langle R_{\alpha}(s) R_{\beta}(w) \rangle \\
  &=&\int_{0}^{t} ds \int_{0}^{s} d\omega e^{i (w \omega - w' s)} \langle R_{\alpha}(w) R_{\beta}(s) \rangle +\int_{0}^{t} ds \int_{0}^{s} d\omega e^{i (w s - w' \omega)} \langle R_{\alpha}(s) R_{\beta}(w) \rangle \\
  &=&\int_{0}^{t} ds \int_{0}^{s} d\omega e^{i (w \omega - w' s)} \langle R_{\alpha}(w-s) R_{\beta} \rangle +\int_{0}^{t} ds \int_{0}^{s} d\omega e^{i (w s - w' \omega)} \langle R_{\alpha}(s-w) R_{\beta} \rangle \\
    &=&\int_{0}^{t} ds \int_{0}^{s} d\xi e^{i ((w - w') s-w\xi)} \langle R_{\alpha}(-\xi) R_{\beta} \rangle +\int_{0}^{t} ds \int_{0}^{s} d\xi e^{i ((w- w') s +\xi w')} \langle R_{\alpha}(\xi) R_{\beta} \rangle\\
     &=& \int_{0}^{t} ds e^{i(w-w')s} (\Gamma_{\beta \alpha}^{*}(w,s)+\Gamma_{\alpha \beta}(w',s))  \\
    &=& \int_{0}^{t} ds e^{i(w-w')s} \gamma_{\alpha \beta}(w,w',s)=\int_{0}^{t} ds \tilde \gamma_{\alpha \beta}(w,w',s)
\end{eqnarray}
Now, we can notice that the derivative of such coefficient corresponds to:
\begin{eqnarray}
 \frac{d}{dt} \xi_{\alpha \beta}(w,w',t)=e^{i(w-w')t} \gamma(w,w',t)= \tilde \gamma(w,w',t)
\end{eqnarray}
Furthermore from \cite{WinczewskiAlicki} we know that:
\begin{eqnarray}
  \frac{d}{dt} \Xi_{\alpha \beta}(w,w',t)=\frac{e^{i(w'-w)t}}{2i}(\Gamma_{\alpha \beta}(w',t)-\Gamma_{\beta \alpha}^{*}(w,t)) = e^{i(w'-w)t} S_{\alpha \beta}(w,w',t)= \tilde S_{\alpha \beta}(w,w',t)
\end{eqnarray}
One may then rewrite $\tilde K_{t}^{(2)}$ as:
\begin{align} 
  \tilde K^{(2)}_t[\rho] = \lambda^2 \int_0^t ds \sum_{\omega,\omega'}\sum_{\alpha \beta} \left(i \tilde{\mathcal{S}}_{\alpha \beta}(\omega,\omega',s) [\rho, A_\alpha^\dag(\omega) A_\beta(\omega')] +  \tilde \gamma_{\alpha \beta}(\omega,\omega',s) \left(A_\beta(\omega') \rho A_\alpha^\dag(\omega) - \frac{1}{2} \{A_\alpha^\dag(\omega) A_\beta(\omega'), \rho \} \right)\right)
\end{align}




 


\section{Comparison with 
 Lamb-shift Hamiltonian and the steady state - qubit case}  \label{qubit_section}

In this section we consider the particular case of a qubit coupled to a bosonic bath given by
\begin{equation}
H= \frac{\omega_0}{2} \sigma_z +\sum_k \Omega_k a^{\dagger}_k a_k + S \sum_{k=1}^{\infty} \lambda_k (a_k +a_{k}^{\dagger}) 
\end{equation}
where we take S to be a general interaction operator in the pauli basis:
\begin{equation}
 S = x \sigma_x + y \sigma_y + z \sigma_z 
\end{equation}
This form of Hamiltonian with $y=0$ has been studied previously in \cite{Guarnieri}, where it was reported that such Hamiltonian have steady-state coherences. In this section, we see that the general framework presented here agrees with that result. Using equation \eqref{mean_force_jump_representation} and this interaction, the second-order correction to the Hamiltonian takes the form:
\begin{eqnarray}
 H_\mathrm{cor}^{(2)} = 
\begin{bmatrix}
z^2 \Upsilon_\mathrm{cor}(0,0) + (x^2 + y^2) \Upsilon_\mathrm{cor}(\omega, \omega) & (x - i y) z (\Upsilon_\mathrm{cor}(0, -\omega) - \Upsilon_\mathrm{cor}(\omega,0)) \\
(x + i y) z (\Upsilon_\mathrm{cor}(0, -\omega) - \Upsilon_\mathrm{cor}(\omega,0)) & z^2 \Upsilon_\mathrm{cor}(0,0) + (x^2 + y^2) \Upsilon_\mathrm{cor}(-\omega, -\omega)
\end{bmatrix} 	    
\end{eqnarray}
where $\mathrm{cor}$ indicates the Lamb-shift ($LS$), steady-state ($st$) or mean-force ($mf$) correction. We can rewrite this  correction as a linear combination of the Pauli Matrices such that:
\begin{eqnarray}
 H_\mathrm{cor}^{(2)} &=&  A \mathbb{1} + B \sigma_x + C \sigma_y + D \sigma_z \\
 A &=& z^{2}\Upsilon_\mathrm{cor}(0,0)+\frac{x^{2}+y^{2}}{2}(\Upsilon_\mathrm{cor} (\omega,\omega)+ \Upsilon_\mathrm{cor} (-\omega,-\omega))\\
 B &=& xz(\Upsilon_\mathrm{cor}(0,-\omega) - \Upsilon_\mathrm{cor}(\omega,0))\\
 C &=& yz (\Upsilon_\mathrm{cor}(0,-\omega) -\Upsilon_\mathrm{cor}(\omega,0)) \\
 D &=& \frac{x^{2}+y^{2}}{_\mathrm{cor}2} (\Upsilon_\mathrm{cor}(\omega,\omega)-\Upsilon_\mathrm{cor}(-\omega,-\omega))
\end{eqnarray}

We can see how the different approaches differ qualitatively by looking at the structure of the different $\Upsilon_\mathrm{cor}(\omega,\omega')$ given by each approach. It is important to remark that any approach that performs the secular approximation will have both $B$ and $C$ equal to zero, meaning the correction will be diagonal and as such won't be able to describe the off-diagonal elements of the steady states accordingly. While nonsecular approaches such as the Bloch-Redfield equation, will have non-diagonal corrections, leading to a more appropriate description of the off-diagonal elements of the correction as well as steady state coherences. Let us for a moment recall the structure of the Bloch-Redfield coefficients which are given by \eqref{upsilon_dyn}, simply substituting the appropriate frequencies for the qubit leads to:
 \begin{eqnarray}
  \Upsilon_\mathrm{LS}(0,-\omega)-\Upsilon_\mathrm{LS}(\omega,0)&=&\frac{\mathcal{S}(-\omega)-\mathcal{S}(\omega)}{2}+i \frac{\gamma(0)-(\gamma(\omega)+\gamma(-\omega))}{4} 
 \end{eqnarray}
and
 \begin{eqnarray}
    \Upsilon_\mathrm{LS}(\omega,\omega)+\Upsilon_\mathrm{LS}(-\omega,-\omega)&=& \mathcal{S}(\omega)+\mathcal{S}(-\omega) \\
    \Upsilon_\mathrm{LS}(\omega,\omega)-\Upsilon_\mathrm{LS}(-\omega,-\omega)&=& \mathcal{S}(\omega)-\mathcal{S}(-\omega)
 \end{eqnarray}
Let us now compare this coefficient with the one obtained with the mean force approach. We will only be considering the off-diagonal of the correction:
\begin{eqnarray}
K(\omega) =  \Upsilon_\mathrm{mf}(0, -\omega) - \Upsilon_\mathrm{mf}(\omega,0) = \frac{1}{2 \pi} \int_{-\infty}^{+\infty}  d \Omega \gamma(\Omega) C(\omega, \Omega),
\end{eqnarray}
where
\begin{eqnarray}
  C(\omega, \Omega) =  \left( \frac{\omega^2 (1-e^{-\beta \Omega}) \coth[\frac{\beta \omega}{2}]+ (1+e^{-\beta \Omega}) \Omega \omega}{\Omega (\Omega^2 - \omega^2)} \right). 
\end{eqnarray}
Additionally, our coefficients satisfy detailed balance conditions such that: 
\begin{eqnarray}
 \gamma(-\Omega) = \gamma(\Omega) e^{-\beta \Omega}, \ \ C(\omega, -\Omega) = C(\omega, \Omega) e^{\beta \Omega} 
\end{eqnarray}
Using those we see that $\gamma(-\Omega) C(\omega, -\Omega) = \gamma(\Omega) C(\omega, \Omega)$ and
\begin{eqnarray}
 K(\omega) = \frac{1}{\pi} \int_{0}^{+\infty}  d \Omega \ \gamma(\Omega) C(\omega, \Omega).
\end{eqnarray}
Let us now separate $\gamma(\Omega)$ into its symmetric and anti-symmetric parts:
\begin{eqnarray}
 & \gamma_s(\Omega) = \frac{1}{2} (\gamma(\Omega)+\gamma(-\Omega)) = \frac{1}{2} (1+e^{-\beta \Omega}) \gamma(\Omega) \\
 & \gamma_a(\Omega) = \frac{1}{2} (\gamma(\Omega) - \gamma(-\Omega)) = \frac{1}{2} (1 - e^{-\beta \Omega}) \gamma(\Omega)
\end{eqnarray}
Then we may write:
\begin{eqnarray}
 K(\omega) = \frac{2}{\pi} \int_{0}^{+\infty}  d \Omega \  \left( \frac{\omega^2 \gamma_a(\Omega)  \coth[\frac{\beta \omega}{2}] + \gamma_s(\Omega)  \Omega \omega}{\Omega (\Omega^2 - \omega^2)} \right).
\end{eqnarray}
As mentioned before the system with $y=0$ had been previously considered in \cite{Guarnieri}. Let us now compare our results to those previously available in the literature. Their effective Hamiltonian is given by:
\begin{eqnarray}
 H_S = 
\begin{bmatrix}
\lambda^2 f_1^2 \Upsilon_{st}(0,0) - \frac{1}{2} (\omega  - 2 \lambda^2 f_2^2 \Upsilon_{st}(\omega, \omega)) & \lambda^2 f_1 f_2 K(\omega) \\
\lambda^2 f_1 f_2 K(\omega) &  \lambda^2 f_1^2 \Upsilon_{st}(0,0) + \frac{1}{2} (\omega  - 2 \lambda^2 f_2^2 \Upsilon_{st}(\omega, \omega))
\end{bmatrix} 	    
\end{eqnarray}
The couplings  in this notation are $x = f_2, z=f_1$ and $y=0$. It is also putted $\omega' = \omega  - 2 \lambda^2 f_2^2 \Upsilon(\omega, \omega)$ and $E_0 = \lambda^2 f_1^2 \Upsilon(0,0)$ such that:

\begin{eqnarray}
 H_S = 
\begin{bmatrix}
E_0 - \frac{\omega'}{2} & \lambda^2 f_1 f_2 K(\omega) \\
\lambda^2 f_1 f_2 K(\omega) & E_0 + \frac{\omega'}{2}
\end{bmatrix} 	    
\end{eqnarray}
Then we may find that
\begin{eqnarray}
 \langle \sigma_x \rangle = \frac{\Tr[\sigma_x e^{-\beta H_S}]}{\Tr[e^{-\beta H_S}]} = - x \frac{\tanh[\sqrt{x^2 + z^2} \beta ]}{\sqrt{x^2 + z^2}}
\end{eqnarray}
where $x =  \lambda^2 f_1 f_2 K(\omega) $ and $z = \frac{\omega'}{2}$, and we expand it up to the second order of $\lambda$, i.e., 
\begin{eqnarray}
 \langle \sigma_x \rangle = -\frac{2 x}{\omega'} \tanh[\frac{\beta \omega}{2}] + O(\lambda^3).
\end{eqnarray}
In \cite{Guarnieri} the authors also put $\omega' = \omega$, such that
\begin{eqnarray}
 \langle \sigma_x \rangle &=& - \frac{4 \lambda^2 f_1 f_2}{\pi \omega} \int_{0}^{+\infty}  d \Omega \  \left( \frac{\gamma_s(\Omega) \omega \tanh[\frac{\beta \omega}{2}]}{ \Omega^2 - \omega^2} + \frac{\omega^2 \gamma_a(\Omega)}{\Omega (\Omega^2 - \omega^2)} \right) \\
\end{eqnarray}

Now let us compare $\gamma_{a,s}(\Omega)$ with the correlation function for a bosonic bath:
\begin{eqnarray}
f(t) = \int_0^\infty d \Omega \ J(\Omega) \left(\coth[\frac{\beta \Omega}{2}] \cos(\Omega t) - i \sin (\Omega t) \right)
\end{eqnarray}
\begin{eqnarray}
f(t) &=& \frac{1}{2} \int_{-\infty}^{+\infty} d\Omega \ e^{- i \Omega t} \gamma(\Omega) = \frac{1}{2} \int_{-\infty}^{+\infty} d\Omega \ \gamma(\Omega) \left(\cos(\Omega t) - i \sin (\Omega t) \right) \\
&=& \frac{1}{2} \int_{0}^{+\infty} d\Omega \ (\gamma_s(\Omega)+\gamma_a(\Omega)) \left(\cos(\Omega t) - i \sin (\Omega t) \right) +  \frac{1}{2} \int_{0}^{+\infty} d\Omega \ (\gamma_s(\Omega)-\gamma_a(\Omega)) \left(\cos(\Omega t) + i \sin (\Omega t) \right) \\
&=&  \frac{1}{\pi} \int_{0}^{+\infty} d\Omega \ \left(\gamma_s(\Omega) \cos(\Omega t) - i \gamma_a(\Omega) \sin (\Omega t) \right) =   \frac{1}{\pi} \int_{0}^{+\infty} d\Omega \ \gamma_a(\Omega) \left(\frac{\gamma_s(\Omega)}{\gamma_a(\Omega)} \cos(\Omega t) - i \sin (\Omega t) \right) \\
&=&  \frac{1}{\pi} \int_{0}^{+\infty} d\Omega \ \gamma_a(\Omega) \left(\frac{1+e^{-\beta \Omega}}{1-e^{-\beta \Omega}} \cos(\Omega t) - i \sin (\Omega t) \right) =   \frac{1}{\pi} \int_{0}^{+\infty} d\Omega  \gamma_a(\Omega) \left(\coth[\frac{\beta \Omega}{2}] \cos(\Omega t) - i \sin (\Omega t) \right)
\end{eqnarray}
According to this, we have the following relations:
\begin{eqnarray}
\gamma_a(\Omega) =  \pi J(\Omega) =  \pi \omega_a(\Omega), \ \ \gamma_s(\Omega) = \pi J(\Omega) \coth[\frac{\beta \Omega}{2}] = \pi \omega_s(\Omega)
\end{eqnarray}
and the final result becomes:
\begin{eqnarray} \label{our_formula}
  \langle \sigma_x \rangle &=&   - \frac{4 \lambda^2 f_1 f_2}{\omega} \int_{0}^{+\infty}  d \Omega \  \left( \frac{\omega_s(\Omega) \omega \tanh[\frac{\beta \omega}{2}]}{ \Omega^2 - \omega^2} - \frac{\omega^2 \omega_a(\Omega)}{\Omega (\Omega^2 - \omega^2)} \right).
\end{eqnarray}
On the other hand, in \cite{Guarnieri} we have
\begin{eqnarray} \label{gaurnieri_formula}
\langle \sigma_x \rangle = \frac{2 \lambda^2 f_1 f_2}{\omega} [\Delta_s (\omega) \tanh[\frac{\beta \omega}{2}] + \Delta_a (\omega) - \Delta_a(0)]
\end{eqnarray}
where
\begin{eqnarray}
&\Delta_s (\omega) = \int_0^\infty d \Omega \ \omega_s(\Omega) \left(\frac{1}{\Omega+\omega} - \frac{1}{\Omega-\omega} \right) = - 2 \int_0^\infty d\Omega  \frac{\omega_s(\Omega) \omega}{\Omega^2 - \omega^2}, \\
&\Delta_a (\omega) = \int_0^\infty d \Omega \ \omega_a(\Omega) \left(\frac{1}{\Omega+\omega} + \frac{1}{\Omega-\omega} \right) = 2 \int_0^\infty d\Omega  \frac{\omega_a(\Omega) \Omega}{\Omega^2 - \omega^2}, \\
& \Delta_a(\omega) - \Delta_a(0) = 2 \int_0^\infty d \Omega \ \omega_a(\Omega) \frac{\Omega^2 - (\Omega^2 - \omega^2)}{\Omega (\Omega^2 - \omega^2)} = 2 \int_0^\infty d \Omega  \frac{\omega_a(\Omega) \omega^2}{\Omega (\Omega^2 - \omega^2)},
\end{eqnarray}
such that \eqref{gaurnieri_formula} is equal to \eqref{our_formula}.

\section{Bloch-Redfield master equation (derivation)}
We shall derive the Bloch-Redfield master equation in terms of $\tilde \gamma$ \eqref{gamma_interaction} and $\tilde{\mathcal{S}}$ \eqref{S_interaction} coefficients starting from the von Neumann equation:
\begin{eqnarray}
\tilde{\mathcal{L}}^R_t[\tilde \rho(t)] &=& - \int_0^{t} ds \ \Tr_R [H_I (t), [H_I(s), \tilde \rho(t) \otimes \gamma_R]],
\end{eqnarray}
which is derived according to the Born-Markov approximation. We expand commutators and put an explicit form of the interaction Hamiltonian \eqref{interaction_hamiltonian}:
\begin{eqnarray*}
  \tilde{\mathcal{L}}^R_t[\tilde \rho(t)] &=& - \int_0^{t} ds \ \Tr_R [H_I (t), [H_I(s), \tilde \rho(t) \otimes \gamma_R]] = \int_0^{t} ds \Tr_R \left([H_I(s)\tilde \rho(t) \otimes \gamma_R, H_I (t)] - [\tilde \rho(t) \otimes \gamma_R \ H_I(s), H_I (t)]\right) \\
  &=& \int_0^{t} ds \Tr_R \big[ H_I(s)\tilde \rho(t) \otimes \gamma_R  H_I (t) - \tilde \rho(t) \otimes \gamma_R H_I(s) H_I (t) - H_I (t) H_I(s)\tilde \rho(t) \otimes \gamma_R  +   H_I (t) \tilde \rho(t) \otimes \gamma_R H_I(s) \big] \\
  &=& \sum_{\alpha, \beta} \int_0^{t} ds \big[ A_\alpha (s)\tilde \rho(t) A_\beta(t) \langle R_\beta (t) R_\alpha (s) \rangle_{\gamma_R} +   A_\beta (t) \tilde \rho(t) A_\alpha(s) \langle R_\alpha (s) R_\beta (t) \rangle_{\gamma_R}  \\ 
  && \ \ \ \ \ \ \ \ \ - \tilde \rho(t) A_\alpha (s) A_\beta (t) \langle R_\alpha (s) R_\beta (t) \rangle_{\gamma_R} - A_\beta (t) A_\alpha (s)\tilde \rho(t) \langle R_\beta (t) R_\alpha (s) \rangle_{\gamma_R} \big]  \\
  &=& \sum_{\alpha, \beta} \int_0^{t} ds \left[\langle R_\alpha (t) R_\beta (s) \rangle_{\gamma_R} \left(A_\beta (s)\tilde \rho(t) A_\alpha(t) - A_\alpha (t) A_\beta (s)\tilde \rho(t)\right) \right] + \text{h.c.} \\
 \end{eqnarray*}
After introducing the jump operators \eqref{jump_operators_definition_appendix}, we get 
 \begin{eqnarray}
 \tilde{\mathcal{L}}^R_t[\tilde \rho(t)] &=& \sum_{\omega,\omega'} \sum_{\alpha, \beta} \int_0^{t} ds \ e^{- i (\omega' s + \omega t)} \langle R_\alpha (t) R_\beta (s) \rangle_{\gamma_R} \left(A_\beta (\omega')\tilde \rho(t) A_\alpha(\omega) - A_\alpha (\omega) A_\beta (\omega')\tilde \rho(t)\right)  + \text{h.c.} \\
 &=& \sum_{\omega,\omega'} \sum_{\alpha, \beta} \int_0^{t} ds \ e^{i (\omega t - \omega' s)} \langle R_\alpha (t) R_\beta (s) \rangle_{\gamma_R} \left(A_\beta (\omega')\tilde \rho(t) A^\dag_\alpha(\omega) - A^\dag_\alpha (\omega) A_\beta (\omega')\tilde \rho(t)\right)  + \text{h.c.} \\
 &=& \sum_{\omega,\omega'} \sum_{\alpha, \beta} \ \tilde \Gamma_{\alpha \beta}(\omega, \omega', t) \left(A_\beta (\omega')\tilde \rho(t) A^\dag_\alpha(\omega) - A^\dag_\alpha (\omega) A_\beta (\omega')\tilde \rho(t)\right) + \text{h.c.}
 \end{eqnarray}
where we put the definition:
   \begin{eqnarray}
   \tilde \Gamma_{\alpha \beta}(\omega, \omega', t) \equiv \int_0^{t} ds \ e^{i(\omega t - \omega' s)} \langle R_\alpha (t) R_\beta (s) \rangle_{\gamma_R}.
 \end{eqnarray}
This can be further simplified to the form:
\begin{eqnarray} \label{gamma_interaction_picture}
 \tilde \Gamma_{\alpha \beta}(\omega, \omega', t) = e^{i(\omega - \omega') t} \int_0^{t} ds \ e^{i \omega' s } \langle R_\alpha (s) R_\beta (0) \rangle_{\gamma_R} \equiv e^{i(\omega - \omega') t} \ \Gamma_{\alpha \beta}(\omega', t)
\end{eqnarray}
where we changed the variables in the integrand $s \to t-s$ and use the property $\langle R_\alpha (t) R_\beta (s) \rangle_{\gamma_R} = \langle R_\alpha (t-s) R_\beta (0) \rangle_{\gamma_R}$. 
Next, we rewritten the hermitian conjugate part in the form: 
 \begin{align}
    &\sum_{\omega,\omega'} \sum_{\alpha, \beta} \tilde \Gamma^*_{\alpha \beta}(\omega, \omega', t) \left(A_\beta (\omega')\tilde \rho(t) A^\dag_\alpha(\omega) - A^\dag_\alpha (\omega) A_\beta (\omega')\tilde \rho(t)\right)^\dag \\ 
    &= \sum_{\omega,\omega'} \sum_{\alpha, \beta} \tilde \Gamma^*_{\alpha \beta}(\omega, \omega', t) \left(A_\alpha (\omega)\tilde \rho(t) A^\dag_\beta(\omega') - \tilde \rho(t) A^\dag_\beta (\omega') A_\alpha (\omega) \right) \\
    &= \sum_{\omega,\omega'} \sum_{\alpha, \beta} \tilde \Gamma^*_{\beta \alpha}(\omega', \omega, t) \left(A_\beta (\omega')\tilde \rho(t) A^\dag_\alpha(\omega) - \tilde \rho(t) A^\dag_\alpha (\omega) A_\beta (\omega') \right)
 \end{align}
Finally, we get
 \begin{eqnarray}
 \tilde{\mathcal{L}}^R_t[\tilde \rho(t)] &=& \sum_{\omega,\omega'} \sum_{\alpha, \beta} \ (\tilde \Gamma_{\alpha \beta}(\omega, \omega', t) + \tilde \Gamma^*_{\beta \alpha}(\omega', \omega, t))  A_\beta (\omega')\tilde \rho(t) A^\dag_\alpha(\omega) \\
 &-& \frac{1}{2} \sum_{\omega,\omega'} \sum_{\alpha, \beta} \left( \tilde \Gamma_{\alpha \beta}(\omega, \omega', t) A^\dag_\alpha (\omega) A_\beta (\omega')\tilde \rho(t) + \tilde \Gamma^*_{\beta \alpha}(\omega', \omega, t) \tilde \rho(t) A^\dag_\alpha (\omega) A_\beta (\omega') \right) \\
&-& \frac{1}{2} \sum_{\omega,\omega'} \sum_{\alpha, \beta} \left( \tilde \Gamma_{\alpha \beta}(\omega, \omega', t) A^\dag_\alpha (\omega) A_\beta (\omega')\tilde \rho(t) + \tilde \Gamma^*_{\beta \alpha}(\omega', \omega, t) \tilde \rho(t) A^\dag_\alpha (\omega) A_\beta (\omega') \right) \\
 &-& \frac{1}{2} \sum_{\omega,\omega'} \sum_{\alpha, \beta} \left( \tilde \Gamma^*_{\beta \alpha}(\omega', \omega, t)) A^\dag_\alpha (\omega) A_\beta (\omega')\tilde \rho(t) + \tilde \Gamma_{\alpha \beta}(\omega, \omega', t) \tilde \rho(t) A^\dag_\alpha (\omega) A_\beta (\omega') \right) \\
  &+& \frac{1}{2} \sum_{\omega,\omega'} \sum_{\alpha, \beta} \left( \tilde \Gamma^*_{\beta \alpha}(\omega', \omega, t)) A^\dag_\alpha (\omega) A_\beta (\omega')\tilde \rho(t) + \tilde \Gamma_{\alpha \beta}(\omega, \omega', t) \tilde \rho(t) A^\dag_\alpha (\omega) A_\beta (\omega') \right) 
 \end{eqnarray}
 where the last two lines sum up to zero. 
After rearranging terms and putting the definition \eqref{gamma_interaction} and \eqref{S_interaction}, we finally obtain the master equation in the form:
 \begin{align}
 \tilde{\mathcal{L}}^R_t[\tilde \rho(t)] &=& \sum_{\omega,\omega'} \sum_{\alpha, \beta} \left[ i \tilde{\mathcal{S}}_{\alpha \beta} (\omega,\omega',t) [\tilde \rho(t), A^\dag_\alpha (\omega) A_\beta (\omega')] + \tilde \gamma_{\alpha \beta}(\omega, \omega', t) \left( A_\beta (\omega')\tilde \rho(t) A^\dag_\alpha(\omega) -\frac{1}{2} \{A^\dag_\alpha (\omega) A_\beta (\omega'), \tilde \rho(t) \} \right) \right].
 \end{align}

\section{The cumulant equation in the Schr{\"o}dinger pircture}

The cumulant equation is originally derived in the interaction picture. In order to transform the cumulant equation into the Schr{\"o}dinger picture we start with a simple observation. 
\begin{align} \label{eqn:PreDiffCum}
    \tilde{\rho}(t) = e^{i[H_0,\cdot]t}{\rho}(t).
\end{align}
The super-operator in the r.h.s. of the equation above has its unique inverse, and $\tilde{\rho}(0)={\rho}(0)$, therefore:
\begin{align} \label{cumulant_equation}
   \rho(t) =  e^{-i[H_0,\cdot]t} e^{\tilde K^{(2)}_t} \rho(0)= e^{ K^{(2)}_t} \rho(0).
\end{align}
The r.h.s. of the equation above defines the Schr{\"o}dinger picture cumulant eqaution super-operator $K^{(2)}_t$:
\begin{align}
    e^{ K^{(2)}_t} = e^{-i[H_0,\cdot]t} e^{\tilde K^{(2)}_t}.
    \label{eqn:SchCumMap}
\end{align}

The explicit form of $K^{(2)}_t$ can be found with the aid of the Baker–Campbell–Hausdorff (BCH) formula. 
\begin{align}
    e^X e^Y = \exp{X+Y+\frac{1}{2} [X,Y]+\frac{1}{12}[X,[X,Y]]-\frac{1}{12}[Y,[X,Y]]+\cdots}.\label{eqn:BCHformula}
\end{align}
We observe that in a generic case the super-operator $K^{(2)}_t$ in not of the GKSL form. This follows from the presence of multi-commutator terms in the formula \eqref{eqn:BCHformula}. These terms do not vanish, as $[K^{(2)}_t,H_0]$ is not central. Therefore, $e^{ K^{(2)}_t}$ is an example of a one-parameter family of CPTP dynamical maps that are not of the GKSL form.

\section{The cumulant equation in the differential form}

We start this Section with the following Lemma on the properties of the derivative of an exponential map.
\begin{lemma} The derivative of the exponential map is given by
\begin{align}
    \frac{d}{dt} e^{X(t)}= \left(\frac{e^{[X(t),\cdot]}-\mathbb{1}}{[X(t),\cdot]}\frac{d X(t)}{dt}\right)e^{X(t)}.
\end{align}
\label{lem:DofExp}
\end{lemma}
\begin{proof}
The proof of the above relation is identical to the proof of Theorem 5 in reference \cite{Wulf_2008} up to small modifications. 
\end{proof}

Using Lemma \ref{lem:DofExp} we instantly obtain the cumulant equation in the differential form:
\begin{align}
    \frac{d}{dt} \tilde{\rho}(t)= \left(\frac{e^{[\tilde{K}^{(2)}_t,\cdot]}-\mathbb{1}}{[\tilde{K}^{(2)}_t,\cdot]}\frac{d \tilde{K}^{(2)}_t}{dt}\right)\tilde{\rho}(t)=\tilde{\mathcal{L}}_t^C \tilde{\rho}(t).
    \label{eqn:IntCumDiff}
\end{align}
This result can also be obtained with integration of equation \eqref{eqn:PreDiffCum}. When truncated to the leading order, the above formula reproduces the Bloch-Redfield master equation since $\frac{d\tilde{K}^{(2)}_t}{dt}= \tilde{\mathcal{L}}^{R}_t$. 

Equation \eqref{eqn:IntCumDiff} can be readily transformed into the Schr{\"o}dinger picture. This is done with iterative application of the $e^{\pm i H_0 t}$ operators to the jump operators $A_i(\omega)$ inside $\tilde{K}^{(2)}_t$ super-operator.
\begin{align}
    \frac{d}{dt}  {\rho}(t)&= \left(-i[H_0,\cdot]+e^{-i[H_0,\cdot]t}\frac{e^{[\tilde{K}^{(2)}_t,\cdot]}-\mathbb{1}}{[\tilde{K}^{(2)}_t,\cdot]}\tilde{\mathcal{L}}^{R}_t e^{i[H_0,\cdot]t}\right) \rho(t) \\
    &=\left(-i[H_0,\cdot]+\frac{e^{[\bar{K}^{(2)}_t,\cdot]}-\mathbb{1}}{[\bar{K}^{(2)}_t,\cdot]}\bar{\mathcal{L}}^{R}_t \right)\rho(t)=\mathcal{L}_t^C {\rho}(t),
    \label{eqn:SchCumDiff1}
\end{align}
where 
\begin{align}
  &\bar K^{(2)}_t[\rho] \nonumber\\
  &= \lambda^2 \int_0^t ds \sum_{\omega,\omega'}\sum_{\alpha \beta}e^{i(\omega-\omega')(s-t)} \left(i {\mathcal{S}}_{\alpha \beta}(\omega,\omega',s) [\rho, A_\alpha^\dag(\omega) A_\beta(\omega')] +   \gamma_{\alpha \beta}(\omega,\omega',s) \left(A_\beta(\omega') \rho A_\alpha^\dag(\omega) - \frac{1}{2} \{A_\alpha^\dag(\omega) A_\beta(\omega'), \rho \} \right)\right),\\
  &\bar{\mathcal{L}}^{R}_t={\mathcal{L}}^{R}_t+i[H_0,\cdot].
\end{align}
Moreover, we observe that:
\begin{align}
    \mathcal{L}_t^C = \mathcal{L}^{R}_t + \mathcal{O}(\lambda^4).
\end{align}

Unfortunately, the problem of the long-time limit of the above super-operator was not resolved yet. This situation makes determination of the higher-order corrections to the steady state of the cumulant equation even more involving. 

Equation \eqref{eqn:SchCumDiff1} can be compared with the differential form of the Schr{\"o}dinger picture cumulant equation obtained with Lemma \ref{lem:DofExp} and the super-operator in equation \eqref{eqn:SchCumMap}.
\begin{align}
    \frac{d}{dt} {\rho}(t)= \left(\frac{e^{[{K}^{(2)}_t,\cdot]}-\mathbb{1}}{[{K}^{(2)}_t,\cdot]}\frac{d {K}^{(2)}_t}{dt}\right){\rho}(t)={\mathcal{L}}_t^C {\rho}(t).
    \label{eqn:SchCumDiff2}
\end{align}
The above formula has only a formal meaning, as the ${K}^{(2)}_t$ super-operator does not possess a closed form formula. We present it only for the curiosity of the reader.

\section{Extracting the correction from a density matrix}

To extract the second order correction from the reaction coordinate, we started by obtaining the steady state density matrix, which is given by a Gibbs state of the form:
\begin{equation}
    \rho_{\lambda} = \frac{e^{-\beta H}}{\mathcal{Z}}.
\end{equation} 
By taking the logarithm, one obtains 
\begin{equation}\label{eq:log}
    \log(\mathcal{Z}) + \log(\rho_{\lambda}) =-\beta H.
\end{equation}
We then expand $H$ 
\begin{align}
    & H = H_0 + \lambda^2 H_2 + \lambda^4 H_4+ \dots 
\end{align}
substituting in Eq. \eqref{eq:log}
\begin{align}\label{eq:log2}
\log(\mathcal{Z}) + \log(\rho_{\lambda}) =-\beta (H_0 + \lambda^2 H_2 + \lambda^4 H_4+\dots) 
\end{align}
We now impose our gauge $\Tr[H]=0$. 
Then tracing out both sides we obtain 
\begin{align}
    \log(\mathcal{Z}) = -\frac{1}{d} \Tr[\log(\rho_{\lambda})]
\end{align}
By substituting back into Eq.\eqref{eq:log2} and rearranging terms one obtains
\begin{equation}
    H_{2} = \frac{1}{\lambda^{2}} \Big[ \frac{1}{\beta} \Big( \frac{1}{d} \Tr[\log(\rho_{\lambda})] -\log(\rho_{\lambda})\Big) - \Big( H_{0}+ \mathcal{O}(\lambda^{4}) \Big)\Big]
\end{equation}
Finally as $\lambda$ approaches zero
\begin{equation}
    \lim_{\lambda \to 0} H_{2} = \lim_{\lambda \to 0} \frac{1}{\lambda^{2}} \Big[ \frac{1}{\beta} \Big( \frac{1}{d} \Tr[\log(\rho_{\lambda})] -\log(\rho_{\lambda})\Big) -  H_{0} \Big]
\end{equation}

\end{document}